\documentclass[11pt,british]{article}
\usepackage{libertineRoman}
\usepackage{biolinum}
\usepackage{libertineMono}
\usepackage[T1]{fontenc}
\usepackage[utf8]{inputenc}
\usepackage{color}
\usepackage{verbatim}
\usepackage{mathtools}
\usepackage{amsmath}
\usepackage{amsthm}
\usepackage{amssymb}
\usepackage{cancel}
\usepackage{graphicx}
\usepackage{geometry}
\usepackage{xargs}[2008/03/08]

\makeatletter
\usepackage{hyperref}
\usepackage{titling}

\hypersetup{
    colorlinks=true,
    linkcolor=blue,
    filecolor=magenta,      
    urlcolor=cyan,
	citecolor=blue,
    }

\hypersetup{hidelinks,
backref=true,
pagebackref=true,
hyperindex=true,
breaklinks=true,
colorlinks=true,%
urlcolor=blue,
bookmarks=true,
bookmarksopen=false}

\usepackage{color}
\definecolor{purple}{RGB}{0,0,0} %
\newcommand\branchcolor[2]{{\color{#1} #2}}

\usepackage{amssymb}
\usepackage{amsfonts}

\usepackage{tikzsymbols}
\usepackage[bottom]{footmisc}
\usepackage{color,graphicx}
\usepackage[ampersand]{easylist}

\makeatother

\providecommand\questionname{Question}
\theoremstyle{plain}
\newtheorem{question}{\protect\questionname}

\providecommand\theoremname{Theorem}
\theoremstyle{plain}
\newtheorem{thm}{\protect\theoremname}

\providecommand\notationname{Notation}
\theoremstyle{remark}
\newtheorem{notation}[thm]{\protect\notationname}
\providecommand\definitionname{Definition}
\theoremstyle{definition}
\newtheorem{defn}[thm]{\protect\definitionname}
\providecommand\remarkname{Remark}
\theoremstyle{remark}
\newtheorem{rem}[thm]{\protect\remarkname}
\providecommand\factname{Fact}
\theoremstyle{plain}
\newtheorem{fact}[thm]{\protect\factname}
\providecommand\lemmaname{Lemma}
\newtheorem{lem}[thm]{\protect\lemmaname}
\providecommand\corollaryname{Corollary}
\newtheorem{cor}[thm]{\protect\corollaryname}
\providecommand\examplename{Example}
\theoremstyle{definition}
\newtheorem{example}[thm]{\protect\examplename}
\providecommand\claimname{Claim}
\theoremstyle{remark}
\newtheorem{claim}[thm]{\protect\claimname}
\providecommand\problemname{Problem}
\theoremstyle{definition}
\newtheorem{problem}[thm]{\protect\problemname}
\theoremstyle{remark}
\newtheorem*{claim*}{\protect\claimname}
\usepackage{babel}
\usepackage{cleveref}
\AddToHook{env/notation/begin}{\crefalias{thm}{notation}}
\AddToHook{env/defn/begin}{\crefalias{thm}{defn}}
\AddToHook{env/rem/begin}{\crefalias{thm}{rem}}
\AddToHook{env/fact/begin}{\crefalias{thm}{fact}}
\AddToHook{env/lem/begin}{\crefalias{thm}{lem}}
\AddToHook{env/cor/begin}{\crefalias{thm}{cor}}
\AddToHook{env/example/begin}{\crefalias{thm}{example}}
\AddToHook{env/claim/begin}{\crefalias{thm}{claim}}
\AddToHook{env/problem/begin}{\crefalias{thm}{problem}}
\crefname{claim}{claim}{claims}
\Crefname{claim}{Claim}{Claims}
\crefname{cor}{corollary}{corollaries}
\Crefname{cor}{Corollary}{Corollaries}
\crefname{defn}{definition}{definitions}
\Crefname{defn}{Definition}{Definitions}
\crefname{example}{example}{examples}
\Crefname{example}{Example}{Examples}
\crefname{fact}{fact}{facts}
\Crefname{fact}{Fact}{Facts}
\crefname{lem}{lemma}{lemmas}
\Crefname{lem}{Lemma}{Lemmas}
\crefname{notation}{notation}{notations}
\Crefname{notation}{Notation}{Notations}
\crefname{problem}{problem}{problems}
\Crefname{problem}{Problem}{Problems}
\crefname{rem}{remark}{remarks}
\Crefname{rem}{Remark}{Remarks}
\crefname{thm}{theorem}{theorems}
\Crefname{thm}{Theorem}{Theorems}
\Crefname{que}{Question}{Questions}

\begin{document}
\pagenumbering{roman}
\title{Computational Work Extraction: The Complexity of Catalysts}
\date{}
\thanksmarkseries{arabic}

\author{Atul Singh \textsc{Arora},\thanks{CQST, IIIT Hyderabad}~~Shantanav
\textsc{Chakraborty},\thanks{CQST and CSTAR, IIIT Hyderabad}~~Alexandru \textsc{Cojocaru},\thanks{QSL, University of Edinburgh}\\
Sreyas \textsc{Saminathan},\thanks{CQST, IIIT Hyderabad}~~and Uttam
\textsc{Singh}\thanks{CQST, IIIT Hyderabad}\thanks{Authors are listed alphabetically.}}
\maketitle
\begin{abstract}
Thermodynamics has been repeatedly reshaped by improving how one models the capabilities of agents that extract work. 
For an isolated quantum system, the maximal extractable work, the \emph{ergotropy}, assumes that the agent can apply any unitary. However, achieving it can be computationally intractable. In this work we introduce and study \emph{computational ergotropy}, restricting extraction to \emph{polynomially-sized} uniform unitary circuits acting on the system alone. 

We prove \emph{maximal separations}: $n$-qubit systems can have $\Theta(n)$ ergotropy, while every efficient process extracts negligible work, even for Hamiltonians consisting of single-qubit terms. We establish an \emph{unconditional} existential separation and give an explicit construction in the \emph{random oracle model}. Assuming the existence of quantum-secure pseudorandom functions, this separation extends to the \emph{plain model}.

This work uncovers an important connection between ergotropy and the \emph{complexity of catalytic computation}---computation where auxiliary qubits must be finally restored to their initial state. 
Relative to a random oracle, we establish relational and decision problems that:
(i) can be solved efficiently with $\lambda$ catalysts; but
(ii) cannot be solved by any algorithm with $c\lambda$ catalysts, for any $c<1$.
We show this by proving query lower bounds for \emph{quantum-space bounded} algorithms. 
    
As a consequence, for computational ergotropy, catalysts prove to be surprisingly powerful---there is a family of Hamiltonians and states for which catalysts enable efficient extraction of the full $\Theta(n)$ ergotropy, while every efficient non-catalytic process extracts negligible work. Furthermore, catalysts also allow us to introduce and instantiate the notion of \emph{pseudoergotropy}---analogous to pseudorandomness. On the other hand, we show catalysts do not change (information-theoretic) ergotropy. 

Finally, our work also sheds light on the classical aspect of the problem. First, most of our constructions rely on classical states and Hamiltonians and therefore imply analogous results for \emph{classical ergotropy}. Second, we show that certain \emph{proof of quantumness} protocols can be used to generically  \emph{separate classical and quantum} catalytic ergotropy.

Overall, these results point towards a theory of thermodynamics where computational complexity plays a fundamental role. 

\vspace{10cm}

\end{abstract}
\global\long\def\Enc{\mathsf{Enc}}%
\global\long\def\Dec{\mathsf{Dec}}%
\global\long\def\ct{\mathrm{ct}}%

\global\long\def\dens{\mathsf{Den}}%
\global\long\def\hams{\mathsf{Ham}}%
\global\long\def\uni{\mathcal{U}}%

\global\long\def\erg{\mathsf{erg}}%
\global\long\def\erghat{\mathsf{\widehat{\erg}}}%
\global\long\def\cerg{{\cal C}\text{-}\mathsf{erg}}%

\global\long\def\caterg{\mathsf{cat\text{-}erg}}%
\global\long\def\caterghat{\mathsf{cat\text{-}\widehat{\erg}}}%
\global\long\def\ccaterg{{\cal C}\text{-}\mathsf{cat\text{-}erg}}%

\global\long\def\StateGen{\mathsf{StateGen}}%
\global\long\def\desc{\mathsf{desc}}%
\global\long\def\hw{\mathsf{hw}}%

\global\long\def\poly{\mathsf{\mathsf{poly}}}%
\global\long\def\negl{\mathsf{negl}}%
\global\long\def\tr{\mathsf{\mathsf{\mathsf{tr}}}}%
\global\long\def\nonnegl{\mathsf{non\text{-}negl}}%

\global\long\def\ket#1{\left|#1\right\rangle }%
\global\long\def\bra#1{\left\langle #1\right|}%
\global\long\def\norm#1{\left\lVert #1\right\rVert }%
\global\long\def\braket#1#2{\left\langle #1\mid#2\right\rangle }%

\global\long\def\den#1{\left|#1\right\rangle \bra{#1}}%
\global\long\def\cats{\den 0^{\otimes\ell}_{L}}%

\global\long\def\gen{\mathsf{Gen}}%
\global\long\def\key{\mathsf{k}}%
\global\long\def\pk{\mathsf{pk}}%
\global\long\def\sk{\mathsf{sk}}%

\global\long\def\finv{\mathsf{FInv}}%
\global\long\def\feval{\mathsf{F}}%
\global\long\def\secp{\lambda}%

\global\long\def\binset{\left\{  0,1\right\}  }%
\global\long\def\cA{\mathcal{A}}%
\global\long\def\N{\mathbb{\mathbb{N}}}%

\global\long\def\catgen#1{(#1)\text{-}\mathsf{\widehat{cat}Gens}}%

\global\long\def\cat#1{(#1)\text{-}\widehat{\mathsf{cat}}}%
\global\long\def\icat#1{(#1)\text{-}\mathsf{cat}}%

\global\long\def\Fcat#1{(#1)\text{-}\mathsf{F}\widehat{\mathsf{cat}}}%
\global\long\def\Dcat#1{(#1)\text{-}\mathsf{D}\widehat{\mathsf{cat}}}%

\global\long\def\IO{\mathsf{IO}}%

\global\long\def\catalyst{\mathsf{catalyst}}%

\global\long\def\rel{\mathsf{rel}}%

\global\long\def\fat#1{\boldsymbol{#1}}%

\global\long\def\hamdesc{\mathsf{H}\text{-}\mathsf{desc}}%

\global\long\def\Exp{\mathsf{Exp}}%
\global\long\def\bs#1{\boldsymbol{#1}}%

\global\long\def\H{\mathcal{H}}%

\global\long\def\Den{\mathsf{Den}}%
\global\long\def\Ham{\mathsf{Herm}}%
\global\long\def\uni{\mathcal{U}}%
\global\long\def\cC{\mathcal{C}}%

\global\long\def\QPT{\mathsf{QPT}}%
\global\long\def\PPT{\mathsf{PPT}}%

\global\long\def\catalyst{\mathsf{catalyst}}%
\global\long\def\IO{\mathsf{IO}}%

\global\long\def\Enc{\mathsf{Enc}}%
\global\long\def\Dec{\mathsf{Dec}}%
\global\long\def\ct{\mathrm{ct}}%
\global\long\def\RO{\mathsf{RO}}%
\global\long\def\k{\mathsf{k}}%

\global\long\def\Exp{\mathsf{Exp}}%
\global\long\def\hw{\mathsf{hw}}%
\global\long\def\TD#1#2{\mathsf{TD}\left(#1,#2\right)}%

\global\long\def\H{\mathcal{H}}%
\global\long\def\C{\mathbb{C}}%
\global\long\def\O{\mathcal{O}}%
\global\long\def\I{\mathbb{I}}%
\global\long\def\E{\mathop{\mathbb{E}}}%
\global\long\def\N{\mathbb{\mathbb{N}}}%
\global\long\def\A{\mathcal{A}}%
\global\long\def\W{\mathcal{W}}%
\global\long\def\P{\mathcal{P}}%

\global\long\def\hamdesc{\mathsf{H}\text{-}\mathsf{desc}}%
\global\long\def\circdesc{\mathsf{desc}}%

\global\long\def\fnleq{\leq_{\infty}}%
\global\long\def\fngeq{\geq_{\infty}}%
\global\long\def\fng{>_{\infty}}%
\global\long\def\fnl{<_{\infty}}%
\global\long\def\fneq{=_{\infty}}%

\global\long\def\lb{\mathsf{lb}}%
\global\long\def\ub{\mathsf{ub}}%
\global\long\def\argmax{\operatorname*{arg\,max}}%

\global\long\def\fatrhoH{\fat{\rho},\fat H}%

\newcommandx\catU[3][usedefault, addprefix=\global, 1=, 2=, 3=]{\mathsf{cat}\text{-}\mathcal{U}_{#3}\left(#1,#2\right)}%
\newcommandx\CcatU[3][usedefault, addprefix=\global, 1=, 2=, 3=]{\mathcal{C}\text{-}\mathsf{cat}\text{-}\mathcal{U}^{#3}\left(#1,#2\right)}%
\newcommandx\PPTcatU[3][usedefault, addprefix=\global, 1=, 2=, 3=]{\mathsf{PPT}\text{-}\mathsf{cat}\text{-}\mathcal{U}^{#3}\left(#1,#2\right)}%

\newcommandx\Cuni[1][usedefault, addprefix=\global, 1=]{\mathcal{C}\text{-}\mathcal{U}\left(#1\right)}%
\global\long\def\CU{\mathcal{C}\text{-}\mathcal{U}}%
\newcommandx\CFcat[3][usedefault, addprefix=\global, 1=, 2=, 3=]{\mathcal{C}\text{-}\mathsf{Fcat}^{#3}\left(#1,#2\right)}%
\newcommandx\PPTFcat[3][usedefault, addprefix=\global, 1=, 2=, 3=]{\left(#1,#2\right)\text{-}\mathsf{F\widehat{cat}}^{#3}}%

\global\long\def\sdec{\mathsf{succ}_{\mathsf{dec}}}%

\global\long\def\ssearch{\mathsf{succ}_{\mathsf{search}}}%

\global\long\def\sinv{\mathsf{succ}_{\mathsf{inv}}}%

\global\long\def\dec{\mathsf{dec}}%

\global\long\def\Pisearch{\Pi^{\mathsf{search}}}%

\global\long\def\Pidec{\Pi^{\mathsf{dec}}}%

\global\long\def\caterghatQ{\mathsf{cat}\text{-}\mathsf{\widehat{ergQ}}_{{\cal D}}}%
\global\long\def\caterghatC{\mathsf{cat}\text{-}\mathsf{\widehat{ergC}}_{{\cal D}}}%

\global\long\def\bad{\mathsf{bad}}%

\global\long\def\CNOT{\mathsf{CNOT}}%

\global\long\def\SWAP{\mathsf{SWAP}}%

\global\long\def\pred{\mathsf{pred}}%

\global\long\def\Proj{\mathsf{Proj}}%

\global\long\def\WC{\W\mathsf{C}}%

\global\long\def\catPerm{\mathsf{cat\text{-}Perm}}%

\global\long\def\Perm{\mathsf{Perm}}%

\global\long\def\Cl{\mathsf{Cl}}%

\tableofcontents{}

\newpage{}
\clearpage
\pagenumbering{arabic}
\section{Introduction}

\noindent
Thermodynamics is the theory of how energy can be transformed, and at its heart lies the distinction between energy and useful energy. Energy is conserved, yet only part of it can be harnessed. Some of it can be extracted as work, while the rest cannot. This raises a basic question, namely \emph{how much work can be extracted from a given system?} Its answer has direct consequences for quantum technologies. It determines how much energy a quantum battery can deliver, how much work a quantum engine can produce in each cycle, and what it costs to power quantum devices as they grow in size~\cite{Campaioli2024,Myers2022}. Remarkable experimental advances have brought these questions within reach, including heat engines whose working medium is a single trapped ion~\cite{Rosnagel2016}, quantum batteries charged collectively in organic microcavities and superconducting circuits~\cite{Quach2022,Hu2022}, and direct measurements of the extractable work stored in individual quantum systems~\cite{Niu2024}.

\emph{Ergotropy.} How much work a system yields is determined not by its state alone but also by the process used to extract it. The answer changes, for instance, depending on whether the system may exchange heat with a reservoir or interact with other systems along the way. The cleanest setting is an isolated system taken through a cyclic process, that is, one that returns its Hamiltonian $H$ to its initial form. Quantum mechanics then dictates that every such process acts on the state $\rho$ of the system as a unitary, $\rho \mapsto U\rho U^\dagger$, and the work it extracts is the resulting decrease in energy. Remarkably, this minimal setting already reproduces the structure of thermodynamics! States from which no unitary can extract work are called \emph{passive}, and those that remain passive even when arbitrarily many copies are processed jointly are called \emph{completely passive}. Completely passive states turn out to precisely be thermal states~\cite{Pusz1978,Lenard:1978thm}. The Kelvin--Planck form of the second law, which forbids any cyclic process from extracting work from a system in thermal equilibrium, thus follows from unitary work extraction alone. For a general state $\rho$ on an $n$-qubit Hilbert space, the maximal work that unitaries can extract is its \emph{ergotropy}~\cite{Allahverdyan2004}, defined as:
\begin{align}
\label{eq:intro-erg}
    \erg_{\rho,H}:=\max_{U\in\uni(n)}\tr\left[H\left(\rho-U\rho U^{\dagger}\right)\right],
\end{align}
where $\uni(n)$ is the group of unitaries acting on the Hilbert space of $n$ qubits.

\emph{Efficiency.} %
While the unitary setting already produces a rich thermodynamic structure, it 
fails to capture various important aspects of the physical world. For instance, since the maximum in Eq.~\eqref{eq:intro-erg} ranges over every unitary on the system, this implicitly assumes optimal unitaries can be physically realised. 
An optimal unitary rotates the state into the eigenbasis of $H$ and rearranges its populations so that the largest occupy the lowest energies. For a pure state it maps the state to the ground state, which amounts to running a preparation of the state in reverse. In any physical realisation, whether by control fields, a sequence of quenches or a circuit of elementary gates, the cost of a unitary grows with its complexity, and almost all unitaries on $n$ qubits require a number of elementary operations exponential in $n$~\cite{knill1995approxcircuits}. The complexity of preparing quantum states and implementing unitaries has become a subject of its own~\cite{aaronson2016states,bostanci2025unitary}. It reveals that the two tasks differ sharply. Every quantum state can be prepared by an efficient algorithm with a single query to a suitable classical oracle~\cite{Rosenthal2024}, whereas some unitaries cannot be implemented in this way even approximately~\cite{Aaronson2007, Lombardi2024}. Work extraction fits neither framework. It asks neither for a particular final state nor for a particular unitary, since any process that lowers the energy extracts work and many processes may come close to the optimum. Therefore, the overarching question in this work is the following: %
\begin{center}
\emph{How does our understanding of work extraction from an isolated system change\\ when extraction is limited by computational efficiency?
}
\end{center}

To make progress towards obtaining concrete answers in this direction, one could consider the simplest setting captured by ergotropy. 
\begin{question}
How does one define efficiency in the context of ergotropy? Does it result in a meaningfully different notion from information-theoretic ergotropy?
\end{question}

\textbf{Computational Ergotropy.} We first develop a framework that allows formally posing and rigorously studying such questions. A prerequisite to this, is a definition of efficiency, and in particular, efficient work extraction. 
Such a notion is not well-defined for a single state, since for any fixed system the optimal unitary is some fixed circuit, however large. Efficiency is meaningful only asymptotically, and we therefore consider families of states and Hamiltonians of growing size. The definition must also specify what the extraction process has access to. We give it polynomially many copies of the state and an efficient description of the Hamiltonian, from which it produces a circuit that is then applied to the system. The \emph{computational ergotropy} of a family is the largest amount of work that efficient processes of this kind extract, asymptotically in the size of the system. 

With this definition in place, our first result is that the computational ergotropy can differ from the ergotropy as much as possible. There are pure states, the ideal batteries of the information-theoretic theory, whose ergotropy is essentially all of their energy, yet from which no efficient process extracts more than a negligible amount of work. This happens even though the Hamiltonian is local and consists of single-qubit terms, so the entire difficulty lies in the state. Since states obtained by an existence argument may themselves be impossible to prepare, we then show that the gap persists relative to a random oracle, the least structured oracle there is, and in the plain model for explicit, efficiently preparable states under standard cryptographic assumptions.

\vspace{1cm}

One could justifiably ask whether the separation between computational ergotropy and ergotropy really tells us something fundamental about nature? Or if it is an artefact of considering work extraction using a very restricted class of operations? For instance, in various places, \emph{catalysts} have been used to show surprising strengthening of operations one can apply, without changing the underlying physical justification for choosing the permitted operation.

\emph{Catalysts.} Registers that are borrowed and returned in their original state, called \emph{catalysts}, are a recurring theme in quantum information and computation~\cite{LipkaBartosik2021,LipkaBartosik2024}. In entanglement theory, an entangled catalyst enables transformations that are impossible under local operations and classical communication~\cite{Jonathan1999,Turgut2007,klimesh2007inequalities}, while a catalyst that is returned only approximately makes essentially any transformation possible~\cite{vanDam2003}. In complexity theory, a memory that must be restored to its initial contents still increases the computational power \cite{Buhrman2014, Cook2025, cook2025structureinplacespaceboundedcomputation}, a phenomenon that has recently been extended to quantum computation \cite{buhrman2025quantumcatalyticspace}\footnote{In catalytic space the borrowed memory starts in arbitrary contents rather than in the all-zero state, which is what makes restoring it a genuine constraint.}. In quantum thermodynamics, the catalysts change the laws themselves. Allowing them relaxes the thermo-majorization criterion for state transitions~\cite{Horodecki2013} to a weaker family of second laws~\cite{Brandao2015}, and this family collapses to a single law once the catalyst is returned exactly but only marginally, that is, once it is allowed to remain correlated with the system~\cite{Muller2018}. 

\emph{Some catalysts are too permissive.} The choice of which borrowed registers to allow is a substantive part of the model, and for our purpose a register that may remain correlated with the system is too permissive (See bottom row of \Cref{fig:models-catalysts}). A unitary acting on the system alone preserves its spectrum, whereas such a register lifts that constraint entirely and makes every state of at least the same entropy reachable to arbitrary accuracy~\cite{Boes2019, Wilming2021}. The extractable work would then be governed by entropy alone, and the question we are asking would change character before any computational restriction enters. Physically, such a register looks untouched when examined on its own, yet the joint state has changed. The work extracted can be traced to the correlations built along the way, and creating correlation is known to carry a thermodynamic cost~\cite{Huber2015}.

\begin{figure}[h!]
    \centering
    \includegraphics[width=0.7\linewidth]{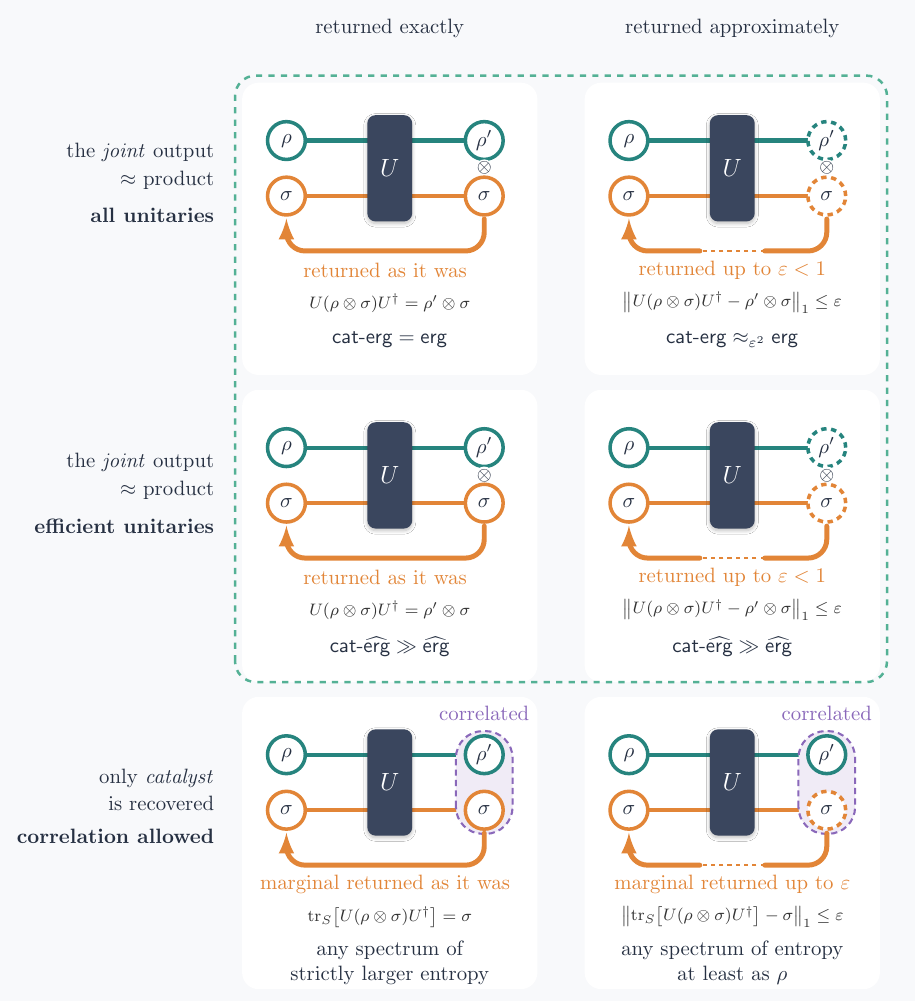}
    \caption{\emph{The role of catalysts in information-theoretic and computational work extraction.~} The system and catalyst initially have states $\rho$ and $\sigma$. The left column requires exact restoration, while the right column allows a trace-norm error of at most $\varepsilon$. In the upper two rows, the final joint state must equal or approximate a product containing the original catalyst state; these rows allow arbitrary unitaries and efficient circuits, respectively. The dashed outline encloses the settings studied in this work. (Nearly) exact product restoration leaves (information-theoretic) ergotropy unchanged, yet can increase computational ergotropy. The bottom row is an instance of a catalyst that is \emph{too permissive}. Here, restoration only of the catalyst marginal, permitting correlations with the system and changes in the system's spectrum.}
    \label{fig:models-catalysts}
\end{figure}

\begin{question}
    How does the notion of ergotropy change when one allows the use of catalysts? What happens if one also restricts to efficient operations? 
\end{question}

\textbf{Understanding Catalytic Computation.} To answer this question, we had to first define and study the \emph{complexity} of the appropriate \emph{catalytic model of computation}. While similar models have been considered in the literature, none of those quite matched our requirements here. We begin with the information-theoretic setting and investigate whether using catalyst qubits, permits one to apply operations beyond what is possible by unitary evolution of the system. We show that even if one is required to return a fixed catalyst to its original state, there is no advantage to using catalysts. In fact, even in the most lenient setting---i.e. even if the joint system (after the operation) is brought $\varepsilon$ close (in trace distance, say) to being a state where the catalyst is restored to being a tensor product---one does not gain much. 

We then turn our attention to considering the \emph{complexity} of \emph{catalysts in the efficient computation setting}. Here, we allow the model to access catalysts initialised to the all-zero state, in addition to the input state, on which it can apply a unitary. We require this unitary to be produced efficiently by an efficient classical machine. The catalytic constraint is that the catalyst must be restored to its initial state exactly. In this \emph{efficient setting}, we show that there is a problem that one can solve with catalysts but cannot solve without catalysts. In fact, we show something stronger---with insufficient quantum space (i.e. even if the catalyst is arbitrary correlated and not required to be restored in any way), the problem cannot be solved. These results are also unconditional, in the random oracle model. More fundamentally, this establishes exponential query-complexity separations for space bounded unitary quantum algorithms, which may be of independent interest. For random functions on $\lambda$-bit strings, the problem admits a classical reversible algorithm using three queries and $\lambda$ auxiliary qubits returned exactly to zero. For every fixed $0<c<1$, quantum unitary algorithms with at most $c\lambda$ auxiliary qubits require $2^{\Omega(\lambda)}$ queries for constant success probability, even when these qubits need not be restored. The lower bound imposes no restriction on the number of oracle-independent gates.

\vspace{0.5cm}

\textbf{Catalytic Ergotropy.} Using our understanding of the complexity of catalytic computation, we study catalytic ergotropy. Starting with (information-theoretic) catalytic ergotropy, we show that catalysts do not help---even under a very lenient requirement on catalytic restoration. We then define computational catalytic ergotropy (i.e. the setting where efficiency is enforced). We show that there is a family of states and Hamiltonians for which one can obtain near-maximal separation between \emph{computational ergotropy} and \emph{computational catalytic ergotropy}. As above, the results are stronger---in the sense that with $c\lambda$ quantum space (for $c<1$), no non-negligible work can be extracted efficiently while with $\lambda$ catalysts, catalytic ergotropy is maximal.

\vspace{0.5cm}

\textbf{Catalytic work-potential.} We define catalytic work-potential as the difference between the work one can efficiently extract with catalysts and without. We prove a sufficient criterion (as a consequence of the preceding analysis) that allows one to identify whether a Hamiltonian has (maximal) catalytic work potential. Informally, it says that any efficiently describable Hamiltonian whose ground state is easy to prepare, can be shown to have maximal catalytic work potential. 

\vspace{0.5cm}

\Cref{fig:models-catalysts} summarizes the contrast between information-theoretic and computational catalytic work extraction across the different catalyst restoration requirements. These results show that in the presence of catalysts, the computational theory changes dramatically while the information-theoretic theory remains unchanged. It may be worth noting that this is different from showing one quantity in the information-theoretic setting is larger than the analogous quantity in the computational setting---here, the comparison is being made between objects of the same theory. The relation between these objects is what has changed. This suggests that computational considerations may result in fundamental new insights in our understanding of nature. One way of probing this, is to ask if something concrete becomes possible using catalysts that was not possible without it.

\vspace{0.5cm}

\textbf{Pseudoergotropy.} It turns out that without catalysts, it is unclear if even simple efficient operations such as applying a keyed permutation are possible. On the other hand, we show that using catalysts, one can construct a family of keyed states and Hamiltonians that exhibit \emph{pseudoergotropy}. Intuitively, this shows that there are two such families of keyed states and Hamiltonians where (i) without the key, they are indistinguishable (computationally), (ii) the first family has low catalytic ergotropy, even when the key is given and (iii) the second family has high catalytic ergotropy, when the key is given. These properties are akin to what one expects from \emph{pseudorandomness}---without the key, the strings are indistinguishable but given the key, the first string has low entropy while the second has high entropy. Here, the property being captured is ``extractable work'' (or useful energy), instead of entropy.

\vspace{1cm}

\emph{Classical ergotropy.} Unlike purely quantum phenomena, such as entanglement \cite{arnon2023CET}, ergotropy remains well defined and meaningful even when one restricts to classical reversible operations. Furthermore, computational bounds apply to the classical world as well. It is therefore perhaps more natural to first ask questions about efficiency and catalysts in the classical setting. Indeed, when the state and the Hamiltonian are diagonal in the computational basis, classical reversible operations permute the populations among the energy levels. For a quantum unitary $U$, the transition probabilities $B_{ij}=|U_{ij}|^2$ form a doubly stochastic matrix, hence a convex combination of permutation matrices by the Birkhoff--von Neumann theorem. The final mean energy is therefore a weighted average of the energies obtained by those permutations. Some permutation extracts at least as much work as $U$, so information-theoretically, classical and quantum ergotropy coincide~\cite{Allahverdyan2004}.

\begin{question}
    How does efficiency affect ergotropy in the classical setting? Are there separations between quantum and classical ergotropy that are absent information-theoretically but emerge in the bounded setting? 
\end{question}

\textbf{Computational Classical Ergotropy.} Besides the existential separation between ergotropy and catalytic ergotropy, all of our aforementioned results were established using \emph{classical states} and Hamiltonians diagonal in the computational basis. Furthermore, the efficient processes we use to  achieve maximal work extraction, are nothing but permutations. Therefore, our preceding results directly translate to analogous classical statements. 

\textbf{Classical-Quantum Ergotropy Separation.} We show that there is a family of states and (extensive) Hamiltonians such that all of the following hold (under cryptographic assumptions): \\
\indent (i) information-theoretically, quantum ergotropy and classical ergotropy for this family is identical. \\
\indent (ii) Classical catalytic ergotropy is negligible.\\
\indent (iii) Quantum catalytic ergotropy is constant. \\
This is another feature of thermodynamics that seems to only appear in the computational theory---and it is again aided by catalysts.

\subsection{Main results}

We now state our results a bit more formally. See \Cref{fig:overview-of-results} for a schematic sketch of our main separations. 

\vspace{0.2cm}
\noindent\textbf{Computational Ergotropy.} We start with our separations between computational ergotropy and (information-theoretic) ergotropy.

\vspace{0.2cm}

\noindent\emph{Existential Separation.} We consider families of states $\fat{\rho}=\{\rho_\lambda\}_\lambda$ and Hamiltonians $\fat H=\{H_\lambda\}_\lambda$ on $n(\lambda)$-qubits, where members are parametrised by $\lambda$.  We use $\erg_{\fatrhoH}$ to denote the (information-theoretic) ergotropy and $\erghat_{\fatrhoH}$ to denote the computational ergotropy. Informally, one may treat both these as functions of $\lambda$. We show that there is an unconditional separation between these two notions.

\begin{thm}[Unconditional existential separation (informal version of \Cref{thm:unconditional-existential-separation})]\label{thm:introUnconditional}
There is a family of states and Hamiltonians $\fatrhoH$ such that $\poly(n)=(1-o(1))n$. Then, there exists a family of pure states $\fat{\rho}$ and a negligible function $\negl$ such that 
\begin{align*}
\erg_{\fat{\rho},\fat H} \geq_{\infty}\poly,\quad \text{and}  \quad\erghat_{\fat{\rho},\fat H} & \leq_{\infty}\negl.
\end{align*}
\end{thm}

To make comparison between functions, let's say $a,b:\N \to \N$ we use $a \fnleq b$ to mean that $a(\lambda) \ge b(\lambda)$ for all large $\lambda$s.

The theorem establishes unconditionally that the work stored in a quantum state can be almost entirely inaccessible to computationally bounded agents. The proof relies on Haar-random states, concentration inequalities and a counting argument. 

\vspace{0.2cm}

\noindent\emph{Constructive separation.}
For constructive separations, it helps to consider \emph{distributions} $\cal D$ over families of states and Hamiltonians $\fatrhoH$. We often write $(\fatrhoH,\fat{O})\leftarrow \cal D$ to denote that $(\fatrhoH,\fat{O})$ is sampled from $\cal D$, where $\fat{O}$ is an oracle (a random oracle in our case). For such distributions, we use $\erg_{\cal D}$ and $\erghat_{\cal D}$ to denote ergotropy and its computational variant.

We show that for a Hamming weight Hamiltonian and state of the form $$\rho_{\secp}:=\frac{1}{2^{n(\lambda)}}\sum_{r\in\binset^{n(\secp)}}\den{f_{\secp}(r)}$$ where $f$ is a length doubling function encoded in the oracle, the following holds. 

\begin{thm}[Separation in the random oracle model (informal version of \Cref{thm:explicit-sep-erg-erghat})]
Let $\poly(\lambda)=(1/2-o(1))n(\lambda)$. In the random oracle model, it holds that
\begin{align*}
\erg_{\cal D} \fngeq \poly,\quad \text{and}  
\quad
\nexists\quad \nonnegl \text{ s.t. } \erghat_{\cal D}  \fngeq\nonnegl.
\end{align*}
where $\nonnegl$ is any non-negligible function.
\end{thm}

Due to maximisation over efficient algorithms in the definition of $\erghat$, it is not exactly correct to say $\erghat \fnleq \negl$ here but we will often say in words that $\erghat$ is negligible if no $\nonnegl$ work can be extracted.

Returning to the construction, the intuition here is that to a computationally bounded algorithm, the state looks maximally mixed, even given access to the random oracle. However, an unbounded algorithm can essentially invert the function and return to a state of the form $\sum_{r}\den{r,0}$. 

While this is not true in general, in our construction, we can instantiate the random oracle using a (post-quantum) pseudorandom function to obtain the following. 

\begin{thm}[Separating $\erg_{{\cal D}}$ from $\erghat_{{\cal D}}$ in the plain
model (informal version of \Cref{thm:explicit-sep-erg-erghat-prf})]
Suppose pseudorandom functions exist. Then, there is an efficiently samplable distribution $\cal D$ over a family of states and Hamiltonians (see \Cref{def:explicit-fam-state-Ham-PRF}) such that 
\begin{align*}
\nexists\ \ \nonnegl\ \ \text{s.t.}\ \ \erghat_{{\cal D}} & \fngeq\nonnegl,\text{ while}\\
\erg_{{\cal D}} & \fngeq\frac{n}{2}
\end{align*}
where $\nonnegl$ is any non-negligible function.
\end{thm}

\paragraph{Understanding Catalytic Computation}

We start with the information theoretic setting. 

\emph{Catalytic Computation (unbounded/information theoretic).}
We denote by $\catU[n][\ell,\varepsilon]$ unitaries that act on an $n$-dimensional $\IO$ register and an $\ell$-dimensional $\catalyst$ register such that the resulting state is $\varepsilon$ close in trace distance to a product state. We call these \emph{approximate catalytic unitaries}. Let us denote by $C^{\eta}_U(\cdot)$ the channel implemented by such unitaries $U$ where the catalyst is $\ket{\eta}$ and one only keeps the $\IO$ register (the $\catalyst$ register is traced out). Then, we show that all such channels are essentially just unitaries acting on the $\IO$ register alone. 

\begin{thm}
Let $0\leq\varepsilon<1$, and let
$U\in\catU[n][\ell,\varepsilon]$.
For every pure catalyst state $\eta=\den{\eta}\in\Den(\ell)$,
there exists a unitary $V_{\eta}\in\uni(n)$ such that
\begin{equation}
\max_{\rho\in\Den(n)}
\TD{C^{\eta}_{U}(\rho)}
{C_{V_{\eta}}(\rho)}
\leq \frac{3\varepsilon^{2}}{2}.
\end{equation}
\end{thm}

\emph{Efficient Catalytic Computation.} Since nothing changes substantially when catalysts are introduced in the information theoretic setting, we now turn to the computationally bounded setting. To consider the efficient catalytic setting, we define $\catU[n][\ell][\left|\eta\right\rangle]$ just as we defined $\catU[n][\ell,\varepsilon]$ but with the added restriction that the catalyst $\ket{\eta}$ must be restored exactly. 

With this, we can describe the class $\catgen{n,\ell}$: A collection of efficient classical algorithms $\cal A$ that on input $1^\lambda$, produce a description of a unitary $U_\lambda \in \catU[n][\ell][\left|\eta\right\rangle]$ and a projector $\Pi_\lambda$ that acts on the $\IO$ register. We define corresponding class of decision problems (i.e. decision problems that algorithms in $\catgen{n,\ell}$ can solve) to be $\cat{n,\ell}$.

\begin{thm}[Decision separation between catalysts (informal version of \Cref{thm:sep-cat-vs-nocat-dec})]
Let $n,\ell,\ell':\N\to\N$ be the polynomials $n(\lambda)=2\lambda$, $\ell'(\lambda)=\lambda$
and $\ell(\lambda)=c\cdot\secp$ for $c<1$. In the random oracle
model, it holds that 
\[
\cat{n,\ell}\subsetneq\cat{n,\ell'}.
\]
\end{thm}

This shows that there is a decision problem that can be solved with $\lambda$ catalysts but not with $c\cdot\lambda$ for any $c<1$. In effect, even though catalysts do not affect the information theoretic picture, the computational complexity of catalysts changes when one restricts to efficient computations. 

\paragraph{Catalytic Ergotropy}

Building on our complexity theoretic results, we now look at how catalysts affect ergotropy. When discussing ergotropy, we restrict to exact catalytic restoration in our explanations. Robustness to arbitrary errors is handled by showing quantum space lower bounds. 

Let $\caterg_{\rho,H}$ denote the work one can extract from the state $\rho$ using catalytic unitaries that restore the catalyst exactly. Then, it is immediate from the preceding discussion that 
\begin{flalign*}
\erg_{\rho,H} & =\caterg_{\rho,H}
\end{flalign*}
for all states and Hamiltonians.

To define the analogous quantity in the computational setting, we denote the work extracted by an efficient algorithm $\cal A$, using $\ell$ catalysts, from a state $\rho$ and Hamiltonian $H$ as 
\begin{align*}
\W^{(\ell)}_{\fat{\rho},\fat H,\A} & :=\mathbb{\E}\Biggr[\tr\left[H\otimes\I_{2^{\ell}}\left(\rho\otimes\den{0^{\ell}}-U\left(\rho\otimes\den{0^{\ell}}\right)U^{\dagger}\right)\right]:\\
 & \qquad\qquad\qquad\qquad\qquad\qquad\qquad\qquad\qquad\begin{array}{r}
\circdesc(U)\leftarrow\A^{O_{\fat{\rho}},\hamdesc}(1^{\secp}),\\
U\in\catU[n][\ell][\left|0^{\ell}\right\rangle ]
\end{array}\Biggr]
\end{align*}
where the number of catalysts and the system size is now a function of $\lambda$. Intuitively, this is computing the work one can extract using the unitary specified by $\cal A$ using $\ell$-many catalysts---where the catalyst is restored exactly.

We can now define \emph{computational ergotropy in the catalytic setting} that we denote by $\caterghat$. We say $\caterghat_{\cal D} \fngeq \lb$ if there is some efficient algorithm $\cal A$ such that $\E_{(\fatrhoH)\leftarrow {\cal D}} {\cal W}^{(\ell)}_{\fatrhoH,\cal A} \fngeq \lb$. We show that $\erghat$ can be maximally separated from $\erghat$.

\begin{thm}[Separation between computational ergotropy and catalytic computational ergotropy (informal version of \Cref{thm:sep-cat-erg-hat-and-erg-hat})]
There is a distribution $\cal D$ over the family of states, Hamiltonians and the random oracle, $\fatrhoH,\fat O$ for which, in the random oracle model, it holds that 
\begin{align*}
\nexists\ \ \nonnegl\ \ \ \text{s.t. }\erghat_{{\cal D}} & \fngeq\nonnegl\text{, while }\\
\caterghat_{{\cal D}} & \fngeq\frac{n}{2}
\end{align*}
where $\nonnegl$ is any non-negligible function.
\end{thm}

\begin{figure}
    \centering
    \includegraphics[width=0.8\linewidth]{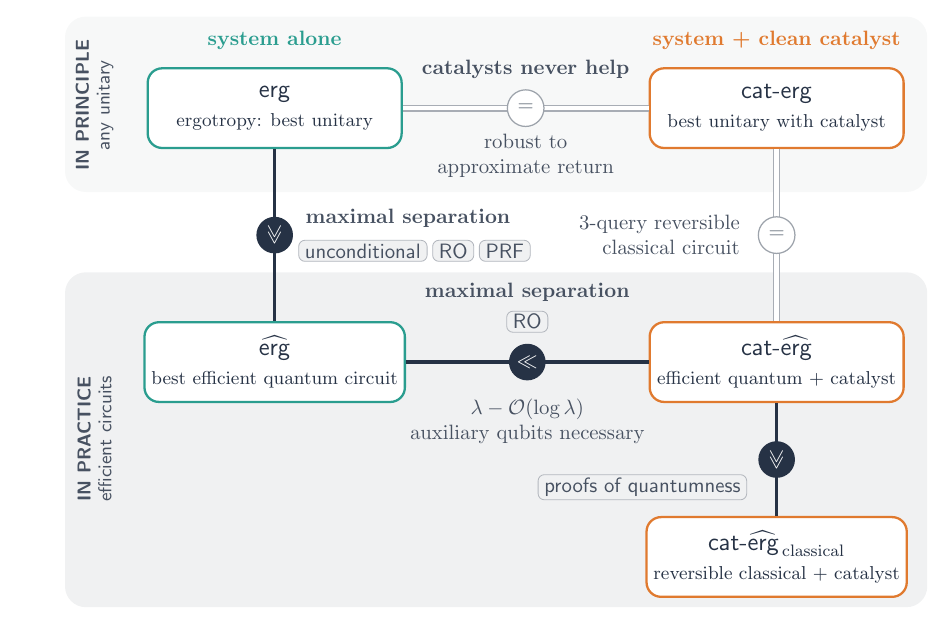}
    \caption{\emph{Overview of the separation results.}}
    \label{fig:overview-of-results}
\end{figure}

\paragraph{Catalytic Work-Potential}

We now give a sufficient condition to conclude that an efficient extensive Hamiltonian has catalytic work potential, i.e. $\Delta$ is not negligible (recall $\Delta$ is the difference between efficiently extractable work with and without catalysts). %

\begin{thm}[Criterion for catalytic-work-potentiol (informal version of \Cref{thm:cat-work-potential})]
Let $n:\N\to\N$ be the polynomial $n(\lambda)=2\lambda$. Let $\fat H=\{H_{\lambda}\}_{\lambda}$
be a family of efficient extensive Hamiltonians acting on $n(\lambda)$
qubits. Then, $\fat{H}$ has catalytic work potential if there is a polynomial $\ell:\N\to\N$ and an algorithm ${\cal A}\in\cat{n,\ell}$
such that $\tr[H_{\lambda}\sigma_{\lambda}]=\nonnegl(\lambda)$ where 
\[
\sigma_{\lambda}:=\E_{\circdesc(\Pi_{\lambda},U_{\lambda})\leftarrow{\cal A}(1^{\lambda})}U\left(\mathbb{I}_{2^\lambda}\otimes\den 0^{\lambda}\otimes\den 0^{\ell(\lambda)}\right)U^{\dagger}.
\]

\end{thm}

\paragraph{Pseudo-ergotropy} 
Here, it helps to introduce keyed distributions. Just as earlier we had distributions over families of states and Hamiltonians, now we also include families of keys. We say a keyed distribution $\cal D$ over states and Hamiltonians has \emph{$(c,d)$-catalytic-pseudoergotropy} if there is another keyed distribution $\cal D'$ such that $\caterghat_{\cal D} \fnleq c$, $\caterghat_{\cal D'} \fngeq d$ (assume the extraction algorithms see the key) and without the key, the states produced by the two distributions are computationally indistinguishable. Here, we require both 
Here, it helps to introduce keyed distributions. Just as earlier we had distributions over families of states and Hamiltonians, now we also include families of keys. We say a keyed distribution $\cal D$ over states and Hamiltonians has \emph{$(c,d)$-catalytic-pseudoergotropy} if there is another keyed distribution $\cal D'$ such that $\caterghat_{\cal D} \fnleq c$, $\caterghat_{\cal D'} \fngeq d$ and without the key, the states produced by the two distributions are computationally indistinguishable. 

\begin{thm}[Pseudoergotropy can be realised (informal version of \Cref{thm:cons1-thomas-inspired-nonlocalH})] Suppose quantum secure pseudorandom permutations exist. Then, for $d(\lambda)=4\lambda$, there is some polynomial $\ell:\N\to\N$ and an efficiently sampleable keyed distribution
${\cal D}$ that has $(0,d)$-catalytic-pseudoergotropy using $\ell$ catalysts. 
\end{thm}

The number of catalysts depends on the details of the pseudorandom permutation used. This seems to be a concrete application of catalysts in the computational setting. It is unclear whether one can realise any reasonable notion of pseudoergotropy without using catalysts.

\paragraph{Classical-Quantum Ergortopy Separation}

The fact that our separations work for classical ergotropy is a direct consequence of our constructions that we detail later. Here we present our separation between quantum and classical catalytic ergotropy, denoted by $\caterghatC(\ell)$ and $\caterghatQ(\ell)$ respectively. Deviating slightly from our previous definitions, here we allow negligible error in catalytic restoration. This is needed since we need to be able to run pseudo-deterministic circuits catalytically. 

\begin{thm}[Classical vs quantum catalytic ergotropy separation (informal version of \Cref{thm:classical-vs-quantum-erg})]
Suppose factoring is hard for efficient classical machines. Then, in the common-reference string model, there is a distribution $\cal D$ over families of states, extensive Hamiltonians and common-reference strings, such that 
\begin{align}
\nexists\quad\nonnegl\quad\text{s.t .}\caterghatC(\ell) & \fngeq\nonnegl\label{eq:soundness-cq-sep}\\
\exists\quad\negl'\quad\text{s.t. }\caterghatQ(\ell) & \fngeq1-\negl'\label{eq:completeness-cq-sep}
\end{align}
where $\ell$ is some polynomial, $\nonnegl$
is a non-negligible function while $\negl'$ is some negligible function.

\end{thm}

This result can be instantiated with any Proof of Quantumness (PoQ) protocol that satisfies certain properties---we require it to be a 2-message, public-coin protocol whose soundness and completeness can be made almost perfect.

We end by emphasising that all our separations are relative to local Hamiltonians.

\subsection{Prior Work}
Recent work has looked at the problem of qubit erasure under computational constraints. \cite{Munson2025} uses complexity entropies to bound the work cost of erasure with a bounded number of computational gates. Subsequently,~\cite{zhao2026learningerasure} obtains optimal many-copy erasure protocols when quantum learning and state preparation are efficient, and proves cryptographic obstructions to optimal erasure and work extraction. \cite{zhao2026learningerasure} also discusses encrypted batteries with restored ancillary registers, but their extraction converts purity into work using a thermal reservoir. Thus, these processes permit heat exchange with a reservoir and changes to the system's spectrum. In contrast, our work concerns ergotropy, the work extracted under unitary transformations from isolated systems. 

Exact catalysts returned uncorrelated with the system are already known to leave unrestricted ergotropy unchanged~\cite{Lipka2024catalysis}. We show that approximate restoration permits only a quadratic increase in extractable work: if a joint unitary maps every pure system state tensored with a fixed pure catalyst to within trace distance $\varepsilon$ of a product state whose catalyst factor is the initial catalyst state, then allowing such unitaries increases the maximum extractable work by at most $\varepsilon^2\|H\|_\infty$. Munson et al.~\cite{Munson2025} also discuss how pure ancillas could reduce erasure complexity. Relative to a random oracle, we show that auxiliary qubits initialized to zero and returned exactly and uncorrelated enable efficient extraction of the full $\Theta(n)$ ergotropy from systems whose computational ergotropy is negligible. 

Classical catalytic space~\cite{Buhrman2014} and its quantum analogue~\cite{buhrman2025quantumcatalyticspace} study computation with borrowed memory whose unknown initial contents must be restored. Our catalytic model instead assumes that the catalyst is in some known fixed state, uncorrelated with the system. In fact, for our algorithms, we assume that the auxiliary qubits are in fact initialized to zero. Hao, Huang, and Liu~\cite{hao2026memory} show that attaining optimal query complexity for short-answer problems can require memory exponential in the answer length. Subsequently,~\cite{gao2026queryspace} exhibits a total Boolean function whose quantum query advantage disappears when workspace is restricted. On the other hand, we establish an exponential separation within unitary computation using only a constant-factor increase in total space, in the random oracle model. This separation makes reusable workspace essential for efficient work extraction despite its leaving information-theoretic ergotropy unchanged.

Recent work studies how computational restrictions affect entanglement~\cite{aaronson2024pseudoentanglement,leone2025entanglement}, magic~\cite{gu2024pseudomagic}, nonlocality~\cite{gluch2024nonlocality}, quantum entropies~\cite{avidan2026quantumcomputationalentropies}, and deep thermalization~\cite{chakraborty2025deepthermalization}. In particular,~\cite{arnon2023CET} constructs computationally indistinguishable ensembles with low entanglement cost and high distillable entanglement, respectively, under efficient local operations and classical communication. Our construction uses their reduced states to conceal an extensive difference in ergotropy between fully separable mixed states. For a public, key-independent Hamiltonian consisting of single-qubit terms, the maximally mixed reference has zero ergotropy, whereas the other ensemble has extensive ergotropy on average over the key. The permutation key enables classical reversible extraction of the full ergotropy with exactly restored auxiliary qubits. Without it, efficient quantum protocols satisfying the same restoration requirement extract negligible expected work.

The same state ensembles, equipped with a zero Hamiltonian, have zero and $\Theta(n)$ free-energy excess above equilibrium, respectively, while remaining computationally indistinguishable given polynomially many copies. Assuming quantum-secure pseudorandom permutations, this gives a counterexample to the conjectured nonexistence of pseudo-nonequilibrium states~\cite{Watanabe2026} (See \Cref{rem:pseudoergotropy-free-energy}).

\section{Technical Overview}
\label{sec:tech-overview}
We briefly explain the key ideas and constructions we use to establish the main results.

\subsection{Separating ergotropy from computational ergotropy}
\label{subsec:overview-ergotropy}
For the unconditional existential separation, we consider the Hamiltonian $H_n=\sum_{i=1}^nX_i$ and a Haar-random pure state. Its energy concentrates near zero, while its ergotropy is nearly $n$ because the ground energy is $-n$. Every fixed unitary preserves the Haar distribution. Concentration is strong enough to take a union bound over all circuits with $2^{n/4}$ gates, simultaneously bounding their extracted work by $2n\,2^{-n/3}$. Thus, learning copies cannot help an algorithm choose a successful polynomial-size circuit (see Lemma~\ref{lem:bounding-erg-Up} and Theorem~\ref{thm:unconditional-existential-separation}).

We first prove a constructive separation in the random-oracle model. Let $n=n(\lambda)$ and let $f_\lambda:\binset^n\to\binset^{2n}$ be a uniformly random function. For our state-Hamiltonian pair, we have
\begin{equation}
\label{eq:overview-explicit-battery}
\rho_\lambda=\frac{1}{2^n}\sum_{r\in\binset^n}\den{f_\lambda(r)},
\qquad H_\lambda=\I_{2^n}\otimes H_{\hw(n)}.
\end{equation}
Sampling $r$ uniformly and making only one oracle query prepares $\rho_\lambda$, which has a rank of at most $2^n$, so its entire support can be moved into the $2^n$-dimensional ground space of $H_\lambda$. Its ergotropy therefore equals its initial energy and has expectation $n/2$ over $f_\lambda$.

Nevertheless, $\rho_\lambda$ is computationally indistinguishable from the maximally mixed state, even when the distinguisher can query $f_\lambda$. The proof compares $f_\lambda(r)$ with an independent uniform string by reprogramming the oracle at the hidden argument $r$ and applying the one-way-to-hiding lemma (see Lemma~\ref{lem:expanding-RO-uniformity}). Since the maximally mixed state is invariant under unitaries, non-negligible work extraction would distinguish the ensembles by estimating their energies. The reduction also generates any copies used to choose the extraction circuit through forward oracle queries. This establishes the random-oracle separation (see Theorem~\ref{thm:explicit-sep-erg-erghat}).

For the plain-model separation, we replace the random function by a quantum-secure PRF whose key is unavailable to the extraction algorithm. PRF security transfers the indistinguishability argument, giving efficiently preparable states with extensive ergotropy but negligible efficient work extraction, leading to a maximal separation (see Theorem~\ref{thm:explicit-sep-erg-erghat-prf}).

\vspace{0.5cm}
Next, we move on to the catalytic case. Exact restoration of a fixed pure catalyst for every system state leaves ergotropy unchanged. We show that this is robust: even if we allow for the final product state of the system and register to be $\varepsilon$-close to the ideal state (with exactly restored catalyst), we are able to show that ergotropy does not change other than by a factor of roughly $\varepsilon^2$. The difficulty in the approximate restoration case, is that we must identify one unitary for all inputs. Suppose every pure output is within trace distance $\varepsilon<1$ of a product containing the original catalyst state $\ket\eta$. Projecting onto $\ket\eta$ defines a system map $L$ whose singular values lie in $[\sqrt{1-\varepsilon^2},1]$. Its polar decomposition gives a unitary $V$ with $\norm{L-V}_\infty\leq\varepsilon^2$. The orthogonal catalyst component has weight at most $\varepsilon^2$, and its cross terms vanish under the partial trace. Thus the reduced channel is within trace distance $3\varepsilon^2/2$ of $V(\cdot)V^\dagger$ on every input. Additional work is at most this quadratic error times the spectral width of $H$ (see Lemma~\ref{lem:dirty-vs-no-cat-unbounded}).

\subsection{Query separations for catalytic computation}
\label{subsec:overview-query-separations}

We prove both relational and decision separations. Let $f,g:\binset^\lambda\to\binset^\lambda$ be independent random functions. For the relational version, we consider the mapping $\ket{x,y}\mapsto \ket{x, y\oplus f(g(x))}$. It is easy to see that with $\ell(\lambda) = \lambda$ catalyst qubits and given access to standard random oracles for $f$ and $g$, only three queries suffice to implement
$$
\ket{x,y,0^\lambda}\mapsto\ket{x,y\oplus f(g(x)),0^\lambda}.
$$ 

Our main technical contribution is to prove that no quantum algorithm making only polynomially many queries to oracles for $f$ and $g$ can solve this problem with constant probability using $\ell(\lambda)=\lfloor c\lambda\rfloor$ auxiliary qubits, for any fixed $0<c<1$. For this, we consider the inverse problem of 
$$
\ket{x,f(g(x)),0^{\ell}}\mapsto \ket{x, 0^{\lambda}, 0^{\ell}}.
$$
In fact, our lower bound proves something much more general: we do not require the $\ell(\lambda)$ auxiliary qubits to be exactly restored. Thus, this result is a query complexity separation for space bound quantum algorithms. We establish this by quantitatively relating the success probability of the relational computation to the probability of recovering $x$ from $g(x)$, where $g$ is a uniformly random function (see \Cref{lem:sep-cat-vs-lesscat}).

The challenge is to rule out quantum circuits that distribute intermediate information across the system and auxiliary registers; one cannot simply assume that they must store $g(x)$ in a separate register. We overcome this by quantitatively comparing erasure success with random-function inversion. Constant erasure success would give an inverter that recovers $x$ from $g(x)$ with probability at least $2^{-\ell}/\poly(\lambda)$. Although this probability can be exponentially small, for $\ell=c\lambda$ it exceeds the $\poly(\lambda)2^{-\lambda}$ success permitted by polynomially many quantum queries.

The decisive technical point is to obtain a loss determined only by the auxiliary register's dimension $2^\ell$. We show that successful erasure must involve components that query $f$ at $g(x)$: the component avoiding this address contributes at most $2^{\ell-\lambda}$ to average success. Such queries provide the starting point for an inverter. Given $g(x)$, it initializes the other registers uniformly at random and reverses a suitable circuit prefix to recover $x$. Randomizing these registers appears to incur a factor $2^{\lambda+\ell}$, but averaging over the uniform value $f(g(x))$ cancels the factor $2^\lambda$ associated with the second register. Apart from polynomial factors in the query counts, the reduction therefore loses only $2^\ell$. Comparing its inversion probability with the quantum query bound for random-function inversion limits average erasure success to a polynomial in $\lambda$ times $2^{-(1-c)\lambda}$. This is negligible, even without requiring the erasing circuit to restore its auxiliary register.

For the decision separation, we consider the problem of distinguishing $(x,f(g(x)))$ from $(x,r)$, where $x,r$ are independent uniform strings. A distinguisher need not erase the second register or preserve $x$. Its accepting subspace can therefore have large dimension, so the dimension bound used for relational computation does not establish the decision lower bound.

We instead bound the difference between the two average acceptance probabilities. As in the relational argument, the non-oracle gates and accepting projector are fixed independently of $f,g$. Conditional on $g,x$ and the other values of $f$, the operator representing the component that never queries $f$ at $g(x)$ is fixed, while both $f(g(x))$ and $r$ are uniform. This component consequently gives the same average acceptance probability in both cases, even if that probability is large.

Distinguishing advantage must therefore involve components that query $f$ at $g(x)$, including their interference with the component avoiding this address. We control this interference through the amplitudes of these components. The same reversed-prefix construction then relates the squared distinguishing advantage to inversion success, with a loss of $2^\ell$ and polynomial factors in the query counts. The random-function inversion bound consequently makes the advantage negligible for polynomially many queries and $\ell(\lambda)=\lfloor c\lambda\rfloor$, for any fixed $0<c<1$. No restoration of the auxiliary register is required (see Lemma \Cref{lem:decision-cat-separation-robust}).

\subsection{Catalytic Work Extraction and Catalytic Work Potential}
\label{subsec:overview-catalytic-work}

We turn the query separations into a computational work-extraction separation by assigning energy to the second register. Once again, consider
\begin{equation}
\label{eq:overview-catalytic-battery}
\rho_\lambda=\frac{1}{2^\lambda}
\sum_{x\in\binset^\lambda}\den{x,f(g(x))},
\qquad
H_\lambda=\I_{2^\lambda}\otimes H_{\hw(\lambda)}.
\end{equation}
With $\lambda$ auxiliary qubits, the three-query circuit maps each $(x,f(g(x)))$ to $(x,0^\lambda)$ and restores every auxiliary qubit exactly to zero. Clearing the second register therefore moves the entire support into the ground space, decreasing the average energy. Hence catalytic computational ergotropy is near-maximal. 

The lower bound for the decision problem immediately proves that lowering the energy is nearly impossible for any quantum algorithm. Indeed, the states $\rho_\lambda$ and $\I/2^{2\lambda}$ have the same expected initial energy, while every system unitary leaves the latter unchanged. Consequently, non-negligible expected extraction from $\rho_\lambda$ would produce a non-negligible difference in their final energies. Measuring a uniformly chosen term of $H_\lambda$ would detect this difference, contradicting the decision lower bound. Thus polynomially many queries give negligible expected extraction without auxiliary qubits, whereas the exactly restored workspace enables full extraction. Here the extraction circuit's non-oracle gates are fixed independently of $f,g$, as required by the query lower bound (see Section~\ref{sec:Computational-Catalytic-Work}). 

The same connection allows us to define a notion of catalytic work potential that works for broad classes Hamiltonians. For any public, oracle-independent local $H_\lambda$ of polynomial operator norm, $\rho_\lambda$ has expected energy $\bar E_\lambda=\tr(H_\lambda)/2^{2\lambda}$, and the decision bound again excludes non-negligible extraction without auxiliary qubits. The catalytic circuit instead makes $(\I_{2^\lambda}/2^\lambda)\otimes\den{0^\lambda}$ efficiently reachable. If an exactly catalytic circuit can transform this state into one of energy $E_\lambda<\bar E_\lambda$, their composition extracts $\bar E_\lambda-E_\lambda$ expected work. A non-negligible gap therefore gives catalytic work potential.

\subsection{Pseudoergotropy}
While for most pseudo-resources instantiating the low resource state using PRS works, it does not suffice for pseudoergotropy, since the Haar random state is a low-ergotropy state.

We leverage the ability to use keyed cryptographic primitives in the catalytic setting (see \Cref{sec:Catalytic-Computation}) to give an explicit distribution ${\mathcal{D}}$ with $(0,\O(n))$-catalytic pseudoergotropy. Taking inspiration from \cite{arnon2023CET}, we define contracting and expanding functions $h:\binset^m \to \binset^{m/2}$ and $g:\binset^{m/2} \to \binset^{m}$, by truncating the output and padding the input to the PRP respectively, and $f:\binset^m\to\binset^m$ which applies the PRP as is. One can show that  expanding$\circ$contracting$(x)$ is computationally indistinguishable from applying the PRP on $x$, when the key is not known (see \Cref{lem:fx_equals_ghx}, \Cref{lem:families-are-indistinguishable}).
Using this, one can construct the states $\sigma=\sum_{x\in\binset^{m(\secp)}}\den{f(x)}$  (maximally mixed), and state $\rho=\sum_{x\in\binset^{m(\secp)}}\den{g(h(x))}$ (relatively less mixed), that are computationally indistinguishable in the absence of the key $\k$ (see \Cref{def:explicit-fam-state-Ham-PRF-1}).
With access to the key $\k$ one can invert the function $g$, bringing the last $m(\secp)-n(\secp)$-qubits of $\rho$ back to 0 (see \Cref{lem:psi-has-high-caterghat}). For the padded Hamming weight Hamiltonian, one can therefore catalytically extract work from the state $\rho$ while no work can be extracted from the state $\sigma$.

\subsection{Separating classical and quantum ergotropy}

Our classical-quantum separation uses classical states and classical Hamiltonians: the advantage comes from the computational power of the extraction operations, rather than from coherence in the initial state. We start from an amplified, pseudo-deterministic public-coin proof of quantumness. Given a challenge $c$, a quantum prover produces its unique accepting response with overwhelming probability, whereas any efficient classical prover succeeds only with negligible probability. We take $c$ to be the common reference string and encode the verifier's computation into a $4$-local classical Hamiltonian $H_c$. The Hamiltonian essentially penalises inconsistent computation steps and rejection, thus its zero-energy basis states encode valid accepting transcripts. An accepting response can therefore be efficiently recovered from any such ground state (see \Cref{claim:CircToHam}).

For the initial state $\rho_c$, we use the valid computation transcript obtained by running the verifier on the all-zero input. This state is efficiently preparable and violates no computation constraints: its only possible penalty is for rejection. Its expected enery is consequently $1 - \negl(\lambda)$, since otherwise the all-zero response would violate classical soundness. The important feature here is that every nonzero energy of $H_c$ is at least one. Thus, a classical extractor's expected work is bounded by the probability that its output is a ground state. Any non-negligible work extraction would therefore produce an accepting response with non-negligible probability, contradicting the soundness of the proof of quantumness.

For quantum extraction, finding an accepting response is not enough: the computation must also return its auxiliary workspace. This is where the pseudo-deterministic property is used. A compute-copy-uncompute transformation turns the quantum prover into a catalytic computation with negligible restoration error (see \Cref{claim:pseudo-det-to-catalytic}). The extractor first undoes the initial transcript computation, then generates the accepting response and constructs its valid verifier transcript, uncomputing the auxiliary workspace. The ideal output is a ground state and the negligible construction error contributes only negligible expected energy as $\norm{H_c}_\infty$ is polynomially bounded. As a result, polynomially many catalyst qubits suffice for quantum extraction of $1 - \negl(\lambda)$ work, while efficient classical extraction yields no non-negligible work (\Cref{subsec:classical-quantum-sep}).

\newpage{}

\subsection*{Acknowledgements}
We thank Andrea Coladangelo and Venkata Koppula for many helpful discussions and comments. ASA, SC and US acknowledge funding from the National Quantum Mission, an initiative of the Department of Science and Technology, Govt of India. ASA, SC and US also acknowledge support provided by the Foundation for QC Innovation (FQCI), DST-NQM T-Hub at IISc Bengaluru, in facilitating this project. SC and US acknowledge support from the Ministry of Electronics and Information Technology (MeitY), Government of India, under Grant No. 4(3)/2024-ITEA. ASA acknowledges funding from ANRF, Government of India, under Grant No. ANRF/ARG/2025/010779/MS. SC thanks Fujitsu Ltd, Japan for funding. AC acknowledges support from the National Science Foundation grant CCF-1813814, from the AFOSR under Award Number FA9550-20-1-0108 and from the Quantum Advantage Pathfinder project.

\newpage{}

\section{Preliminaries and Notation}

\branchcolor{darkgray}{We introduce the relevant notation and the requisite preliminaries.}

\subsection{Notations}

\subsubsection*{Quantum Notation}

\branchcolor{darkgray}{We start by specifying standard quantum notation.}
\begin{notation}
\label{nota:basic-notations}We adopt the following notations throughout
this article.
\end{notation}

\begin{itemize}
\item Symbols for integer valued parameters:
\begin{itemize}
\item $\secp$: Security parameter,
\item $n$: Number of qubits in the system,
\item $\ell$: Number of catalytic qubits,
\item $q$: Number of queries (typical; depends on the context).
\end{itemize}
\item Registers are denoted by math sans-serif symbols (e.g. $\mathsf{IO,catalyst}$)
and the corresponding\\
Hilbert spaces are denoted by ${\cal H}$ and the register in the
subscript (e.g. ${\cal H}_{\mathsf{IO}},\mathcal{H}_{\mathsf{catalyst}}$).
\item $\Den(n)$: The set of density operators on $\left(\C^{2}\right)^{\otimes n}$.
\item $\Ham(n)$: The set of Hermitian operators on $\left(\C^{2}\right)^{\otimes n}$.
\item $\Proj(n)$: The set of all projectors on $n$ qubits, i.e. $\Pi\in\Ham(n)$
satisfying $\Pi^{2}=\Pi$.
\item $\Perm(n)$: The set of all permutations over $n$-bit strings.
\item Hamiltonians are denoted by capital $H$.\\
Random oracles (and permutations) are generally denoted by $O,O'$.\\
For the big-O notation, we use $\O(\cdot)$.
\item $\uni(n)$: The unitary group on $\left(\C^{2}\right)^{\otimes n}$.
\item $\hw(s)$: The Hamming weight of $s\in\binset^{*}$.
\item $\TD{\rho}{\sigma}:=\frac{1}{2}\norm{\rho-\sigma}_{1}$: The trace
distance between the states $\rho$ and $\sigma$. For pure states,
$\TD{\den a}{\den b}=\sqrt{1-\left|\braket ab\right|^{2}}$.
\item $\negl(\secp)$: A function that grows slower than $\frac{1}{p(\secp)}$
for every polynomial $p$, asymptotically.
\item By \emph{canonical circuit description}, we mean the collection of
gates and wires (qubits) the gate acts on, sequenced by the order
in which the gate is present in the circuit.
\item $\circdesc(U)$: A canonical description of  a circuit implementing
the unitary $U$, corresponding to a fixed constant sized universal
gate set $G$.
\item The pure states $\den{\psi},\den{\phi}$ are abbreviated as $\psi,\phi$
when no confusion arises. 
\item $X,Y,Z$ 
\begin{itemize}
\item These are used to represent the Pauli matrices, and by $X_{i}$ we
denotes the $X$ gate acting on the $i$-th qubit.
\item These may also be used to denote sets and registers. This should be
clear from the context.
\end{itemize}
\item PPT: Probabilistic Poly Time Turing Machine 
\item QPT: Quantum Poly-time Turing Machine (see \Cref{def:QPT})
\item For a set $S$, by $s\leftarrow S$ we mean $s$ is sampled uniformly
from $S$.
\item For $H=\sum_{i}\lambda_{i}\left|i\right\rangle \left\langle i\right|\in\Ham$,
$\left\Vert H\right\Vert =\left\Vert H\right\Vert _{\infty}:=\max_{i}\lambda_{i}$
and $\left\Vert H\right\Vert _{1}:=\sum_{i}|\lambda_{i}|$.
\item By $\left[n\right]$ we denote the set $\left\{ 1,2,\cdots,n\right\} $.
\item $\binom{n}{k}:=\frac{n!}{k!(n-k)!}.$
\end{itemize}

\subsubsection*{Hamiltonians}

\branchcolor{darkgray}{Most physically relevant Hamiltonians happen to be local in the following
sense.}
\begin{defn}[$k$-local Hamiltonian]
\label{def:k-local-ham}A Hamiltonian $H\in\Ham(n)$ is a\emph{ $k$-local
Hamiltonian }if it can be decomposed as follows:
\[
H=\sum^{m}_{i=1}H^{(i)}
\]
where each $H^{(i)}$ is a Hamiltonian on $n$-qubits that acts non-trivially
only on $k$-qubits, i.e. each $H^{(i)}$ corresponds to a subset
$S_{i}\subseteq\left[n\right]$ (see \Cref{nota:basic-notations})
such that $\left|S_{i}\right|\leq k$, such that
\[
H^{(i)}=h^{(i)}\otimes\I_{[n]\backslash S_{i}}
\]
for some operator $h^{(i)}\in\Ham\left(\left|S_{i}\right|\right)$.
\end{defn}

\begin{rem}
We say a family of $k$-local Hamiltonians $H_{n}\in\Ham(n)$ is simply
a \emph{local Hamiltonian} if $k$ does not depend on $n$.
\end{rem}

\branchcolor{darkgray}{An instructive instance of a local Hamiltonian is the Hamming-Weight
Hamiltonian that we use at various points in our exposition. Not only
does it serve as a simple running example throughout this article
but also models a physically relevant system---independent spins
in a fixed magnetic field.}
\begin{defn}[Hamming-Weight Hamiltonian]
\label{def:HammingWeight-Ham}Let $n$ denote the size of the system.
We define the Hamming weight Hamiltonian over $n$-qubits as follows,
\[
H_{\hw(n)}=\sum^{n}_{i\in1}\frac{(\mathbb{I}-Z_{i})}{2}=\sum_{s\in\{0,1\}^{n}}\hw(s)\left|s\right\rangle \left\langle s\right|
\]
where $\hw$ is as in \Cref{nota:basic-notations}. 
\end{defn}

\subsection{Ergotropy}

\branchcolor{darkgray}{The main quantity of interest in this work, is ergotropy. Here, we
recall how it is conventionally defined, before introducing its computational
variant.

Consider an $n$-qubit quantum system in a state $\rho$ with Hamiltonian
$H$. The extractable work is a process-dependent quantity. The work
extractable from any quantum system can be defined in multiple ways
depending upon the quantum operations used for work extraction. Ergotropy
captures the notion of maximum work extractable from a quantum system
under \emph{unitary} operation. Formally, it is defined as follows.}
\begin{defn}[Ergotropy]
\label{def:ergotropy}Given an $n$-qubit state $\rho\in\Den(n)$
and its associated Hamiltonian $H\in\Ham(n),$ we define \emph{ergotropy}
as
\[
\erg_{\rho,H}:=\max_{U\in\uni(n)}\tr\left[H\left(\rho-U\rho U^{\dagger}\right)\right]=\tr\left[H\rho\right]-\min_{U\in\uni(n)}\tr\left[HU\rho U^{\dagger}\right].
\]
\end{defn}

\branchcolor{darkgray}{To establish an unconditional separation between between the computational
notion of ergotropy (that we introduce later), and the notion above,
it is helpful to consider so-called passive states.

A state is said to be \emph{passive} with respect to a Hamiltonian
if no work can be extracted from it through unitary transformations.}
\begin{defn}[Passive state]
\label{def:passiveState} An $n$-qubit quantum state $\rho\in\Den(n)$
is passive with respect to $H\in\Ham(n)$ if $\erg_{\rho,H}=0.$ 
\end{defn}

\branchcolor{darkgray}{There is a well known characterisation of passive states that will
be helpful. One direction is immediate---the other was established
by \cite{Lenard:1978thm}.}
\begin{thm}[Characterisation of passive states \cite{Lenard:1978thm}]
\label{thm:char-passive-states} Let $H\in\Ham(n)$ have a spectral
decomposition $H=\sum_{i}e_{i}\den{e_{i}}$ with $e_{1}\leq e_{2}\leq\cdots,$
and let $\rho\in\Den(n)$. The state $\rho$ is passive (see \Cref{def:passiveState}
above) with respect to $H$ iff $\rho$ and $H$ are simultaneously
diagonalisable in some orthonormal basis $\left\{ \ket{v_{i}}\right\} _{i},$
such that $\rho=\sum_{i}\lambda_{i}\den{v_{i}}$ and $H=\sum_{i}\den{v_{i}}$,
where the eigenvalues are anti-ordered, i.e.
\[
e_{i}<e_{j}\implies\lambda_{i}\geq\lambda_{j}.
\]
\end{thm}

\branchcolor{darkgray}{We end by stating the following fact (that can also been as a simple
application of the theorem).}
\begin{fact}[Ergotropy of pure states]
\label{fact:erg-of-pure-states}For every $n$-qubit pure state $\ket{\psi}$
and every $H\in\Ham(n)$,
\[
\erg_{\psi,H}=\bra{\psi}H\ket{\psi}-E_{\min}(H),
\]
where $\psi=\den{\psi}$ (see \Cref{nota:basic-notations}) and $E_{\min}(H)$
is the minimum eigenvalue of the Hamiltonian $H$.
\end{fact}

\subsection{Haar measure}

\branchcolor{darkgray}{We use the Haar measure to establish an existential separation between
ergotropy and its computational variant. To this end, we recall the
relevant definition and the key results.}
\begin{defn}[Haar measure]
 \label{def:The-Haar}The \emph{Haar measure on the unitary group
$\uni(n)$} is the unique probability measure $\mu$ that is both
left and right invariant over the group $\uni(n)$, i.e. for all integrable
functions $f$ and for all $V\in\uni(n),$ we have
\[
\int_{\uni(n)}f(U)d\mu(U)=\int_{\uni(n)}f(UV)d\mu(U)=\int_{\uni(n)}f(VU)d\mu(U).
\]
By \emph{sampling an $n$-qubit pure state from the Haar measure},
we mean sampling a unitary from the Haar measure on the unitary group
$\uni(n)$, and applying it on $\ket 0^{\otimes n}.$
\end{defn}

\begin{lem}[Levy's lemma \cite{Ledoux2001}]
 \label{thm:Levy-lemma}Consider the set $\mathbb{S}^{2d-1}:=\left\{ v\in\mathbb{C}^{d}:\norm v_{2}=1\right\} .$
Let $f:\mathbb{S}^{2d-1}\to\mathbb{R}$ be a function satisfying the
Lipschitz condition $\left|f(v)-f(w)\right|\leq L\norm{v-w}_{2}.$
For all $\varepsilon\geq0,$ we have the probability bound:
\[
\underset{\ket{\phi}\sim\mu}{\Pr}\left[\left|f\left(\phi\right)-\underset{\ket{\phi}\sim\mu}{\mathbb{E}}\left[f\left(\psi\right)\right]\right|\geq\varepsilon\right]\leq2\exp\left(-\frac{2d\varepsilon^{2}}{9\pi^{3}L^{2}}\right).
\]
\end{lem}

\begin{cor}[Levy's lemma for expectation value of observables]
\label{lem:Levy-lemma-on-herm}Let $H$ be a Hermitian operator on
$\C^{d}$ and $\mu_{H}$ be a Haar measure on $\mathbb{C}^{d}$. For
all $\varepsilon\geq0$, we have
\[
\underset{\ket{\psi}\sim\mu}{\mathrm{Prob}}\left(\left|\bra{\psi}H\ket{\psi}-\frac{\tr(H)}{d}\right|\ge\varepsilon\right)\le2\exp\left(-\frac{d\varepsilon^{2}}{18\pi^{3}\norm H^{2}_{\infty}}\right).
\]
This is a straightforward consequence of Levy's lemma \cite{Ledoux2001}.
For a pedagogical exposition, see \cite{Mele_2024}.
\end{cor}

\subsection{Oracles}\label{subsec:Oracles}

\branchcolor{darkgray}{Many of our results are in the random oracle model. We review and
setup the notation for algorithms having oracle access, and state
the O2H lemma.}

\paragraph{Oracles Access}

The \emph{oracle} $O$ corresponding to a function $f:\{0,1\}^{n}\to\{0,1\}^{m}$
is defined as $O\left|x\right\rangle _{\mathsf{X}}\left|y\right\rangle _{\mathsf{Y}}=\left|x\right\rangle _{\mathsf{X}}\left|y\oplus f(x)\right\rangle _{\mathsf{Y}}$
where $\mathsf{X}$ and $\mathsf{Y}$ denote registers with Hilbert
space $\H_{\mathsf{X}}=(\C^{2})^{\otimes n}$ and $\H_{\mathsf{Y}}=(\C^{2})^{\otimes m}$
respectively. 
\begin{defn}[Quantum Oracle Algorithm]
 A \emph{$q$-query quantum oracle algorithm} $\A^{O}$ that acts
on an input register $\mathsf{I}$ and produces the output in register
$\mathsf{O}$, is a quantum circuit that has at most $q$-many $O$
gates. 
\end{defn}

\begin{rem}
\label{rem:A^f=00003DA^O}When no confusion arises, we use ${\cal A}^{f}$
instead of ${\cal A}^{O}$ to refer to an algorithm ${\cal A}$ having
oracle access to $f$. 
\end{rem}

\paragraph{The O2H Lemma}

\branchcolor{darkgray}{In the random oracle model, the algorithms are given access to a (uniformly)
random function $O:\binset^{n}\to\binset^{m}$ such that, for each
input $x\in\binset^{n}$, the output $O(x)$ is independently and
uniformly distributed in $\binset^{m}.$ In cryptographic contexts,
one often considers two oracles $O,O'$ that are drawn from some joint
distribution and differ only on a small domain. Classically, the only
way to distinguish these oracles is to query them on this domain.
The quantum analogue of this statement is the so-called One-way-to-Hiding
lemma---an abstraction of the well-known BBBV theorem \cite{Bennett_1997}.
The lemma also allows the algorithm to take as input, a string that
is arbitrarily correlated with the oracles.}
\begin{lem}[O2H Lemma (rephrased) \cite{OG_O2H_14,semiclassical_O2H_19}]
\label{lem:O2H} Let $S\subseteq\binset^{n}$ be a random subset,
$O,O':\binset^{n}\to\binset^{m}$ be random functions satisfying $\forall x\notin S,O(x)=O'(x)$.
Let $z$ be a random string in $\binset^{m}$. Note that $S,O,O'$
and $z$ may be drawn from a joint distribution. Let $\A$ be a quantum
oracle algorithm making at most $q$ queries. Let $\mathcal{B}^{O}$
be a quantum oracle algorithm with oracle access to $O$, which upon
input $z$ does the following: 
\begin{itemize}
\item Picks $i\leftarrow\{1,2,\cdots,q\}$
\item Runs $\A^{O}(z)$ until just before the $i$-th query, measures all
query input registers in the computational basis, outputs the set
$T$ of measurement outcomes.
\end{itemize}
Define
\begin{align*}
P_{\mathsf{left}} & :=\Pr\left[1\leftarrow\A^{O}(z)\right],\\
P_{\mathsf{right}} & :=\Pr\left[1\leftarrow\A^{O'}(z)\right],\\
P_{\mathsf{guess}} & :=\Pr\left[S\cap T\neq\phi:T\leftarrow\mathcal{B}^{O}(z)\right].
\end{align*}
Then, it holds that 
\[
\dfrac{\left|P_{\mathsf{left}}-P_{\mathsf{right}}\right|}{2q}\leq\sqrt{P_{\mathsf{guess}}}.
\]
\end{lem}

\branchcolor{darkgray}{We end this discussion by considering families of oracles and QPT
machines.}

\paragraph{QPT machines with oracle access.}

Let $n,m:\mathbb{N}\to\mathbb{N}$ be functions. Let $f_{\lambda}:\{0,1\}^{n(\lambda)}\to\{0,1\}^{m(\lambda)}$.
The family of oracles $\fat O:=\{O_{\lambda}\}_{\lambda}$ corresponding
to the family of functions $\{f_{\lambda}\}_{\lambda}$, is defined
as $O_{\lambda}\left|x\right\rangle _{\mathsf{X}}\left|y\right\rangle _{\mathsf{Y}}=\left|x\right\rangle _{\mathsf{X}}\left|y\oplus f(x)\right\rangle _{\mathsf{Y}}$
where $\mathsf{X}$ and $\mathsf{Y}$ denote registers with Hilbert
space $\H_{\mathsf{X}}=(\C^{2})^{\otimes n(\lambda)}$ and $\H_{\mathsf{Y}}=(\C^{2})^{\otimes m(\lambda)}$
respectively. 

An algorithm with oracle access to $\fat O$ is denoted by ${\cal A}^{\fat O}$.
By ${\cal A}^{\fat O}$ we mean that the algorithm can use $O_{\lambda}$
(for any choice of $\lambda$) as it would any other gate. 
\begin{defn}[QPT with oracle access]
\label{def:QPT}A \emph{quantum polynomial time Turing machine} $\left(\QPT\right)$
is specified by a probabilistic polynomial time\emph{ }Turing machine
$\left(\PPT\right)$ that on input $1^{\lambda}$, produces a canonical
description (see \Cref{nota:basic-notations}) of a quantum circuit
$C_{\lambda}$ that takes an $a(\lambda)$-qubit input and produces
a $b(\lambda)$-qubit output. 

On input $\left|\psi\right\rangle _{\mathsf{A}(\lambda)}$, $\QPT(1^{\secp},\left|\psi\right\rangle _{\mathsf{A}(\lambda)})$
runs $\PPT(1^{\lambda})$ to obtain $C_{\lambda}$ and returns $C_{\lambda}(\left|\psi\right\rangle _{\mathsf{A(\lambda)}})$
where $\mathsf{A}(\lambda)$ is an $a(\lambda)$-qubit input register.

Let $\fat O:=\{O_{\lambda'}\}_{\lambda'}$ be as described above.
Then, a \emph{quantum polynomial time Turing machine with oracle access}
to $\fat O$, $\QPT^{\fat O}$ is specified by a probabilistic polynomial
time machine $\PPT$ that on input $1^{\lambda}$, produces a canonical
description of a quantum circuit $C_{\lambda}$---that can use any
of $\{O_{\lambda'}\}_{\lambda'}$ as gates---and acts on an $a(\lambda)$
qubit input and produces a $b(\lambda)$ qubit output.

On input $\left|\psi\right\rangle _{\mathsf{A}(\lambda)}$, $\QPT^{\fat O}(1^{\secp},\left|\psi\right\rangle _{\mathsf{A}(\lambda)})$
runs $\PPT^{\fat O}(1^{\lambda})$ to obtain $C_{\lambda}$ and returns
$C_{\lambda}(\left|\psi\right\rangle _{\mathsf{A}(\lambda)})$ where
$\mathsf{A}(\lambda)$ is an $a(\lambda)$-qubit input register.
\end{defn}

\paragraph{Random Oracle Model.}

As alluded above, a random function $f:\{0,1\}^{*}\to\{0,1\}$ is
one that, on each input, outputs an independent uniformly sampled
bit. In the \emph{Random Oracle Model}, each party/algorithm, is given
oracle access to $f$. 

It is clear that $f$ outputting a bit is not a limitation---it can
be used to produce arbitrary length outputs (e.g. by using $f(i|\cdot)$
where $i$ encodes the $i$-th bit of the desired output). 

In our constructions, for instance, we consider families of random
functions of the form $f_{\lambda}:\{0,1\}^{n(\lambda)}\to\{0,1\}^{m(\lambda)}$
and access to this is understood in the sense described in \Cref{def:QPT}
above.

\subsection{Cryptography}

\branchcolor{darkgray}{For our first constructive separation in the plain model, we rely
on pseudorandom functions.}
\begin{defn}[PRF]
\label{def:PRF}A Pseudorandom Function Family (PRF) is defined by
a classical deterministic polynomial-time algorithm as follows: 
\begin{itemize}
\item $\feval\left(\key,x\right)$ takes as input a key $k\in\binset^{n}$
and an input $x\in\binset^{n}$. It returns $y\in\binset^{n}$.
\end{itemize}
We note that here we chose the key and input lengths to be equal (and
implicitly both are also equal to the security parameter $\secp$)
for the sake of minimising the number of parameters. We require the
following property.
\end{defn}

\begin{itemize}
\item Security: For any polynomial-time quantum algorithm $\mathcal{A}$
with binary output, it holds that
\[
\left|\Pr_{k\gets\binset^{n}}\left[1\leftarrow\A^{\feval\left(\key,\cdot\right)}\left(1^{n}\right)\right]-\Pr_{f\leftarrow\mathcal{F}_{n}}\left[1\leftarrow\A^{f}\left(1^{n}\right)\right]\right|=\negl\left(n\right).
\]
Here, ${\cal F}_{n}$ denotes the set of all possible functions mapping
$\binset^{n}$ to $\binset^{n}$.The superscript notation $\A^{\feval\left(\key,\cdot\right)}$
means that the (quantum) algorithm $\A$ is allowed to make arbitrary
(quantum) queries to evaluate the functions $\feval\left(\key,\cdot\right)$.
\end{itemize}
\branchcolor{darkgray}{In the construction of pseudo-ergotropy (see \Cref{sec:Pseudoergotropy})
we use pseudorandom permutations. This in turn, also gives us other
plain-model separations (e.g. between ergotropy and its computational
variant, see \Cref{thm:explicit-sep-erg-erghat-prf}). The definition
below and the following result have been taken verbatim from \cite{arnon2023CET}.}
\begin{defn}[PRP]
\label{def:PRP}A \emph{Pseudorandom Permutation Family (PRP)} is
defined by a pair of classical deterministic polynomial time algorithms
as follows.
\end{defn}

\begin{itemize}
\item $\feval\left(\key,x\right)$ takes as input a key $k\in\binset^{n}$
and an input $x\in\binset^{n}$. It returns $y\in\binset^{n}$.
\item $\finv\left(\key,y\right)$ takes as input a key $k\in\binset^{n}$
and an input $y\in\binset^{n}$. It returns $x\in\binset^{n}$.
\end{itemize}
We note that here we chose the key and input lengths to be equal (and
implicitly both are also equal to the security parameter $\secp$)
for the sake of minimising the number of parameters. We require the
following properties.
\begin{enumerate}
\item Correctness: For all $n$, for all $\key\in\binset^{n}$, and for
all $x,y\in\binset^{n}$, it holds that $\feval\left(\key,x\right)=y$
if and only if $\finv\left(\key,y\right)=x$
\item Pseudorandomness:\textbf{ }For any polynomial time quantum algorithm
$\A$ with binary output it holds that
\[
\left|\Pr_{k\gets\binset^{n}}\left[1\leftarrow\A^{\feval\left(\key,\cdot\right),\finv\left(\key,\cdot\right)}\left(1^{n}\right)\right]-\Pr_{\pi\gets S_{\binset^{n}}}\left[1\leftarrow\A^{\pi,\pi^{-1}}\left(1^{n}\right)\right]\right|=\negl\left(n\right).
\]
Here, the superscript notation $\A^{\feval\left(\key,\cdot\right),\finv\left(\key,\cdot\right)}$
means that the (quantum) algorithm $\A$ is allowed to make arbitrary
(quantum) queries to evaluate the functions $\feval\left(\key,\cdot\right)$
and $\finv\left(\key,\cdot\right)$.
\end{enumerate}
\begin{lem}[Permutation $\approx$ Expansion $\circ$ Contraction; \cite{arnon2023CET}]
\label{lem:fx_equals_ghx}Consider a PRP family $(\feval,\finv)$
and let $n,m\in\mathbb{N}$ be s.t. $m>n$. We consider $\key_{f},\key_{g},\key_{h}\in\binset^{m}$,
and define the following functions: $f:\binset^{m}\to\binset^{m}$
as $f(x)=\feval\left(\key_{f},x\right)$, $g:\binset^{n}\to\binset^{m}$
as $g(x)=\feval\left(\key_{g},x\|0^{m-n}\right)$, and $h:\binset^{m}\to\binset^{n}$
as $h(x)=\feval\left(\key_{h},x\right)_{|n}$ (i.e. the $n$-bit prefix
of the outcome). Then for any polynomial function $m=m\left(n\right)$
and for any polynomial-time quantum algorithm $\A$ with binary output
it holds that 
\[
\left|\Pr_{\key_{g},\key_{h}}\left[1\leftarrow\A^{g(h(\cdot))}\left(1^{n}\right)\right]-\Pr_{k_{f}}\left[1\leftarrow\A^{f(\cdot)}\left(1^{n}\right)\right]\right|=\negl\left(n\right).
\]
\end{lem}

For completeness, we include the proof in \Cref{sec:defElementary}. 

\section{Computational work extraction}\label{sec:comp-work-ex}

\branchcolor{darkgray}{Recall that Ergotropy (see \Cref{def:ergotropy}) is defined to be
the maximum amount of work that can be extracted from a state (given
the corresponding Hamiltonian) using a \emph{unitary} operation. However,
there are no restrictions on which unitary one uses---in this sense,
it is an information theoretic quantity. Clearly, unitary operations
that extract the maximum work, may not be efficiently realisable.
Therefore, it is physically more realistic to revise the notion of
work extraction to account for computational boundedness. In this
section, we define such a computational notion of work extraction---computational
ergotropy (see \Cref{def:computational-erg})---and show that indeed,
this notion is distinct from \Cref{def:ergotropy}. More concretely,
we show an unconditional existential separation between ergotropy
and computational ergotropy, followed by a concrete separation in
the random oracle model. \Cref{subsec:The-access-model} introduces
computational ergotropy while \Cref{subsec:erg-vs-erghat} establishes
the separations.}

\branchcolor{darkgray}{There are a few challenges in defining a computational notion of ergotropy.
We start by first looking at why one must consider families of states
as opposed to studying a single state (and their corresponding Hamiltonians).}

\subsection{Defining computational ergotropy $\protect\erghat$}\label{subsec:The-access-model}

\branchcolor{darkgray}{Ergotropy (as in \Cref{def:ergotropy}) considers a given state-Hamiltonian
pair. However, the notion of computational efficiency is only defined
asymptotically---one cannot say whether a circuit with $500$ gates
is efficient or not in isolation. Computational efficiency is therefore
defined in terms of, for instance, the number of gates in a circuit
as a function of the input size. Typically, when one uses the word
``efficient'' one thinks of quantities such as number of gates,
being polynomial in the input size.

Thus, to restrict to computationally efficient unitaries in ergotropy,
one must first reformulate ergotropy in the asymptotic setting. In
the following discussion, we start by parametrising the system using
a security parameter $\lambda\in\mathbb{N}$ and replace single instances
of state-Hamiltonian pairs, with families $\{\rho_{\secp}\}$ and
$\{H_{\secp}\}$. Since we are no longer in the information theoretic
setting, one needs to be more careful about defining how these objects
are accessed by an algorithm.}

Let $n:\N\to\N$ be a function that denotes the size (number of qubits)
of the quantum system. For all $\secp\geq1,$ let $\fat{\rho}=\left\{ \rho_{\secp}\right\} {}_{\secp}$
and $\fat H=\left\{ H_{\secp}\right\} _{\secp}$ denote a family of
states and Hamiltonians, such that $\rho_{\secp}\in\Den\left(n(\secp)\right),H_{\secp}\in\Ham\left(n(\secp)\right)$. 
\begin{defn}[State generation oracle]
\label{def:state-gen}The \emph{state generation oracle} $O_{\fat{\rho}}$,
upon being queried at $1^{\secp}$, outputs the state $n(\secp)$-qubit
state $\rho_{\secp}\in\fat{\rho}$, i.e.
\[
O_{\fat{\rho}}\left(\den{1^{\secp}}\right)=\rho_{\secp}.
\]
\end{defn}

\begin{defn}[Efficiently describable Hamiltonian]
\label{def:H-efficient}
A family of Hamiltonians $\fat H$ is \emph{efficiently
describable} if there is a PPT machine (see \Cref{nota:basic-notations})
$\hamdesc$ that on input $1^{\secp}$ outputs a canonical circuit
description of $H_{\lambda}$. We say $\fat H$ is \emph{extensive}
if $\left\Vert H_{\lambda}\right\Vert _{\infty}\in{\cal O}(n(\lambda))$.

\end{defn}

\begin{example}
A $k$-local Hamiltonian (see \Cref{def:k-local-ham}) is efficiently
describable (implicit in the proof of \Cref{lem:trHrho=00003DtrHsigma}
that appears later). 
\end{example}

\branchcolor{darkgray}{In this work, we restrict ourselves to studying only normalised Hamiltonians.
This is because by using arbitrary Hamiltonians, one can scale the
separation between information theoretic ergotropy and computational
ergotropy (see \Cref{subsec:erg-vs-erghat}) arbitrarily.}
\begin{defn}[Work extracted by an algorithm]
\label{def:work-ex-algo} Let $\fat{\rho},\fat H$ be as above, and
let $\A$ be a quantum oracle algorithm. We define the work extracted
by the algorithm $\A$ from the state $\rho_{\secp}$ with respect
to the Hamiltonian $H_{\secp}$ as
\[
\W_{\fat{\rho},\fat H,\A}(\secp):=\mathbb{\E}\left[\tr\left[H_{\secp}\left(\rho_{\secp}-U_{\secp}\rho_{\secp}U^{\dagger}_{\secp}\right)\right]:\begin{array}{c}
\circdesc(U_{\secp})\leftarrow\A^{O_{\fat{\rho}},\hamdesc}(1^{\secp}),\\
U_{\secp}\in\uni(n(\secp))
\end{array}\right].
\]
\end{defn}

\branchcolor{darkgray}{Since ergotropy captures the maximum work that can be extracted from
a state with respect to a Hamiltonian through a unitary operation,
one can define computational ergotropy as the maximum of $\W_{\fat{\rho},\fat H,\A}$
over all algorithms $\A$ of a particular computational class. Since
$\W_{\fat{\rho},\fat H,\A}$ is a function, $\W_{\fat{\rho},\fat H,\A}:\N\to\mathbb{R},$
we need a notion of comparison of these functions to meaningfully
define the maximum of $\W_{\fat{\rho},\fat H,\A}$ over all algorithms
$\A.$}
\begin{notation}[Function Comparison]
\label{nota:function-comparison} Let $f,g:\N\to\mathbb{R}$ be functions.
By $f\fnleq g$, we mean that there is a $\secp_{0}\in\N$ such that
for all $\secp\geq\secp_{0}$, it holds that $f(\secp)\leq g(\secp)$.
We say $f\fneq g$ if $f\fnleq g$ \emph{and} $f\fngeq g$ simultaneously. 

\branchcolor{darkgray}{With all these notions in place, we can finally formally state what
we mean by computational ergotropy.}
\end{notation}

\begin{defn}[Bounds on Computational Ergotropy $(\cerg,\erghat)$]
\label{def:computational-erg} Let ${\cal C}$ be a class of Turing
machines (e.g. QPT, PPT), $\ub,\lb:\N\to\mathbb{R}$ be functions,
and $\fat{\rho},\fat H$ be families of states and Hamiltonians as
above. We say that the \emph{computational ergotropy for the class
${\cal C}$}, denoted by $\cerg_{\fat{\rho},\fat H}$, is $(\lb,\ub)$
bounded if 
\begin{align}
\forall\ \A\in\mathcal{C},\text{ it holds that } & \W_{\fat{\rho},\fat H,\A}\fnleq\ub,\text{ and }\label{eq:ub-C-erg}\\
\exists\ \A\in\mathcal{C}\text{ s.t. } & \W_{\fat{\rho},\fat H,\A}\fngeq\lb.\label{eq:lb-C-erg}
\end{align}
When ${\cal C}$ is taken to be the class of QPT (or PPT) algorithms
(see \Cref{def:QPT}), we say \emph{computational ergotropy}, denoted
by $\erghat_{\fat{\rho},\fat H}$, is $(\lb,\ub)$ bounded. 
\end{defn}

\branchcolor{darkgray}{In subsequent discussions, we give upper bounds for computational
ergotropy. To express them concisely, we overload \Cref{nota:function-comparison}
as follows.}
\begin{notation}
\label{nota:erg-fnleq-et-al}We use $\cerg_{\fatrhoH}\fnleq\ub$ to
denote $\cerg$ is upper bounded by $\ub$ (i.e. \Cref{eq:ub-C-erg}
holds) and $\cerg_{\fatrhoH}\fngeq\lb$ to denote $\cerg$ is lower
bounded by $\lb$ (i.e. \Cref{eq:lb-C-erg}). 

\branchcolor{darkgray}{While we defined $\cerg$ in terms of upper and lower bounds, in the
following, we explain how one might intuitively think about $\cerg$
when these bounds coincide.}
\end{notation}

\begin{rem}[Matching bounds to specify $\cerg$]
\label{rem:Condition-well-defined-Cerg} For cases, such as when
$\fat{\rho}$ is just the maximally mixed state, it is intuitively
clear that $\cerg$ must be $0$. Formally, this corresponds to cases
where $\lb$ and $\ub$ match. 
\begin{enumerate}
\item We say $\cerg_{\fat{\rho},\fat H}$ is represented by a function $\mathsf{b}:\mathbb{N}\to\mathbb{R}$
if $\cerg$ is $(\mathsf{b},\mathsf{b})$ bounded.
\item Note that if $\cerg$ is represented both by functions $\mathsf{b}$
and $\mathsf{b'}$ then $\mathsf{b}=_{\infty}\mathsf{b'}$.
\item When $\cerg$ is represented by a function $\mathsf{b}$, then there
is an algorithm ${\cal A}^{*}\in{\cal C}$ such that 
\[
\W_{\fat{\rho},\fat H,{\cal A}^{*}}\fneq\mathsf{b}.
\]
\item Another way of viewing ${\cal A}^{*}$ is as the best algorithm for
extracting work, asymptotically, i.e. 
\[
{\cal A}^{*}\in\argmax_{{\cal A}\in{\cal C}}\W_{\fat{\rho},\fat H,{\cal A}}
\]
 where we say an algorithm ${\cal A}'\in\argmax_{{\cal A}\in{\cal C}}\W_{\fat{\rho},\fat H,{\cal A}}$
if for all ${\cal A}\in{\cal C}$, $\W_{\fat{\rho},\fat H,{\cal A}}\fnleq\W_{\fat{\rho},\fat H,{\cal A}'}$.
\end{enumerate}
\end{rem}

\branchcolor{darkgray}{To study the separation between ergotropy and its computational variant,
it also helps to define ergotropy for families of states and Hamiltonians.
}
\begin{notation}
Let $\fat{\rho}=\left\{ \rho_{\secp}\right\} _{\secp}$ be a family
of states, and $\fat H=\left\{ H_{\secp}\right\} _{\secp}$ be a family
of Hamiltonians. By $\erg_{\fat{\rho},\fat H}(\secp)$ we denote the
quantity $\erg_{\rho_{\secp},H_{\secp}}$, the ergotropy for each
state-Hamiltonian pair in the family.
\end{notation}

\branchcolor{darkgray}{With $\erg_{\fat{\rho,\fat H}}(\secp)$ now defined, we show that
by taking ${\cal {\cal C}}$ to be the class of unbounded computations,
$\cerg$ becomes essentially equivalent to $\erg_{\fat{\rho},\fat H}$.}
\begin{rem}
\label{rem:C-erg=00003Derg_when_C_is_unbounded}Let $\epsilon:\N\to\mathbb{R}$
be any function such that with $\epsilon(\lambda)>0$ for all $\lambda\in\mathbb{N}$.
Then, we have the following.
\begin{enumerate}
\item There exists a (unbounded) quantum Turing machine $\A$ such that
\begin{align*}
\erg_{\fat{\rho},\fat H}(\secp)-\W_{\fat{\rho},\fat H,\A}(\secp) & \geq\epsilon(\secp) & \text{for all }\secp\in\N.
\end{align*}
\item It holds that $\cerg$ is $(\erg_{\fat{\rho},\fat H}-\epsilon,\erg_{\fat{\rho},\fat H})$
bounded, when ${\cal C}$ is the class of all quantum Turing machines
that halt on all inputs. 
\end{enumerate}
\branchcolor{darkgray}{The first point in \Cref{rem:C-erg=00003Derg_when_C_is_unbounded}
is immediate: by querying the state generation oracle $O_{\fat{\rho}}$
(see \Cref{def:state-gen}) exponentially many times, one can do state
tomography, and obtain the description of the state up to arbitrary
precision. With the description of the state and Hamiltonian, one
can describe a unitary that maps this state to a passive state (as
characterised in \Cref{thm:char-passive-states}) therefore extracting
maximal work.

The second point in \Cref{rem:C-erg=00003Derg_when_C_is_unbounded}
is an immediate consequence of the first point: For each $\lambda$,
$\erg_{\fat{\rho},\fat H}(\lambda)$ is already the information theoretically
maximal amount of work one can extract---so even with unbounded computation,
one cannot exceed this value; and point 1 shows that there is an algorithm
${\cal A}$ that extracts at least $\erg_{\fat{\rho},\fat H}(\lambda)-\epsilon(\lambda)$
amount of energy.}

\end{rem}

\subsection{Unconditional existential separation between $\protect\erg$ and
$\protect\erghat$}\label{subsec:erg-vs-erghat}

\branchcolor{darkgray}{Given that $\erghat_{\fatrhoH}$ is most $\erg_{\fatrhoH}$, (i.e.
$\erg_{\fat{\rho},\fat H}\fngeq\erghat_{\fat{\rho},\fat H}$ using
\Cref{nota:erg-fnleq-et-al} and \Cref{rem:C-erg=00003Derg_when_C_is_unbounded}
above), one can ask the following natural question. 
\begin{quote}
\emph{What is the maximal separation possible $\erg_{\fatrhoH}$ and
$\erghat_{\fatrhoH}$?}
\end{quote}
We show a near-maximal separation between ergotropy and its computational
variant is possible. \Cref{subsec:Unconditional-Existential-Separa}
establishes an unconditional existential separation and \Cref{subsec:A-constructive-separation}
gives an explicit construction in the random oracle model (see last
paragraph in \Cref{subsec:Oracles}). We note that our proof techniques
generalise to most physically relevant extensive Hamiltonians (see
\Cref{def:H-efficient}).}

\subsubsection{Unconditional Existential Separation}\label{subsec:Unconditional-Existential-Separa}

\branchcolor{darkgray}{Here, we show that there exists a family of states and Hamiltonians
$(\fatrhoH)$ such that $\erghat_{\fatrhoH}\fnleq\negl$ for some
negligible function, while $\erg_{\fatrhoH}\fngeq\poly-\negl'$ for
$\poly(n)=n/2$ and some other negligible function $\negl'$.

This result is established using a counting argument. To this end,
we first upper bound the number of circuits, in terms of the number
of gates and wires it uses.}

For \Cref{lem:circuit-count,lem:bounding-erg-Up} below, let $G$
be a (finite) set of universal gates  consisting of one and two-qubit
gates. We assume $\circdesc$ (see \Cref{nota:basic-notations}) is
defined relative to $G$. 
\begin{lem}[Circuit Count]
\label{lem:circuit-count}Let $p:\N\to\N$ be a function. The number
of unitary operations that can be implemented on $n$-qubits with
at most $p(n)$ gates from $G$ is at most $\O(n^{p(n)}).$ 
\end{lem}

\branchcolor{black}{\begin{proof}[Proof sketch.]
A circuit of size at most $p(n)$ gates, is a sequence of $p(n)$
slots, where for each slot, we choose 2 out of the $n$-qubits the
gate acts on ($\O(n^{2})$), and specific gate in $G$ ($\left|G\right|=\O(1)$).
Therefore, there are at most $\O(n^{2}.\left|G\right|)^{p(n)}=\O(n^{p(n)})$
circuits. Some of these circuits might correspond to the same unitary,
therefore, $\left|\uni_{p}(n)\right|\leq\O(n^{p(n)}).$
\end{proof}

}

In \Cref{lem:bounding-erg-Up} and \Cref{thm:unconditional-existential-separation}
we consider the $n$-qubit Hamiltonian $H_{n}=\sum^{n}_{i=1}X_{i}$
(see \Cref{nota:basic-notations}). \branchcolor{darkgray}{For the lemma, it would also be useful to consider a variant of ergotropy
where the unitary operations are restricted to having at most $p$
many gates.} We define 
\[
\erg^{\mathcal{U}_{p}(n)}_{\rho_{n},H_{n}}:=\tr[H_{n}\rho_{n}]-\min_{U_{n}\in\mathcal{U}_{p}(n)}\tr[H_{n}U_{n}\rho_{n}U^{\dagger}_{n}]
\]
where $U\in{\cal U}_{p}(n)$ iff $\circdesc(U)$ has at most $p(n)$
gates.

\branchcolor{darkgray}{The following lemma shows that there exists a state for which $\erg$
is essentially maximal while $\erg^{{\cal U}_{p}}$ is negligible.
We first show that for Haar random pure states, $\erg$ is maximal
with high probability. Then using the unitary invariance of Haar measures
and union bound, we argue that, by maximising only over unitaries
in $\uni_{p}$, one can't extract non-negligible amount of work from
a Haar random pure states.}
\begin{lem}
\label{lem:bounding-erg-Up}Let $p:\mathbb{N}\to\mathbb{N}$ be the
function $p(n)=2^{n/4}$. For the $n$-qubit Hamiltonian $H_{n}$
and any choice of $\delta(n)=2^{-n/3}$, there exists $n_{0}=n_{0}(\delta,p,G)$
such that for all $n\geq n_{0}$ there exists an $n$-qubit pure state
$\left|\psi_{n}\right\rangle $ satisfying
\begin{align*}
\erg_{\psi_{n},H_{n}} & \geq(1-\delta(n))\cdot n,\text{ and }\\
\erg^{\mathcal{U}_{p}(n)}_{\psi_{n},H_{n}} & <2\delta(n)\cdot n.
\end{align*}
 
\end{lem}

\branchcolor{black}{\begin{proof}
We break this proof into the following two claims.
\begin{claim}
\label{claim:lower-bound-erg-for-haar}For a Haar random state over
$n$-qubits (\Cref{def:The-Haar}), the information theoretic ergotropy
is at least $n(1-\delta)$ with high probability, i.e. 
\[
\Pr_{\phi_{n}\sim\mu}\left[\erg_{\psi_{n},H_{n}}<\left(1-\delta(n)\right)n\right]\leq2\exp\left(-\dfrac{2^{n}\delta^{2}(n)}{18\pi^{3}}\right).
\]
\end{claim}

\branchcolor{black}{\begin{proof}
Let $\phi_{n}$ be an $n$-qubit Haar random pure state, i.e. $\phi_{n}\sim\mu$.
Then for any $\varepsilon>0$, it holds that,
\begin{equation}
\Pr_{\phi_{n}\sim\mu}\left[\left|\bra{\phi_{n}}H_{n}\ket{\phi_{n}}\right|\geq\delta(n)\cdot n\right]\leq2\exp\left(-\dfrac{2^{n}\delta^{2}(n)}{18\pi^{3}}\right).\label{eq:conc-psin-hn}
\end{equation}
This follows from directly from \Cref{lem:Levy-lemma-on-herm}, since
for the Hamiltonian $H_{n},$ $\tr\left[H_{n}\right]=0$ and $-E_{\min}(H_{n})=n=E_{\max}(H_{n})=\norm{H_{n}}_{\infty}$.

From \Cref{fact:erg-of-pure-states} we know that the ergotropy of
a pure state $\phi$ with respect to the Hamiltonian $H$ can be written
as,
\[
\erg_{\phi,H}=\bra{\phi}H\ket{\phi}-E_{\min}(H).
\]
Combining this with \Cref{eq:conc-psin-hn}, we get, 
\begin{align*}
\Pr_{\phi_{n}\sim\mu}\left[\erg_{\psi_{n},H_{n}}<\left(1-\delta(n)\right)n\right] & =\Pr_{\phi_{n}\sim\mu}\left[\bra{\phi_{n}}H_{n}\ket{\phi_{n}}+n<\left(1-\delta(n)\right)n\right]\\
 & =\Pr_{\phi_{n}\sim\mu}\left[\bra{\phi_{n}}H_{n}\ket{\phi_{n}}<-\delta(n)\cdot n\right]\\
 & \leq\Pr_{\phi_{n}\sim\mu}\left[\left|\bra{\phi_{n}}H_{n}\ket{\phi_{n}}\right|\geq\delta(n)\cdot n\right]\\
 & \leq2\exp\left(-\dfrac{2^{n}\delta^{2}(n)}{18\pi^{3}}\right).
\end{align*}
\end{proof}

}
\begin{claim}[$\erg^{\mathcal{U}_{p}(n)}_{\phi'_{n},H_{n}}<2\delta(n)\cdot n$ w.h.p.]
\label{claim:upperbound-ergUp-for-haar}For a Haar random state over
$n$-qubits, the ergotropy in the setting where the circuit size is
at most $p(n)$, is upper bounded by $2\delta(n)\cdot n$, with high
probability, i.e. there exists a negligible function $\negl_{2}$
such that 
\[
\Pr_{\phi'_{n}\sim\mu}\left[\erg^{\mathcal{U}_{p}(n)}_{\phi'_{n},H_{n}}>2\delta(n)\cdot n\right]\leq\negl_{2}(n).
\]
\end{claim}

\branchcolor{black}{\begin{proof}
From \Cref{eq:conc-psin-hn} we know that 
\[
\Pr_{\phi'_{n}\sim\mu}\left[\left|\bra{\phi'_{n}}H_{n}\ket{\phi'_{n}}\right|\geq\delta(n)\cdot n\right]\leq2\exp\left(-\dfrac{2^{n}\delta^{2}(n)}{18\pi^{3}}\right).
\]
Since Haar unitaries are left and right unitary invariant, for any
fixed $n$-qubit unitary $V$, the following holds,
\begin{align*}
\Pr_{\phi'_{n}\sim\mu}\left[\left|\bra{\psi}V^{\dagger}HV\ket{\psi}\right|\geq\delta(n)\cdot n\right] & \leq2\exp\left(-\dfrac{2^{n}\delta^{2}(n)}{18\pi^{3}}\right)\\
\implies\Pr_{\phi'_{n}\sim\mu}\left[\exists V\in\mathcal{U}_{p}(n)\text{ s.t. }\left|\bra{\psi}V^{\dagger}HV\ket{\psi}\right|\geq\delta(n)\cdot n\right] & \leq2\exp\left(-\dfrac{2^{n}\delta^{2}(n)}{18\pi^{3}}\right)\cdot\left|\mathcal{U}_{p}(n)\right| & \text{through union bound}\\
\implies\Pr_{\phi'_{n}\sim\mu}\left[\exists V\in\mathcal{U}_{p}(n)\text{ s.t. }\left|\bra{\psi}V^{\dagger}HV\ket{\psi}\right|\geq\delta(n)\cdot n\right] & \leq2\exp\left(-\dfrac{2^{n}\delta^{2}(n)}{18\pi^{3}}\right)\cdot\O(n^{p(n)}) & \text{}\\
 & =\mathcal{O}\left(n^{p(n)}/2^{2^{n}\delta^{2}(n)}\right) & \text{for }p(n)=2^{n/4}\text{ and }\\
 & =\negl\left(n\right) & \delta^{2}(n)=2^{-n/3}.
\end{align*}
We get the third inequality because the size of $\mathcal{U}_{p}(n)\leq\O(n^{p(n)})$
(from \Cref{lem:circuit-count}).

We now bound the ergotropy optimised over the set $\uni_{p}(n)$ as,

\begin{align*}
\erg^{\mathcal{U}_{p}(n)}_{\phi'_{n},H_{n}} & =\bra{\phi'_{n}}H_{n}\ket{\phi'_{n}}-\min_{V_{n}\in\mathcal{U}_{p}(n)}\bra{\phi'_{n}}V_{n}H_{n}V^{\dagger}_{n}\ket{\phi'_{n}}\\
 & \Downarrow\\
\Pr_{\phi'_{n}\sim\mu}\left[\erg^{\mathcal{U}_{p}(n)}_{\phi_{n},H_{n}}>2\delta(n)\cdot n\right] & \leq1-\left(1-\Pr_{\phi_{n}\sim\mu}\left[\left|\bra{\phi'_{n}}H_{n}\ket{\phi'_{n}}\right|\geq\delta(n)\cdot n\right]\right)\\
 & \quad\quad\times\left(1-\Pr_{\phi_{n}\sim\mu}\left[\exists V\in\mathcal{U}_{p}\ s.t.\left|\bra{\phi'_{n}}V_{n}H_{n}V^{\dagger}_{n}\ket{\phi'_{n}}\right|\geq\delta(n)\cdot n\right]\right)\\
 & \Downarrow\\
\Pr_{\phi'_{n}\sim\mu}\left[\erg^{\mathcal{U}_{p}(n)}_{\phi_{n},H_{n}}>2\delta(n)\cdot n\right] & \leq\Pr_{\phi_{n}\sim\mu}\left[\left|\bra{\phi'_{n}}H_{n}\ket{\phi'_{n}}\right|\geq\delta(n)\cdot n\right]\\ & \qquad \qquad +\Pr_{\phi_{n}\sim\mu}\left[\exists V\in\mathcal{U}_{p}\text{ s.t. }\left|\bra{\phi'_{n}}V_{n}H_{n}V^{\dagger}_{n}\ket{\phi'_{n}}\right|\geq\delta(n)\cdot n\right]\\
 & \leq2\exp\left(-\dfrac{2^{n}\delta^{2}(n)}{18\pi^{3}}\right)+\mathcal{O}\left(2^{n}/2^{2^{n}\delta^{2}(n)}\right)=\negl\left(n\right).
\end{align*}
\end{proof}

}

Putting together \Cref{claim:lower-bound-erg-for-haar,claim:upperbound-ergUp-for-haar}
we get the following:
\begin{enumerate}
\item For a $1-\negl_{1}(n)$ fraction of states, $\erg_{\phi_{n},H_{n}}\geq\left(1-\delta(n)\right)\cdot n$,
for some negligible function $\negl_{1}$
\item For a $1-\negl_{2}(n)$ fraction of states, $\erg^{\mathcal{U}_{p}(n)}_{\phi'_{n},H_{n}}\leq2\delta(n)\cdot n$
where $\negl_{2}$ is as in \Cref{claim:upperbound-ergUp-for-haar}
\end{enumerate}
Therefore there exists at least one $n$-qubit pure state $\ket{\psi_{n}}$
for which $\erg_{\psi_{n},H_{n}}\geq\left(1-\delta(n)\right)\cdot n$
and $\erg^{\mathcal{U}_{p}}_{\psi_{n},H_{n}}\leq2\delta(n)\cdot n$.
\end{proof}

}

\branchcolor{darkgray}{It is immediate that the proof generalises to any extensive Hamiltonian
(i.e. Hamiltonians with polynomial operator norm). With this proof
in place, we are now ready to state the main result of this section.}
\begin{thm}[Unconditional Existential Separation]
\label{thm:unconditional-existential-separation}Let $\mathsf{poly}(n):=(1-o(1))n$
and $\fat H=\{H_{n}\}_{n}$. Then, there exists a family of pure states
$\fat{\rho}=\left\{ \den{\psi_{n}}\right\} _{n}$ and a negligible
function $\negl$ such that 
\begin{align*}
\erg_{\fat{\rho},\fat H} & \geq_{\infty}\mathsf{poly}\\
\erghat_{\fat{\rho},\fat H} & \leq_{\infty}\negl
\end{align*}
where we are using \Cref{nota:function-comparison,nota:erg-fnleq-et-al}.
\end{thm}

\branchcolor{black}{\begin{proof}
From \Cref{lem:bounding-erg-Up} we know that there exists a family
of states $\left\{ \psi_{n}\right\} _{n}$ such that $\erg_{\fat{\rho},\fat H}(n)\geq(1-\delta(n))n$
and $\erg^{\mathcal{U}_{p}(n)}_{\psi_{n},H_{n}}<2\delta(n)\cdot n$
for $p\in\O(2^{n/4})$. Since the output of any QPT algorithm is at
most polynomially sized, for each QPT algorithm $\A,$ there exists
an $n_{\A}\in\N$ such that for all $n\geq n_{\A},$the circuit whose
description is output by $\A$ is in $\uni_{p}(n)$. Therefore, $\erghat_{\fat{\rho},\fat H}\leq_{\infty}\erg^{{\cal U}_{p}}_{\fatrhoH}\leq_{\infty}\negl$
(using \Cref{nota:function-comparison,nota:erg-fnleq-et-al}) where
$\erg^{{\cal U}_{p}}_{\fatrhoH}(n):=\erg^{\mathcal{U}_{p}(n)}_{\psi_{n},H_{n}}$.
Therefore, we have $\erghat_{\fat{\rho},\fat H}\leq_{\infty}\negl$
while $\erg_{\fat{\rho},\fat H}(n)\geq(1-\O(2^{-n/2}))\cdot n=(1-o(1))\cdot n$
for all $n\in\N$.
\end{proof}

}

\subsection{Defining computational ergotropy for distributions $\protect\erghat_{{\cal D}}$}

\branchcolor{darkgray}{For the explicit separation, we first introduce ergotropy for distributions. }

Let ${\cal D}$ denote a probability distribution over families of
states $(\fat{\rho}_{i})_{i}$ and families of Hamiltonian $(\fat H_{i})_{i}$.
We use $\fatrhoH\leftarrow{\cal D}$ to denote the result of sampling
a pair of state and Hamiltonian families from ${\cal D}$.
\begin{defn}[Bounds on Computational Ergotropy for a distribution ${\cal D}$ ($\cerg_{{\cal D}},\erghat_{{\cal D}}$)]
\label{def:computational-erg-D} Let ${\cal C}$ be a class of Turing
machines (e.g. QPT, PPT), $\ub,\lb:\mathbb{N}\to\mathbb{R}$ be functions,
and ${\cal D}$ be a distribution over families of states and Hamiltonians
(as described above). We say that the \emph{computational ergotropy
for the class ${\cal C}$ and distribution ${\cal D}$}, denoted by
$\cerg_{{\cal D}}$, is $(\lb,\ub)$ bounded if 
\begin{align*}
\forall\quad{\cal A}\in{\cal C},\text{ it holds that }\E_{\fatrhoH\leftarrow{\cal D}}\left[{\cal W}_{\fatrhoH,{\cal A}}\right]\fnleq & \ub,\text{ and }\\
\exists\quad{\cal A}\in{\cal C}\ \text{s.t.}\quad\E_{\fatrhoH\leftarrow{\cal D}}\left[{\cal W}_{\fatrhoH,{\cal A}}\right]\fngeq & \lb.
\end{align*}
When ${\cal C}$ is taken to be the class of PPT (or QPT) algorithms,
we say computational ergotropy (for the distribution ${\cal D}$),
denoted by $\erghat_{{\cal D}}$, is $(\lb,\ub)$ bounded. 
\end{defn}

\branchcolor{darkgray}{It would be helpful to make statements of the form ``$\cerg$ is
negligible''. However, stating $\cerg\fnleq\negl$ for some negligible
function turns out to be too strong in most cases---it requires that
for every algorithm ${\cal A}\in C$, the work extracted is bounded
by the same negligible function $\negl$. Therefore, we instead use
the following to captures the same intuition that no algorithm can
extract non-negligible work.}
\begin{notation}
\label{nota:erghat-is-negl}By \emph{$\cerg$ is negligible} we denote
the following property: there is no non-negligible function $\nonnegl$
for which $\cerg\fngeq\nonnegl.$
\end{notation}

\branchcolor{darkgray}{The perspective of a distribution over the family of states and Hamiltonians
is particularly helpful in the random oracle model. One can define
families of states and Hamiltonians that depends on the random oracle.
}
\begin{rem}[Ergotropy in the random oracle model]
 Consider the random oracle model where $\fat O$ is the random oracle.
\begin{enumerate}
\item The probability distribution ${\cal D}$ is over the family of states,
Hamiltonians $\fatrhoH$ and the oracle $\fat O$.
\item For a particular family of states, Hamiltonians and oracle, $(\fatrhoH,\fat O)\leftarrow{\cal D}$,
we define
\[
\W_{\fat{\rho},\fat H,\A^{\fat O}}(\secp):=\E\left[\tr\left[H_{\secp}\left(\rho_{\secp}-U_{\secp}\rho_{\secp}U^{\dagger}_{\secp}\right)\right]:\begin{array}{c}
\circdesc(U_{\secp})\leftarrow\A^{O_{\fat{\rho}},\hamdesc,\fat O}(1^{\secp})\\
U_{\secp}\in\uni(n(\secp))
\end{array}\right]
\]
where $\E$ is for a fixed $\fat{\rho},\fat H,\fat O$. 
\item We say $\cerg_{{\cal D}}$, is $(\lb,\ub)$ bounded if
\begin{align*}
\forall\quad{\cal A}\in{\cal C},\text{ it holds that }\E_{\fatrhoH,\fat O\leftarrow{\cal D}}\left[{\cal W}_{\fatrhoH,{\cal A}^{\fat O}}\right]\fnleq & \ub,\text{ and }\\
\exists\quad{\cal A}\in{\cal C}\ \text{s.t.}\quad\E_{\fatrhoH,\fat O\leftarrow{\cal D}}\left[{\cal W}_{\fatrhoH,{\cal A}^{\fat O}}\right]\fngeq & \lb.
\end{align*}
\end{enumerate}
\end{rem}

\subsection{Constructive separation between $\protect\erg_{{\cal D}}$ from $\protect\erghat_{{\cal D}}$}

\branchcolor{darkgray}{In \Cref{thm:unconditional-existential-separation} above, we showed
the \emph{existence} of states and Hamiltonians for which ergotropy
is near-maximally separated from its computational variant. We now
leverage computational hardness assumptions to \emph{explicitly} specify
states and Hamiltonians that exhibit \emph{near-maximal} separation. 

Here, we first work in the random oracle model and subsequently show
how to instantiate the oracle using a pseudorandom function to obtain
a separation in the plain model. 

The separation in the random oracle model is established through a
reprogramming technique and arguing that the output of a length doubling
random oracle (on a hidden input) is indistinguishable from a uniformly
random string. }

\subsubsection{Constructive Separation in the random oracle model}\label{subsec:A-constructive-separation}

\branchcolor{darkgray}{We start by specifying the states and Hamiltonians that we use for
the separation. These can be prepared efficiently. }

Let $n:\N\to\N$ be a polynomial which describes the size of the
system. By $\mathcal{F}$ we denote the set of functions 
\[
{\cal F}:=\left\{ g:g=\left\{ g_{\secp}\right\} _{\secp}\text{ where }g_{\secp}\text{ is from }\binset^{n(\secp)}\to\binset^{2n(\secp)}\right\} .
\]
Let $\fat O_{f}=\{O_{\secp}\}_{\secp}$ be the oracle, corresponding
to a \emph{length doubling function} $f=\{f_{\secp}\}_{\secp}$. 
\begin{defn}
\label{def:explicit-fam-state-Ham-}We define the distribution ${\cal D}$
over $\left(\fatrhoH,\fat O\right)$ in the random oracle model as follows.
\end{defn}

\begin{itemize}
\item The length doubling function  $f=\{f_{\secp}\}_{\secp}$ is uniformly
sampled from the set $\mathcal{F}$. 
\item For each fixed $\fat O_{f}$ (the oracle corresponding to $f$) we
define the families of states $\fat{\rho}^{\fat O_{f}}=\left\{ \rho_{\secp}\right\} _{\secp}$
and $\fat{\sigma}^{\fat O_{f}}=\left\{ \sigma_{\secp}\right\} _{\secp}$
over $2n(\lambda)$-qubits where 
\begin{equation}
\rho_{\secp}:=\frac{1}{2^{n(\lambda)}}\sum_{r\in\binset^{n(\secp)}}\den{f_{\secp}(r)}\label{eq:fLambda}
\end{equation}
 and 
\begin{equation}
\sigma_{\secp}:=\frac{1}{2^{2n(\lambda)}}\sum_{u_{\secp}\in\binset^{2n(\secp)}}\den{u_{\secp}}.\label{eq:maxmixedstate}
\end{equation}
 Clearly, $\sigma_{\lambda}$ is the maximally mixed state (and independent
of $\fat O_{f}$).
\item $\fat H:=\{H_{\lambda}\}_{\lambda}$ is independent of $\fat O$ and
is defined as 
\begin{equation}
H_{\lambda}:=\mathbb{I}\otimes H_{\hw(n(\lambda))}\label{eq:n_ham_n}
\end{equation}
 acting on $2n(\lambda)$ qubits.
\end{itemize}

\branchcolor{darkgray}{Using the O2H lemma, we show that no efficient procedure can distinguish
whether it holds $\fat{\rho}$ or $\fat{\sigma}$, even given access
to the oracle $\fat O$.}
\begin{lem}
\label{lem:expanding-RO-uniformity}Let $\fat O$ be the length doubling
random oracle, and $q_{1},q_{2}:\mathbb{N}\to\mathbb{N}$ be polynomials.
For every $q_{1}$-query quantum algorithm $\mathcal{B}$ with oracle
access to $\fat O$, there exists a $q_{2}$-query quantum algorithm
$\A$ with oracle access to $O_{\secp}$ corresponding to the function
$f_{\secp}$ such that,
\begin{align}
 & \left|\Pr_{{\cal D}}\left[\mathcal{B}^{\fat O}(1^{\secp},\rho_{\secp})=1\right]-\Pr_{{\cal D}}\left[\mathcal{B}^{\fat O}(1^{\secp},\sigma_{\secp})=1\right]\right|\nonumber \\
= & \left|\Pr_{{\cal D}}\left[\cA^{O_{\secp}}(1^{\secp},\rho_{\secp})=1\right]-\Pr_{{\cal D}}\left[\cA^{O_{\secp}}(1^{\secp},\sigma_{\secp})=1\right]\right|\label{eq:BO=00003DAO}\\
\leq & \negl(\secp).\label{eq:AO=00003Dnegl}
\end{align}
\end{lem}

\branchcolor{black}{\begin{proof}
We first prove that there exists an algorithm ${\cal A}$ such that
the output distributions of ${\cal B}^{\fat O}(1^{\lambda},\tau_{\lambda})$
and ${\cal A}^{O_{\lambda}}(1^{\lambda},\tau_{\lambda})$ are identical
for any $\tau_{\lambda}$ that depends only on $O_{\lambda}$. Observe
that each $f_{\lambda}$ is sampled uniformly and independently from
${\cal F}_{\lambda}$.  

Now, for every algorithm ${\cal B}^{\fat O}$, one can construct such
an algorithm ${\cal A}^{O_{\lambda}}$ as follows: ${\cal A}$ runs
${\cal B}$ identically except when ${\cal B}$ makes queries. When
${\cal B}$ makes a query to $\fat O$ at $\{0,1\}^{n(\lambda)}$,
then ${\cal A}$ uses $O_{\lambda}$, and when ${\cal B}$ makes any
other query, it simulates a random oracle independently. This establishes
\Cref{eq:BO=00003DAO}.

To establish \Cref{eq:AO=00003Dnegl}, we use a hybrid argument and
the One-Way to Hiding (O2H) Lemma. We define a sequence of hybrid
experiments $\mathsf{Exp}_{0},\mathsf{Exp}_{1},\mathsf{Exp}_{2}$
as follows
\begin{itemize}
\item $\mathsf{Exp}_{0}$: The challenger samples $O_{\secp}\leftarrow\mathcal{F_{\secp}}$
and $r\leftarrow\{0,1\}^{n(\secp)}$. It outputs $\mathcal{A}^{O_{\secp}}(O_{\secp}(r))$.
\item $\mathsf{Exp}_{1}$: The challenger samples $O_{\secp}\leftarrow\mathcal{F_{\secp}}$,
$r\leftarrow\{0,1\}^{n(\secp)}$ and $u\leftarrow\binset^{2n(\secp)}$.
It defines $O_{\secp}'$ such that $O_{\secp}'(r)=u$ and $O_{\secp}'(x)=O_{\secp}(x)$
for all $x\neq r$. It outputs $\mathcal{A}^{O_{\secp}'}(u)$.
\item $\mathsf{Exp}_{2}$: The challenger samples $O_{\secp}\leftarrow\mathcal{F_{\secp}}$,
$r\leftarrow\{0,1\}^{n(\secp)}$ and $u\leftarrow\binset^{2n(\secp)}$.
It outputs $\mathcal{A}^{O_{\secp}}(u)$.
\end{itemize}
Now, evidently $\Exp_{0}$ and $\Exp_{1}$ are identically distributed
(we just relabelled some variables).
\begin{claim}
$\Pr_{O_{\secp},r}\left[1\leftarrow\Exp_{0}\right]=\Pr_{O_{\secp},r,u}\left[1\leftarrow\Exp_{1}\right]$.\label{claim:exp0=00003D1}
\end{claim}

Furthermore, $\Exp_{1}$ and $\Exp_{2}$ have negligibly close output
distributions from the O2H lemma (see \Cref{lem:O2H}).
\begin{claim}
$\left|\Pr_{O_{\secp},r,u}\left[1\leftarrow\Exp_{1}\right]-\Pr_{O_{\secp},u}\left[1\leftarrow\Exp_{2}\right]\right|\leq\negl(n)$.\label{claim:exp1=00003D2}
\end{claim}

We defer the proofs of these claims to \Cref{sec:defElementary} in
the Appendix, as these are standard arguments. Now using the hybrid
argument, we conclude $\left|\Pr_{O_{\secp},r}\left[1\leftarrow\Exp_{0}\right]-\Pr_{O_{\secp},u}\left[1\leftarrow\Exp_{2}\right]\right|\leq\negl(n)$
proving the theorem statement.
\end{proof}

}

\branchcolor{darkgray}{The following allows us to show that when two states are computationally
indistinguishable, then the average energy one can efficiently extract
is essentially the same. }

\begin{lem}[Indistinguishability of states under efficiently preparable observables]
\label{lem:trHrho=00003DtrHsigma}Let $n:\N\to\N$ be a polynomial
describing the size of the system and $\fat{\tau_{1}}=\left\{ \tau_{1,\secp}\right\} _{\secp}$
and $\fat{\tau_{2}}=\left\{ \tau_{2,\secp}\right\} _{\secp}$ be families
of states, such that the $n(\secp)$-qubit state $\tau_{1,\secp}$
is computationally indistinguishable from $\tau_{2,\secp}$ when given
a single copy of the state. For any fixed $k\in\N$ and any family
of extensive k-local Hamiltonians $\fat H=\left\{ H_{\secp}\right\} _{\secp}$
(see \Cref{def:k-local-ham}), there exists a negligible function
$\negl$ and $\lambda_{0}\in\N$, such that
\begin{equation}
\left|\tr\left[H_{\secp}\tau_{1,\secp}\right]-\tr\left[H_{\secp}\tau_{2,\secp}\right]\right|\leq\negl(\lambda)\label{eq:trHsame_whentautau}
\end{equation}
for some $\lambda>\lambda_{0}$. Furthermore, if $U_{\lambda}$ is
produced by a QPT algorithm on input $1^{\lambda}$, then 
\begin{equation}
\left|\tr\left[H_{\secp}U_{\lambda}\tau_{1,\secp}U^{\dagger}_{\lambda}\right]-\tr\left[H_{\secp}U_{\lambda}\tau_{2,\secp}U^{\dagger}_{\lambda}\right]\right|\leq\negl(\lambda).\label{eq:trUHUsame_whentautau}
\end{equation}
\end{lem}

\branchcolor{black}{\begin{proof}
Since every $H_{\secp}\in\Ham(n(\secp))$ is a $k$-local Hamiltonian,
it can be written as,
\begin{equation}
H_{\secp}=\sum^{m}_{i=1}h^{(i)}_{\secp}\otimes\I_{[n(\secp)]\backslash S_{i}}\label{eq:k-local-ham}
\end{equation}
where $S_{i,\secp}\subseteq[n(\secp)]$, such that $\left|S_{i,\secp}\right|\leq k$
and $h_{i,\secp}\in\Ham(|S_{i}|)$ (see \Cref{def:k-local-ham}).
Since $k\in\N$ is a constant, the total number of subsets of $[n(\secp)]$
of size at most $k$ is $\binom{n(\lambda)+k}{k}\in\O(n(\secp)^{k})$
(see \Cref{nota:basic-notations}) therefore, $m\in\O(n(\secp)^{k})=\poly(\secp)$
for some polynomial function $\poly.$

Furthermore, each $h^{(i)}$ can be written as a linear combination
of projectors as follows:
\[
h^{(i)}_{\secp}=\sum^{2^{k}}_{j=1}e^{(ij)}_{\secp}\Pi^{(ij)}_{\lambda}.
\]
This follows because $h^{(i)}_{\lambda}\in\Ham(|S_{i}|)$ is an operator
on a Hilbert space of dimension at most $2^{k}$. Therefore, its spectral
decomposition has at most $2^{k}$ elements. Combining this with \Cref{eq:k-local-ham}
and using the bound on $m$, we get that 
\begin{align*}
H_{\secp} & =\sum^{\poly(\secp)}_{i=1}\sum^{2^{k}}_{j=1}e^{(ij)}_{\secp}\Pi'^{(ij)}_{\lambda} & \text{ where }\Pi'^{(ij)}_{\lambda}:=\Pi^{(ij)}_{\lambda}\otimes\I_{[n(\secp)]\backslash S_{i}}
\end{align*}
Therefore, every Hamiltonian $H_{\secp}\in\fat H$ can be written
as the linear combination of $\O\left(\poly(\secp)\right)$ projectors.
Replacing $H_{\secp}$ with its linear combination of projectors,
we get
\begin{align*}
\left|\tr\left[H_{\secp}\tau_{1,\secp}\right]-\tr\left[H_{\secp}\tau_{2,\secp}\right]\right| & =\left|\tr\left[\sum^{\poly(\secp)}_{i=1}\sum^{2^{k}}_{j=1}e^{(ij)}_{\secp}\Pi'^{(ij)}_{\lambda}\tau_{1,\secp}\right]-\tr\left[\sum^{\poly(\secp)}_{i=1}\sum^{2^{k}}_{j=1}e^{(ij)}_{\secp}\Pi'^{(ij)}_{\lambda}\tau_{2,\secp}\right]\right|\\
 & \leq\sum^{\poly(\secp)}_{i=1}\sum^{2^{k}}_{j=1}\underbrace{\left|e^{(ij)}_{\secp}\right|\cdot\left|\tr\left[\Pi'^{(ij)}_{\lambda}\tau_{1,\secp}\right]-\tr\left[\Pi'^{(ij)}_{\lambda}\tau_{2,\secp}\right]\right|}_{:=\alpha_{ij}(\lambda)} & \text{using \ensuremath{\triangle}-inequality}.
\end{align*}

For contradiction, suppose the left-hand side, i.e. $\left|\tr\left[H_{\secp}\tau_{1,\secp}\right]-\tr\left[H_{\secp}\tau_{2,\secp}\right]\right|=\nonnegl(\lambda)$.
Then, from the pigeon-hole principle one can conclude that there is
some index $i^{*}(\lambda),j^{*}(\lambda)$ such that $\alpha_{i^{*}(\lambda)j^{*}(\lambda)}(\lambda)\ge\nonnegl(\lambda)/(2^{k}\poly(\lambda))$.
This means that 
\begin{equation}
\left|\tr\left[\Pi'^{(ij)}_{\lambda}\tau_{1,\secp}\right]-\tr\left[\Pi'^{(ij)}_{\lambda}\tau_{2,\secp}\right]\right|\ge\nonnegl'(\lambda)\label{eq:projDiffNonNeglHam}
\end{equation}
 (as the $e_{\lambda}\cdot|\text{something bounded by 2| is non-negligible so \ensuremath{e_{\lambda}}}$is
non-negligible). Now consider the following distinguisher ${\cal B}$: 
\begin{itemize}
\item on input $(\tau_{b,\lambda},1^{\lambda})$, it samples $i'\leftarrow[\poly(\lambda)],j'\leftarrow[2^{k}]$
(see \Cref{nota:basic-notations}) 
\item measures $\tau_{b,\lambda}$using $\Pi'^{(i'j')}_{\lambda}$ and outputs
$1$ if the outcome corresponds to $\Pi'^{(i'j')}_{\lambda}$ and
$0$ otherwise. 
\end{itemize}
The distinguishing probability of ${\cal B}$ is at least 
\begin{align*}
\frac{1}{2^{k}\cdot\poly(\lambda)}\sum_{ij}\left|\tr[\tau_{1,\lambda}\Pi'^{(ij)}_{\lambda}]-\tr[\tau_{2,\lambda}\Pi'^{(ij)}_{\lambda}]\right| & \ge\frac{1}{2^{k}\cdot\poly(\lambda)}\left|\tr[\tau_{1,\lambda}\Pi'^{(i^{*}(\lambda)j^{*}(\lambda))}_{\lambda}]-\tr[\tau_{2,\lambda}\Pi'^{(i^{*}(\lambda)j^{*}(\lambda))}_{\lambda}]\right|\\
 & =\frac{1}{2^{k}\cdot\poly(\lambda)}\cdot\nonnegl'(\lambda) \qquad\qquad \qquad  \text{using }\cref{eq:projDiffNonNeglHam}\\
 & =\nonnegl''(\lambda).
\end{align*}
However, this contradicts the indistinguishability of $\tau_{1,\secp}$
and $\tau_{2,\secp}$. Thus, 
\[
\left|\tr\left[H_{\secp}\tau_{1,\secp}\right]-\tr\left[H_{\secp}\tau_{2,\secp}\right]\right|\le\negl(\lambda)
\]
which proves \Cref{eq:trHsame_whentautau}. 

Given \Cref{eq:trHsame_whentautau} holds, establishing \Cref{eq:trUHUsame_whentautau}
is rather straightforward: if no efficient algorithm can distinguish
between $\tau_{1,\secp}$ and $\tau_{2,\secp}$, then no efficient
algorithm can distinguish between $U_{\lambda}\tau_{1,\secp}U^{\dagger}_{\lambda}$
and $U_{\lambda}\tau_{2,\secp}U^{\dagger}_{\lambda}$ as well, where
$\desc(U_{\secp})$ is produced by an efficient algorithm.
\end{proof}

}
\begin{rem}
\label{rem:trHrho=00003DtrHsigma_rel}Consider the cases where one
is working relative to oracles $\fat O$, i.e. where the distinguishing
algorithms have access to an oracle $\fat O$, and where $\fat{\tau}$s
can depend on $\fat O$. Given that $\fat{\tau}_{1}$ is computationally
indistinguishable from $\fat{\tau}_{2}$, one can verify that the
results in \Cref{lem:trHrho=00003DtrHsigma} continue to hold. 
\end{rem}

\branchcolor{darkgray}{We are now ready to give the explicit families of states and Hamiltonians
for which we obtain maximal separation.}

\branchcolor{darkgray}{\Cref{thm:explicit-sep-erg-erghat} below, is one of our main results
which shows that for the family states and extensive Hamiltonians
we described in \Cref{def:explicit-fam-state-Ham-}, $\erg_{{\cal D}}$
is at least $n/2$ asymptotically while $\erghat_{{\cal D}}$ is negligible.}

\begin{thm}[Separating $\erg_{{\cal D}}$ from $\erghat_{{\cal D}}$ in the random
oracle model]
\label{thm:explicit-sep-erg-erghat}In the random oracle model, for the family of states, Hamiltonians
and the oracle, $\fatrhoH,\fat O$ sampled from distribution ${\cal D}$ as defined above (see \Cref{def:explicit-fam-state-Ham-}), it holds that 
\begin{align*}
\nexists\ \ \nonnegl\ \ \text{s.t.}\ \ \erghat_{{\cal D}} & \fngeq\nonnegl,\text{ while}\\
\erg_{{\cal D}} & \fngeq\frac{n}{2}
\end{align*}
where $\nonnegl$ is any non-negligible function.
\end{thm}

\branchcolor{black}{\begin{proof}
We show the first inequality in \Cref{claim:erghat-is-negl-as-always}
and the second in \Cref{claim:erg=00003Dhalf-of-n}.
\begin{claim}
\label{claim:erghat-is-negl-as-always}In the random oracle model,
there does not exist a non-negligible function $\nonnegl$ such that
$\erghat_{{\cal D}}\fngeq\nonnegl$.
\end{claim}

\branchcolor{black}{\begin{proof}
For every $\QPT$ algorithm $\mathcal{B}^{\fat O,O_{\fat{\rho},}\hamdesc}(1^{\secp})$
that outputs the description of a unitary, one can do the following
analysis to upper bound $\W_{\fatrhoH,{\cal B}}$.  We proceed in
two steps. First we show that one can extract a unitary $U$ from
${\cal B}$ without giving access to $O_{\fat{\rho}}$. Then, from
the fact that $\rho$ and $\sigma$ are computationally indistinguishable,
we can use \Cref{rem:trHrho=00003DtrHsigma_rel} and prove that the
associated energies are the same, up to negligible factors.

More precisely, one can construct a simulator $\mathcal{D}^{\fat O}$,
which upon receiving the input state $\tau_{\secp}\in\left\{ \rho_{\secp},\sigma_{\secp}\right\} $,
${\cal D}^{\fat O}(1^{\secp})$ simulates $\mathcal{B}^{\fat O,O_{\fat{\rho},}\hamdesc}(1^{\secp})$
to obtain unitary $U_{\secp}$. ${\cal D}^{\fat O}$ does the following:
\begin{enumerate}
\item ${\cal D}$ runs ${\cal B}$ exactly, when ${\cal B}$ does not make
any queries to the oracles.
\item When ${\cal B}$ queries the oracle $\fat O$, ${\cal D}$ also queries
the oracle $\fat O$. 
\item When ${\cal B}$ queries the oracle $O_{\fat{\rho}}$, ${\cal D}$
simulates the output of $O_{\fat{\rho}}$ by randomly sampling the
bit string $r_{1}\leftarrow\binset^{n(\secp)}$, querying the oracle
$\fat O$ and returning the state $\den{f_{\secp}(r_{1})}$. Since
$r_{1}$ is sampled uniformly and independently, the state so prepared
is identical to the state $\rho_{\secp}.$
\item When ${\cal B}$ queries the oracle $\hamdesc$, ${\cal D}$ can simulate
$\hamdesc$ exactly, because it is independent of the oracle $\fat O$.
Therefore, ${\cal D}$ need not make any additional queries to $\fat O$.
\end{enumerate}
From \Cref{lem:expanding-RO-uniformity} we know that $\rho_{\secp}\text{ and }\sigma_{\secp}$
are computationally indistinguishable. Using \Cref{lem:trHrho=00003DtrHsigma}
and \Cref{rem:trHrho=00003DtrHsigma_rel} this means that for any
unitary $U_{\secp}$ obtained by ${\cal D}^{\fat O}(1^{\lambda})$
(which is a QPT algorithm), there exists negligible functions $\negl_{1},\negl_{2}$
such that
\begin{align*}
\left|\tr\left[H_{\secp}\rho_{\secp}\right]-\tr\left[H_{\secp}\sigma_{\secp}\right]\right| & \leq\negl_{1}(\secp),\text{ and }\\
\left|\tr\left[H_{\secp}U_{\secp}\rho_{\secp}U^{\dagger}_{\secp}\right]-\tr\left[H_{\secp}U_{\secp}\sigma_{\secp}U^{\dagger}_{\secp}\right]\right| & \leq\negl_{2}(\secp).
\end{align*}
Using triangle inequality, we have
\begin{align*}
\left|\left(\tr\left[H_{\secp}\rho_{\secp}\right]-\tr\left[H_{\secp}U_{\secp}\rho_{\secp}U^{\dagger}_{\secp}\right]\right)-\underbrace{\left(\tr\left[H_{\secp}\sigma_{\secp}\right]-\tr\left[H_{\secp}U_{\secp}\sigma_{\secp}U^{\dagger}_{\secp}\right]\right)}_{0\text{ since }\sigma_{\secp}\text{ is invariant under unitary operations}}\right| & \leq\negl_{1}(\secp)+\negl_{2}(\secp)\\
\implies\left(\tr\left[H_{\secp}\rho_{\secp}\right]-\tr\left[H_{\secp}U_{\secp}\rho_{\secp}U^{\dagger}_{\secp}\right]\right) & \leq\negl_{1}(\secp)+\negl_{2}(\secp)\\
\implies\W_{\fatrhoH,{\cal B}}(\secp) & \leq\negl_{1}(\secp)+\negl_{2}(\secp)\\
 & =\negl_{{\cal B}}(\secp).
\end{align*}

Since for all ${\cal B}\in\QPT$, we can construct a simulator ${\cal D}$
which only uses the random oracle $\fat O$, there exists some negligible
function $\negl_{{\cal B}}$ such that, $\W_{\fat{\rho},\fat H,\mathcal{B}}(\secp)\leq\negl_{{\cal B}}(\secp)$,
it follows that there does not exist a non-negligible function $\nonnegl$
such that $\erghat_{{\cal D}}\fngeq\nonnegl$.
\end{proof}

}

\begin{claim}
\label{claim:erg=00003Dhalf-of-n}It holds that $\erg_{{\cal D}}(\secp)=\frac{n(\secp)}{2}$.
\end{claim}

\branchcolor{black}{\begin{proof}
For the family of states and Hamiltonians $\fatrhoH$ it holds that
\[
\tr\left[H_{\secp}\rho_{\secp}\right]=\frac{n(\secp)}{2}.
\]
This is because, each of the last $n(\secp)$-qubits of $\rho_{\secp}$
corresponds to 0 or 1 in the computational basis with equal probability
and since $H_{\secp}$ is the hamming weight Hamiltonian on the last
$n(\secp)$-qubits, the expected hamming weight of the last $n(\secp)$-qubits
is $\frac{n(\secp)}{2}.$

A query unbounded algorithm $\A^{\fat O,O_{\fat{\rho}},\hamdesc}(1^{\secp})$
can query the random oracle at all points in $\binset^{n(\secp)}$
and find the complete truth table of the random oracle $O_{\secp}.$
With this information, it can output the description of a unitary
$U_{\secp},$ which maps $\ket{f_{\secp}(r)}\mapsto\ket{r^{*},0^{n(\secp)}}$,
where $r^{*}$ is the lexicographically smallest inverse of $f_{\secp}(r)$,
i.e. $r^{*}\in\binset^{n(\secp)}$ is the smallest value such that
$f_{\secp}(r^{*})=f_{\secp}(r).$

For such a unitary $U_{\secp}$, it holds that
\begin{align*}
\tr\left[H_{\secp}U_{\secp}\rho_{\secp}U^{\dagger}_{\secp}\right] & =0.
\end{align*}
This is because after applying the unitary $U_{\secp}$ on the state
$\rho_{\secp}$, the resulting state would always have $0^{n(\secp)}$
in the last $n(\secp)$-qubits and therefore, the hamming weight over
the last $n(\secp)$-qubits is 0.

Therefore, for using this explicit unbounded algorithm $\A$, we are
able to construct a unitary $U_{\secp}$ such that, 
\[
\tr\left[H_{\secp}\rho_{\secp}\right]-\tr\left[H_{\secp}U_{\secp}\rho_{\secp}U^{\dagger}_{\secp}\right]=\frac{n(\secp)}{2}.
\]
From this, we conclude that $\erg_{{\cal D}}\fngeq\frac{n(\secp)}{2}.$
\end{proof}

}

Combining \Cref{claim:erghat-is-negl-as-always,claim:erg=00003Dhalf-of-n}
we have the proof for the theorem. 
\end{proof}

}

\subsubsection{Constructive Separation in the plain model}

\branchcolor{darkgray}{In \Cref{thm:explicit-sep-erg-erghat-prf} below, we lift the separation
in the random oracle model (see \Cref{thm:explicit-sep-erg-erghat})
to the plain model by instantiating the random oracle with a pseudorandom
function (PRF) whose key is hidden.}

Let $F_{\lambda}:\{0,1\}^{k(\lambda)}\times\{0,1\}^{n(\lambda)}\to\{0,1\}^{2n(\lambda)}$
be an expanding pseudorandom function (see \Cref{def:PRF}). 
\begin{defn}
\label{def:explicit-fam-state-Ham-PRF}Let $n,\kappa:\N\to\N$ be
polynomials denoting the size of the quantum system and size of the
key respectively. We define a distribution ${\cal D}$ over families
of states $(\fat{\rho},\fat H)$ as follows. 
\begin{itemize}
\item The family of keys $\fat{\k}=\{\k_{\lambda}\}_{\lambda}$ where $\k_{\lambda}\leftarrow\{0,1\}^{\kappa(\lambda)}$.
\item The family of states $\fat{\rho}=\left\{ \rho_{\secp}\right\} _{\secp}$
acting on $2n(\lambda)$-qubits where 
\[
\rho_{\lambda}:=\frac{1}{c}\sum_{r\in\{0,1\}^{n(\lambda)}}\den{F_{\secp}(\key_{\lambda},r)}
\]
where $c$ is chosen to ensure $\tr(\rho_{\lambda})=1$. 
\item The family of Hamiltonians $\fat H=\left\{ H_{\secp}\right\} _{\secp}$
(independent of $\fat{\k}$) where $H_{\lambda}$ is the padded Hamming
weight Hamiltonian (see \Cref{eq:n_ham_n}).
\end{itemize}
\end{defn}

\branchcolor{darkgray}{Intuitively, we have PRF $F_{\lambda}(\key,\cdot)$ play the role
of the random oracle $O_{\lambda}$. While this cannot be done in
general, it is straightforward to see that the proof of \Cref{thm:explicit-sep-erg-erghat}
can be extended to work for a PRF.}
\begin{thm}[Separating $\erg_{{\cal D}}$ from $\erghat_{{\cal D}}$ in the plain
model]
\label{thm:explicit-sep-erg-erghat-prf}Assuming the existence of quantum secure PRFs, for the family of states,
Hamiltonians, $\fat{\rho},\fat H$ sampled from ${\cal D}$ as defined
above (see \Cref{def:explicit-fam-state-Ham-PRF}), it holds that
\begin{align*}
\nexists\ \ \nonnegl\ \ \text{s.t.}\ \ \erghat_{{\cal D}} & \fngeq\nonnegl,\text{ while}\\
\erg_{{\cal D}} & \fngeq\frac{n}{2}
\end{align*}
where $\nonnegl$ is any non-negligible function.
\end{thm}

\branchcolor{black}{\begin{proof}[Proof sketch]

To establish $\erg_{{\cal D}}\fngeq\frac{n}{2}$, observe that an
unbounded algorithm $\A$ can query the state generation oracle $O_{\fat{\rho}}$
exponentially many times, to enumerate all possible values $F_{\secp}(r)$
assumes. With this information, $\A$ can output the description of
a unitary that maps $\left\{ \ket{F_{\secp}(r)}\right\} _{r\in\binset^{n(\secp)}}\mapsto\left\{ \ket{r,0^{n(\secp)}}\right\} _{r\in\binset^{n(\secp)}}.$
This way, the hamming weight of the second half of the string is 0,
therefore attaining $\erg_{{\cal D}}=\frac{n}{2}.$ 

To establish an upper bound on $\erghat_{{\cal D}}$, we consider
a slightly relaxed problem, i.e. we allow all algorithms to access
$F_{\lambda}(\k,\cdot)$. Since we are only giving extra information
to the algorithms, work extracted in this relaxed setting, will still
serve as an upper bound on $\erghat_{{\cal D}}$.

In this relaxed setting, we want to ensure that
\begin{enumerate}
\item $\fat{\rho}$ is computationally indistinguishable from $\fat{\sigma}$
(see \Cref{eq:maxmixedstate}),
\item The unitary $U_{\lambda}$ produced by any efficient algorithm ${\cal B}^{O_{\fat{\rho}},\hamdesc,F_{\lambda}(\k,\cdot)}(1^{\lambda})$
can also be produced by an efficient algorithm ${\cal C}^{F_{\lambda}(\k,\cdot)}(1^{\lambda})$
(where $\hamdesc$ can be assumed to be hardcoded in ${\cal C}$ because
it describes a local Hamiltonian).
\end{enumerate}
The first point follows by applying a hybrid argument as follows:
(i) replace the PRF with a random oracle (with negligible error using
\Cref{def:PRF}) and (ii) in the random oracle model, use the indistinguishability
of $\fat{\rho}_{\RO}$ and $\fat{\sigma}$ (see \Cref{lem:expanding-RO-uniformity})
where $\fat{\rho}_{\RO}\propto(\sum_{r\in\{0,1\}^{n(\lambda)}}\den{f_{\lambda}(r)})$
is the state in \Cref{eq:fLambda}. The second point can be proved
by proceeding as we did in the proof of \Cref{claim:erghat-is-negl-as-always}.

One can now use \Cref{lem:trHrho=00003DtrHsigma} and \Cref{rem:trHrho=00003DtrHsigma_rel}
(as we did in the proof of \Cref{claim:erghat-is-negl-as-always})
to conclude that there is a negligible function $\negl'$ such that
${\cal W}_{\fatrhoH,{\cal B}}\fnleq\negl'$ (for every QPT algorithm
${\cal B}$). This establishes that $\erghat_{{\cal D}}$ is negligible
(see \Cref{nota:erghat-is-negl}).

\end{proof}

}

\section{Catalytic Computation}\label{sec:Catalytic-Computation}

\branchcolor{darkgray}{The definition of computational work extraction we considered above
(see \Cref{def:work-ex-algo}), requires that the unitary output by
the algorithm, must act strictly only on $n(\secp)$-qubits. This
requirement, however, appears to be too restrictive. For instance,
even the most basic cryptographic applications rely on keyed cryptographic
primitives, where the key $\k$ is used explicitly. If one is limited
to using only as many qubits as the input, one quickly runs into such
issues as the following problem. 
\begin{problem}
\label{prob:in-place-permutations-are-hard}Given the key $\k\in\binset^{\secp},$
implement a pseudo-random permutation $\pi_{\k}$ on $n(\secp)$-bit
strings, using only $\lambda+n(\secp)$-qubits, i.e. \emph{efficiently}
implement $\left|\k,x\right\rangle \mapsto\left|\k,\pi_{\k}(x)\right\rangle $.
\end{problem}

The standard compilation of arbitrary classical circuits into reversible
quantum circuits usually requires auxiliary space to store intermediate
computations \cite{bennet1989_timespace_tradeoffs_rev_computing}.
The standard pseudorandom permutation oracle acts on $\secp+2n(\secp)$-qubits,
$\secp$-qubits to key the permutation to use, $n(\secp)$-qubits
for input and $n(\secp)$-qubits for output,
\[
O_{\pi}\ket{\k}_{\key}\ket x_{i}\ket 0_{o}=\ket{\k}_{\key}\ket x_{i}\ket{\pi_{\k}(x)}_{o}
\]
Without auxiliary space, it is unclear whether standard cryptographic
primitives can be implemented efficiently. One may be tempted to resolve
this by allowing quantum algorithms to use auxiliary qubits (ancillas).
While this approach is often taken when one considers polynomial-time
machines, in the physical context here, this introduces a fundamental
flaw---unconstrained fresh qubits may be exploited to extract near-maximal
work even from passive states (see \Cref{def:passiveState}), as illustrated
by the following.
\begin{example}
\label{exa:HammingSwapping}Consider the Hamming-weight Hamiltonian
$H_{\hw(n)}$ (see \Cref{def:HammingWeight-Ham}) acting on $n$ qubits
and consider the maximally mixed state $\sigma_{n}\propto\mathbb{I}_{n}$.
Clearly, one should not be able to extract any work from $\sigma_{n}$
using a unitary operation. However, if one allows for ancillas initialised
to $\den{0^{n}}$, then one can easily swap the system and ancilla
registers (which is a unitary operation) to extract $n/2$ units of
work.
\end{example}

The discussion so far, suggests that the physically relevant definition
of computational work extraction must be some kind of generalisation
of \Cref{def:work-ex-algo}. We propose the following generalisation:
allow auxiliary qubits that must be restored exactly to the state
they were initialised in. We call such auxiliary qubits, catalyst
qubits and the corresponding model of computation, the \emph{catalytic
model of computation}. Clearly, the pathological case described in
\Cref{exa:HammingSwapping} no longer applies---and indeed, as we
show later (see \Cref{sec:Computational-Catalytic-Work}), the operational
meaning of ergotropy remains meaningful.

Furthermore, in this catalytic model of computation, it is immediate
that one can now solve \Cref{prob:in-place-permutations-are-hard}
by implementing the PRP as follows: 
\label{eq:PRP-using-cats}\[
\left|\k,x,0\right\rangle \mapsto\left|\k,x,\pi_{\k}(x)\right\rangle \mapsto\left|\k,0,\pi_{\k}(x)\right\rangle \mapsto\left|\k,\pi_{\k}(x),0\right\rangle .
\]
where the second step erases $x$ by computing $\pi^{-1}_{\k}$ of
the third register. Each classical logic gates $\mathsf{gate}$ is
mapped to controlled gate that acts as $\left|x\right\rangle \left|b\right\rangle \mapsto\left|x\right\rangle \left|b\oplus\mathsf{gate}(x)\right\rangle $.

Clearly, some catalytic models of computation allows for potentially
more operations. There could be various notions of catalysts (e.g.
whether they are restored exactly or approximately, whether they can
be left slightly entangled, whether the catalytic condition is required
for a fixed state or all states, etc.). A priori, it is unclear whether
these additional operations allow one to map an input state $\left|\psi\right\rangle $
to a larger class of output states $\left|\phi\right\rangle $. We
address this in both the computationally bounded and unbounded setting---starting
with the latter.

}

\subsection{Catalytic Unitaries}

\branchcolor{darkgray}{Consider a bipartite system with Hilbert space ${\cal H}=\H_{\IO}\otimes\H_{\catalyst}$,
where $\H_{\IO}$ denotes the Hilbert space corresponding to an \emph{input/output}
(unconstrained) quantum space and $\H_{\catalyst}$ denotes the \emph{catalytic
space}. We start by considering a model of catalytic computation where
unitaries act on ${\cal H}$ in such a way that every state in the
catalytic space ${\cal H}_{\catalyst}$, is restored exactly after
the computation. }

Let $n,\ell\in\N$, $\H_{\IO}=(\mathbb{C}^{2})^{\otimes n}\text{ and }\H_{\catalyst}=(\mathbb{C}^{2})^{\otimes\ell}$.
\begin{defn}[{Catalytic Unitaries $\catU[n][\ell]$}]
\label{def:cat-unitaries}Let the $(n+\ell)$-qubit composite Hilbert
space be denoted by ${\cal H}:=\H_{\IO}\otimes\H_{\catalyst}.$ Then,
the set of \emph{catalytic unitaries} $\catU[n][\ell]$ is defined
as

\[
\catU[n][\ell]:=\left\{ U\in{\cal U}(n+\ell):\begin{array}{r}
\forall\ \left|\phi\right\rangle \otimes\ket{\eta}\in\H,\quad\exists\ \left|\psi\right\rangle \in{\cal H}_{\IO}\\
\text{s.t. }U\left|\phi\right\rangle \otimes\left|\eta\right\rangle =\left|\psi\right\rangle \otimes\left|\eta\right\rangle 
\end{array}\right\} .
\]
\end{defn}

\branchcolor{darkgray}{Since the primary system of interest is the register $\IO$, one would
often want to consider only the reduced state on $\IO$. For brevity,
we introduce the following notation to represent the reduced state
of the system after the application of a catalytic unitary.}
\begin{notation}
\label{nota:partial-tr-channel-catU}Let ${\cal V}\subseteq\uni(n+\ell)$
be set of unitaries on ${\cal H}_{\IO}\otimes{\cal H}_{\catalyst}$.
Then, by $C_{{\cal V}}$ we denote the set of \emph{channels} \emph{associated}
with the action of these unitaries using the catalyst as
\[
C_{{\cal V}}:=\{C^{\sigma}_{U}(\cdot):U\in{\cal V},\ \sigma\in\Den(\ell)\}.
\]
where $C^{\sigma}_{U}(\cdot):=\tr_{\catalyst}U(\ \cdot\ \otimes\sigma)U^{\dagger}$
is the channel associated with $(U,\sigma)$. 
\end{notation}

\begin{example}
These may help to clarify the notation. 
\begin{enumerate}
\item Note that for $\ell=0$, $C_{{\cal U}(n)}=\{V(\cdot)V^{\dagger}:V\in{\cal U}(n)\}$.
\item Note also that $C_{\catU[n][\ell]}$ is simply $\{\tr_{\catalyst}U(\ \cdot\ \otimes\sigma)U^{\dagger}:U\in\catU[n][\ell],\ \sigma\in\Den(\ell)\}$.
\end{enumerate}
\end{example}

\branchcolor{darkgray}{Clearly, it holds that $C_{\uni(n)}=C_{\catU[n][\ell]}$ when $\ell=0$
because ${\cal U}(n)=\catU[n][0]$. On the other hand, with $\ell>0$,
$C_{\catU[n][\ell]}$ may allow one to map $\left|\phi\right\rangle $
to a larger class of states $\left|\psi\right\rangle $, i.e. the
containment 
\[
C_{\uni(n)}\subseteq C_{\catU[n][\ell]}
\]
could potentially be strict. 

A little thought reveals that the two classes are equivalent, i.e.
$C_{\catU[n][\ell]}=C_{\uni(n)}.$ Suppose $U\in\catU[n][\ell]$.
For an orthonormal basis $\left\{ \left|\phi_{i}\right\rangle \right\} $
of ${\cal H}_{\IO}$, and for any $\left|\eta\right\rangle \in{\cal H}_{\catalyst}$,
the action of $U$ must be $U\left|\phi_{i}\right\rangle \otimes\left|\eta\right\rangle =\left|\psi_{i}\right\rangle \otimes\left|\eta\right\rangle $.
Since $U$ must map orthonormal vectors to orthonormal vectors, we
have that $\{\left|\psi_{i}\right\rangle \}_{i}$ must also be an
orthonormal basis of ${\cal H}_{\IO}$. Therefore, for $V:=\sum_{i}\left|\psi_{i}\right\rangle \left\langle \phi_{i}\right|\in\uni(n)$,
we have that $C^{\eta}_{U}(\phi)=C_{V}(\phi)$ for all $\ket{\phi}\in\H_{\IO}.$
 }
\begin{claim}[{$C_{\uni(n)}=C_{\catU[n][\ell]}$}]
\label{claim:exactCatalystUseless}For all $U\in\catU[n][\ell]$
there exists $V\in\uni(n)$ such that, for all $\ket{\phi}_{\IO}\otimes\ket{\eta}_{\catalyst}\in\H_{\IO}\otimes\H_{\catalyst}$
it holds that
\[
C^{\eta}_{U}\left(\phi\right)=C_{V}\left(\phi\right).
\]
\end{claim}

\branchcolor{darkgray}{In view of this, one might argue that perhaps the definition is too
strict---requiring \emph{all} catalysts to be restored \emph{exactly}
is not very physical. Indeed, no physical process happens with zero
error. Further, using a completely unknown state in a register that,
must be restored exactly, seems rather difficult. 

We thus, address both of these arguments at once and consider a more
relaxed notion of catalytic computation: we only require that for
some fixed state, the catalytic property holds---and that too up
to some error $\varepsilon$. In fact, we even allow the $\catalyst$
register to stay slightly entangled with the $\IO$ register, as long
as it is close to being a tensor product. More formally, we consider
the following notion. }

\begin{defn}[{Approximate Catalytic Unitaries $\catU[n][\ell,\varepsilon]$}]
\label{def:approx-cat-unitaries}Let the $(n+\ell)$-qubit composite
Hilbert space be denoted by ${\cal H}:=\H_{\IO}\otimes\H_{\catalyst}.$
Then, the set of \emph{approximate catalytic unitaries} $\catU[n][\ell,\varepsilon]$
is defined as 
\[
\catU[n][\ell,\varepsilon]:=\left\{ U\in{\cal U}(n+\ell):\begin{array}{r}
\exists\left|\eta\right\rangle \in{\cal H}_{\catalyst}\ \forall\ \left|\phi\right\rangle \in\H_{\IO},\quad\exists\ \left|\psi\right\rangle \otimes\ket{\xi}\in\H\\
\text{s.t. }\TD{U\left(\phi\otimes\eta\right)U^{\dagger}}{\psi\otimes\xi}\le\varepsilon
\end{array}\right\} .
\]
\end{defn}

\branchcolor{darkgray}{We now investigate the same question: does relaxing the catalyst
restoration to being approximate, result in more computational power?
To wit: 
\[
C_{\catU[n][\ell,\varepsilon]}\stackrel{?}{\subsetneq}C_{\uni(n)}.
\]
It turns out that, up to some error (that vanishes with $\varepsilon$)
these operations are still the same. We prove this in \Cref{lem:dirty-vs-no-cat-unbounded}
and \Cref{thm:approx-cat-restoration-is-useless} below. 

 }
\begin{lem}
\label{lem:dirty-vs-no-cat-unbounded}Let $S$ be a finite-dimensional quantum system with Hilbert space
$\H_S$, and let $C$ be a catalyst register with Hilbert space $\H_C$.
Let the catalyst be initially prepared in a fixed pure state
$\ket{\eta}_C$, and let $U_{SC}$ be a fixed unitary on
$\H_S\otimes\H_C$. Define the isometry
$W:\H_S\to\H_S\otimes\H_C$ by
\begin{equation}
\label{eq:def-isometry-W}
W\ket{\phi}
:=
U_{SC}\bigl(\ket{\phi}_S\ket{\eta}_C\bigr),
\end{equation}
and let the induced channel on $S$ be
\begin{equation}
\label{eq:def-induced-system-channel}
\mathcal{N}_S(\rho)
:=
\tr_C\left[W\rho W^\dagger\right].
\end{equation}

Suppose that $0\leq\varepsilon<1$ and that, for every pure state
$\ket{\phi}_S$, there exists a pure state
$\ket{\psi_\phi}_S$ such that
\begin{equation}
\label{eq:joint-approximate-restoration-assumption}
\mathsf{TD}\left(
W\den{\phi}W^\dagger,\,
\den{\psi_\phi}_S\otimes\den{\eta}_C
\right)
\leq \varepsilon.
\end{equation}
Then there exists a fixed unitary $V_S$ on $\H_S$ such that, for every $\rho\in\Den(S)$,
\begin{equation}
\label{eq:approx-cat-close-unitary-channel}
\mathsf{TD}\left(
\mathcal{N}_S(\rho),\,
V_S\rho V_S^\dagger
\right)
\leq
\frac{3\varepsilon^2}{2}.
\end{equation}
\end{lem}

\branchcolor{black}{\begin{proof}
Define the linear operator $\mathcal{L}:\H_S\to\H_S$ by
\begin{equation}
\label{eq:def-L-joint-restoration}
\mathcal{L}
:=
\left(\mathbb{I}_S\otimes\bra{\eta}_C\right)W.
\end{equation}
For every pure state $\ket{\phi}_S$, the two states appearing in
\Cref{eq:joint-approximate-restoration-assumption} are pure. Hence,
using
\begin{equation}
\label{eq:pure-state-trace-distance-overlap}
\mathsf{TD}\left(\den a,\den b\right)
=
\sqrt{1-\left|\braket{a}{b}\right|^2},
\end{equation}
we obtain
\begin{align}
\left|
\left(
\bra{\psi_\phi}_S\otimes\bra{\eta}_C
\right)
W\ket{\phi}
\right|
&\geq
\sqrt{1-\varepsilon^2}
\nonumber\\
\implies
\left|
\bra{\psi_\phi}\mathcal{L}\ket{\phi}
\right|
&\geq
\sqrt{1-\varepsilon^2}.
\label{eq:L-overlap-lower-bound}
\end{align}
Since $\ket{\psi_\phi}$ is normalized,
\Cref{eq:L-overlap-lower-bound} implies
\begin{equation}
\label{eq:L-norm-lower-bound-new}
\norm{\mathcal{L}\ket{\phi}}
\geq
\sqrt{1-\varepsilon^2}
\qquad
\text{for every normalized }\ket{\phi}\in\H_S.
\end{equation}
Moreover, $\mathcal{L}$ is a contraction because it is obtained by
projecting the isometry $W$ onto the catalyst state $\ket{\eta}$.
Therefore,
\begin{equation}
\label{eq:L-norm-sandwich-new}
\sqrt{1-\varepsilon^2}
\leq
\norm{\mathcal{L}\ket{\phi}}
\leq 1
\qquad
\text{for every normalized }\ket{\phi}\in\H_S.
\end{equation}
Equivalently,
\begin{equation}
\label{eq:LdaggerL-sandwich}
(1-\varepsilon^2)\mathbb{I}_S
\preceq
\mathcal{L}^\dagger\mathcal{L}
\preceq
\mathbb{I}_S.
\end{equation}

Let
\begin{equation}
\label{eq:L-svd-new}
\mathcal{L}
=
\sum_{j=1}^{d}
s_j\ket{u_j}\bra{v_j}
\end{equation}
be a singular-value decomposition. By
\Cref{eq:LdaggerL-sandwich},
\begin{equation}
\label{eq:singular-values-L-new}
\sqrt{1-\varepsilon^2}
\leq
s_j
\leq 1
\qquad
\text{for every }j\in[d].
\end{equation}
Define
\begin{equation}
\label{eq:def-V-from-L-new}
V_S
:=
\sum_{j=1}^{d}
\ket{u_j}\bra{v_j}.
\end{equation}
Since $\varepsilon<1$, all singular values of $\mathcal{L}$ are
strictly positive, and hence $V_S$ is unitary. Furthermore,
\begin{align}
\norm{\mathcal{L}-V_S}_\infty
&=
\max_j |s_j-1|
\nonumber\\
&\leq
1-\sqrt{1-\varepsilon^2}.
\label{eq:L-V-operator-bound-new}
\end{align}

We now separate the component of $W$ in which the catalyst is restored
exactly from the component orthogonal to $\ket{\eta}_C$. Define
$K:\H_S\to\H_S\otimes\H_C$ by
\begin{equation}
\label{eq:def-K-leakage}
K
:=
W-
\left(\mathcal{L}\otimes\ket{\eta}_C\right).
\end{equation}
By the definition of $\mathcal{L}$,
\begin{equation}
\label{eq:K-orthogonal-catalyst}
\left(\mathbb{I}_S\otimes\bra{\eta}_C\right)K
=
0.
\end{equation}
Consequently, the two terms in
\Cref{eq:def-K-leakage} have orthogonal catalyst components. Since
$W^\dagger W=\mathbb{I}_S$, we therefore have
\begin{equation}
\label{eq:KdaggerK}
K^\dagger K
=
\mathbb{I}_S-\mathcal{L}^\dagger\mathcal{L}
\preceq
\varepsilon^2\mathbb{I}_S,
\end{equation}
where the final inequality follows from
\Cref{eq:LdaggerL-sandwich}.

Using \Cref{eq:def-K-leakage,eq:K-orthogonal-catalyst}, the reduced system channel can be written as
\begin{equation}
\label{eq:reduced-channel-L-K}
\mathcal{N}_S(\rho)
=
\mathcal{L}\rho\mathcal{L}^\dagger
+
\tr_C\left[K\rho K^\dagger\right].
\end{equation}
Indeed, the cross terms vanish after tracing out the catalyst because the image of $K$ is orthogonal to the catalyst state $\ket{\eta}_C$.

Let
\begin{equation}
\label{eq:def-leakage-state-R}
R_\rho
:=
\tr_C\left[K\rho K^\dagger\right].
\end{equation}
This is positive semidefinite and, by
\Cref{eq:KdaggerK},
\begin{equation}
\tr R_\rho
=
\tr\left[K^\dagger K\rho\right]
\nonumber\leq
\varepsilon^2.
\label{eq:leakage-trace-bound}
\end{equation}

We first compare the first term in \Cref{eq:reduced-channel-L-K} with the unitary channel generated by $V_S$. Using the H\"older inequality, $\norm{\rho}_1=1$, $\norm{\mathcal{L}}_\infty\leq1$, and
$\norm{V_S}_\infty=1$, we obtain
\begin{align}
\frac{1}{2}
\norm{
\mathcal{L}\rho\mathcal{L}^\dagger
-
V_S\rho V_S^\dagger
}_1
&\leq
\frac{1}{2}
\norm{
(\mathcal{L}-V_S)\rho\mathcal{L}^\dagger
}_1
+
\frac{1}{2}
\norm{
V_S\rho(\mathcal{L}^\dagger-V_S^\dagger)
}_1
\nonumber\\
&\leq
\norm{\mathcal{L}-V_S}_\infty
\nonumber\\
&\leq
1-\sqrt{1-\varepsilon^2}.
\label{eq:L-channel-close-V-channel}
\end{align}
Moreover, since $R_\rho\succeq0$,
\begin{equation}
\label{eq:trace-distance-leakage-term}
\frac{1}{2}\norm{R_\rho}_1
=
\frac{1}{2}\tr R_\rho
\leq
\frac{\varepsilon^2}{2}.
\end{equation}
Combining
\Cref{eq:reduced-channel-L-K,eq:L-channel-close-V-channel,eq:trace-distance-leakage-term}
with the triangle inequality yields
\begin{equation}
\mathsf{TD}\left(
\mathcal{N}_S(\rho),
V_S\rho V_S^\dagger\right)
\leq
1-\sqrt{1-\varepsilon^2}
+\frac{\varepsilon^2}{2}
\nonumber
\leq
\frac{3}{2}\varepsilon^2.
\label{eq:final-joint-restoration-bound}
\end{equation}
\end{proof}}

\branchcolor{darkgray}{By inspecting the definitions of $\catU[n][\ell,\varepsilon]$ and
$\uni(n)$ above, \Cref{lem:dirty-vs-no-cat-unbounded} can be recast
more compactly as follows.}

Recall that $n,\ell\in\N$, $\H_{\IO}=(\mathbb{C}^{2})^{\otimes n}\text{ and }\H_{\catalyst}=(\mathbb{C}^{2})^{\otimes\ell}$.

\begin{thm}[{$C_{\catU[n][\ell,\varepsilon]}\approx_{\varepsilon^{2}}C_{\uni(n)}$
}]
\label{thm:approx-cat-restoration-is-useless}Let $0\leq\varepsilon<1$, and let
$U\in\catU[n][\ell,\varepsilon]$.
For every pure catalyst state $\eta=\den{\eta}\in\Den(\ell)$,
there exists a unitary $V_{\eta}\in\uni(n)$ such that
\begin{equation}
\label{eq:approx-cat-channel-close-unitary}
\max_{\rho\in\Den(n)}
\TD{C^{\eta}_{U}(\rho)}
{C_{V_{\eta}}(\rho)}
\leq \frac{3\varepsilon^{2}}{2}.
\end{equation}
\end{thm}

\branchcolor{black}{\begin{proof}
Fix a pure catalyst state $\eta=\den{\eta}$.
By the definition of $\catU[n][\ell,\varepsilon]$, for every pure
state $\ket{\phi}\in\H_{S}$ there exists a pure state
$\ket{\psi_{\phi}}\in\H_{S}$ such that
\begin{equation}
\label{eq:joint-restoration-for-theorem}
\TD{
U\left(\den{\phi}\otimes\eta\right)U^{\dagger}
}{
\den{\psi_{\phi}}\otimes\eta
}
\leq\varepsilon.
\end{equation}
Therefore, from \Cref{lem:dirty-vs-no-cat-unbounded}, there exists a unitary $V_{\eta}\in\uni(n)$ such that, for every $\rho\in\Den(n)$,
\begin{equation}
\TD{
\tr_{\catalyst}\left[
U\left(\rho\otimes\eta\right)U^{\dagger}
\right]
}{
V_{\eta}\rho V_{\eta}^{\dagger}
}
\leq \frac{3\varepsilon^{2}}{2}.
\label{eq:apply-joint-restoration-lemma}
\end{equation}
Using the definitions
$$
C^{\eta}_{U}(\rho)
=
\tr_{\catalyst}\left[
U\left(\rho\otimes\eta\right)U^{\dagger}
\right]
$$
and
$$
C_{V_{\eta}}(\rho)
=
V_{\eta}\rho V_{\eta}^{\dagger},
$$
we obtain
$$
\max_{\rho\in\Den(n)}
\TD{C^{\eta}_{U}(\rho)}{C_{V_{\eta}}(\rho)}
\leq
\frac{3\varepsilon^{2}}{2}.
$$
\end{proof}}

\branchcolor{darkgray}{As alluded to earlier, from \Cref{thm:approx-cat-restoration-is-useless}
above, we essentially get $C_{\catU[n][\ell,\varepsilon]}\approx_{\varepsilon^{2}}C_{\uni(n)}$,
i.e. the catalytic model of computation does not offer an advantage
in the computationally unbounded setting---even where the catalyst
may be correlated with the system.
\begin{quote}
\emph{Could it be that in the computationally bounded setting, catalysts
turn out to be advantageous?}
\end{quote}
To answer this question, we first define computationally efficient
variants of catalytic unitaries discussed above.  }

\subsection{Computational Catalytic setting}

\branchcolor{darkgray}{The most natural way to define a computational version of the catalytic
model, seems to be one where the catalyst register is initialised
in some fixed state $\left|0^{\ell}\right\rangle $ of $\ell$ qubits
and must be restored exactly. Restricting to $\left|0^{\ell}\right\rangle $
ensures that one cannot inject ``advice'' into the model, making
it potentially too powerful. Requiring that $\left|0^{\ell}\right\rangle $
is restored---but other states can be mapped arbitrarily---makes
it easier to consider to use the catalyst in a meaningful way. Finally,
the requirement that the catalyst be exactly restored, seems reasonable
in the idealised setting (e.g. where fault-tolerance allows one to
focus on logical qubits). One may justifiably argue that perhaps allowing
some error could substantially increase the power of this model. We
account for this later where we prove impossibilities (i.e. problems
one cannot solve with few/no catalysts in the computational setting),
we allow arbitrary errors---in fact, we prove space lower bounds. 

It would be helpful to first define a class of catalytic unitaries,
for a fixed $n,\ell\in\N$, that is required to restore a given catalyst
state $\left|\eta\right\rangle $ on $\ell$-qubits. }

\begin{defn}[{Catalytic unitaries given a catalyst $\catU[n][\ell][\left|\eta\right\rangle ]$}]
Let the $(n+\ell)$-qubit composite Hilbert space be denoted by ${\cal H}:=\H_{\IO}\otimes\H_{\catalyst}.$
Then, the set of \emph{catalytic unitaries} $\catU[n][\ell][\left|\eta\right\rangle ]$
given the catalyst $\left|\eta\right\rangle \in\H_{\catalyst}$, is
defined as

\[
\catU[n][\ell][\left|\eta\right\rangle ]:=\left\{ U\in{\cal U}(n+\ell):\begin{array}{r}
\forall\ \left|\phi\right\rangle \in\H_{\IO},\quad\exists\ \left|\psi\right\rangle \in{\cal H}_{\IO}\\
\text{s.t. }U\left|\phi\right\rangle \otimes\left|\eta\right\rangle =\left|\psi\right\rangle \otimes\left|\eta\right\rangle 
\end{array}\right\} .
\]
\end{defn}

\branchcolor{darkgray}{To restrict to a computational notion of catalytic computation, we
now define a class of PPT algorithms that generate the description
of catalytic unitaries that are in $\mathsf{cat}\text{-}{\cal U}_{\left|0^{\ell}\right\rangle }$.}

Fix polynomials $n,\ell:\mathbb{N}\to\mathbb{N}$. For $\secp\in\mathbb{N}$,
let $\H_{\IO,\secp}=(\mathbb{C}^{2})^{\otimes n(\secp)}$ and $\H_{\catalyst,\secp}=(\mathbb{C}^{2})^{\otimes\ell(\secp)}$
be the family of Hilbert spaces indexed by $\secp$, corresponding
to the quantum workspace and the catalyst respectively.
\begin{defn}[$\catgen{n,\ell}$]
\label{def:ql-cat} A PPT algorithm $\cA$ belongs to $\catgen{n,\ell}$,
if for all $\lambda\in\N$, running ${\cal A}$ on input $1^{\lambda}$
produces a circuit description of a projector $\Pi_{\lambda}\in\Proj(n(\lambda))$
and a catalytic unitary $U_{\lambda}\in\catU[n(\lambda)][\ell(\lambda)][\ket{0^{\ell(\lambda)}}]$,
i.e. $\circdesc(\Pi_{\lambda}),\desc(U_{\lambda})\leftarrow\mathcal{A}(1^{\lambda})$.
\end{defn}

\branchcolor{darkgray}{Using this notion of computational catalytic computation, one can
define a corresponding class of relational problems. We first setup
some notation to describe ``valid relations'' based on the sizes
of the $\IO$ register and the $\catalyst$ register. Informally,
we require that both the input and output of the relation, to fit
into the $\IO$ register. }

Using the notation setup above, we denote the set of \emph{valid relations}
as 
\[
R:=\left\{ \rel=\left\{ \rel_{\secp}\right\} _{\secp}:\begin{array}{c}
\exists m_{1},m_{2}:\N\to\N\text{ s.t. }\forall\secp\in\N,\ \ n(\secp)\geq\max(m_{1}(\secp),m_{2}(\secp))\\
\rel_{\secp}\subseteq\binset^{m_{1}(\secp)}\times\binset^{m_{2}(\secp)}
\end{array}\right\} .
\]

\branchcolor{darkgray}{We can now define the relation $\mathsf{F\widehat{cat}}$ that is
a set of relations one can solve with some computational catalytic
algorithm, i.e. an algorithm in $\mathsf{\widehat{cat}Gens}$.}
\begin{defn}[$\Fcat{n,\ell}$]
We say a relation $\rel\in R$ belongs to the complexity class $\Fcat{n,\ell}$
if there exists an algorithm ${\cal A}\in\catgen{n,\ell}$ such that
with probability at least $2/3$rds, ${\cal A}$ produces $U_{\lambda}$
satisfying the following: For all $(x_{\lambda},\ \cdot\ )\in\rel_{\lambda}$,
\end{defn}

\begin{itemize}
\item there exists $y_{\secp}$ such that $(x_{\secp},y_{\secp})\in\rel_{\secp}$
and 
\item some normalised state $\left|\varphi\right\rangle $ such that
\end{itemize}
\[
U_{\lambda}\underbrace{\ket{x_{\secp},0^{n(\secp)-|x_{\secp}|}}}_{\IO}\underbrace{\ket{0^{\ell(\lambda)}}}_{\catalyst}=\underbrace{\ket{y_{\secp}}\ket{\varphi}}_{\IO}\underbrace{\ket{0^{\ell(\lambda)}}}_{\catalyst}.
\]
 
\begin{defn}
\branchcolor{darkgray}{To clarify the notation, consider the following problem---that we
eventually use for one of our separations.}
\end{defn}

\begin{example}
\label{exa:sample-P}Given $x\in\binset^{*}$, produce
\[
y:=(x,f(g(x))).
\]
Where $f,g:\binset^{*}\to\binset^{*}$ are arbitrary functions such
that, $f=\{f_{\secp}\}_{\secp}$ where $f_{\secp}:\binset^{\secp}\to\binset^{\secp}$
and similarly $g=\{g_{\secp}\}_{\secp}$ where $g_{\secp}:\binset^{\secp}\to\binset^{\secp}$.
In our notation, this problem can be described as follows. 
\end{example}

\begin{itemize}
\item $\rel=\left\{ \rel_{\secp}\right\} _{\secp}$ where $\rel_{\secp}=\left\{ (x,y):y=(x,f(g(x))),\forall x\in\binset^{\secp}\right\} $, 
\item For this relation $m_{1}(\secp)=\secp,m_{2}(\secp)=2\secp$.
\end{itemize}
Note that this relation is only valid when the size of $\H_{\IO}$,
$n:\N\to\N$, is such that $n(\secp)\geq2\secp$.

\branchcolor{darkgray}{Using this complexity class, one can concretely ask whether catalysts
help in the computational setting: 
\[
\Fcat{n,0}\stackrel{?}{\subsetneq}\Fcat{n,\ell}
\]
for some $\ell>0$? 

In the following discussion, we work in the oracle setting and therefore
first describe oracular versions of the definitions considered so
far. These are fairly natural but we include them for completeness.}

\subsubsection{Computational catalytic computation with oracles}\label{subsec:Computational-catalytic-setting-oracles}

Let $c\in\N$ be a constant and oracles $\mathbf{O}=\left\{ O_{1},\dots,O_{c}\right\} $,
where $O_{i}$ encodes a function $f_{i}:{\cal Q}\to{\cal R}$ via
\[
O_{i}\left|q,r\right\rangle =\left|q,r\oplus f_{i}(q)\right\rangle ,\qquad q\in{\cal Q},\ r\in{\cal R}
\]
A machine $\A\in\PPT$ that outputs a circuit description of the unitaries
can do the following actions when given access to $\fat O$:
\begin{enumerate}
\item At any step, ${\cal A}$ may make queries to $\mathbf{O}$;
\item The circuit description of the unitary, $\desc(U_{\secp})\leftarrow\A^{\fat O}(1^{\secp})$,
that is output by $\cA$, can use the oracles $O_{1},\dots,O_{c}$
as a gate. Therefore, the unitary circuit $U_{\secp}$ can query $\fat O$
at any point in its  execution. 
\end{enumerate}
In the oracle setting, we allow the set of relations $R^{\fat O}$
to also depend on $\fat O$.
\begin{rem}[$\catgen{n,\ell}^{\fat O}$ and $\Fcat{n,\ell}^{\fat O}$]
 We extend the definitions of $\catgen{n,\ell}$ and $\Fcat{n,\ell}$
naturally.
\end{rem}

\begin{enumerate}
\item A PPT algorithm $\cA^{\fat O}$ belongs to $\catgen{n,\ell}^{\fat O}$,
if for all $\lambda\in\N$, running ${\cal A}^{\fat O}$ on input
$1^{\lambda}$ produces a circuit description of a projector $\Pi_{\lambda}\in\Proj(n(\lambda))$
and a catalytic unitary $U_{\lambda}\in\catU[n(\lambda)][\ell(\lambda)][\left|0^{\ell(\lambda)}\right\rangle ]$
i.e. $(\circdesc(\Pi_{\lambda}),\desc(U_{\lambda}))\leftarrow\mathcal{A}^{\fat O}(1^{\lambda})$.
\item We say a relation $\rel\in R^{\fat O}$ belongs to the complexity
class $\Fcat{n,\ell}^{\fat O}$ if there exists an algorithm ${\cal A}^{\fat O}\in\catgen{n,\ell}^{\fat O}$
such that with probability at least $2/3$rds, ${\cal A}^{\fat O}$
produces $U_{\lambda}$ satisfying the following: For all $(x_{\lambda},\ \cdot\ )\in\rel^{\fat O}_{\lambda}$,
\begin{itemize}
\item there exists $y_{\secp}$ such that $(x_{\secp},y_{\secp})\in\rel^{\fat O}_{\secp}$
and 
\item some normalised state $\left|\varphi\right\rangle $ such that
\end{itemize}
\[
U_{\lambda}\underbrace{\ket{x_{\secp},0^{n(\secp)-|x_{\secp}|}}}_{\IO}\underbrace{\ket{0^{\ell(\lambda)}}}_{\catalyst}=\underbrace{\ket{y_{\secp}}\ket{\varphi}}_{\IO}\underbrace{\ket{0^{\ell(\lambda)}}}_{\catalyst}.
\]

\end{enumerate}

\subsection{Relational Separation $\protect\Fcat{2\lambda,c\lambda}\subsetneq\protect\Fcat{2\lambda,\lambda}$
(for $c<1$)}

\branchcolor{darkgray}{With the notation in place, we now focus on establishing one of our
main results. As we alluded to earlier (in \Cref{exa:sample-P}),
the problem we consider is as follows.}

Let $\fat O:=(O_{f},O_{g})$ be oracles corresponding to the \emph{length
preserving} random functions $f,g:\binset^{*}\to\binset^{*}$.
\begin{problem}
\label{prob:the-probem}We denote by $\mathcal{P}:=\left\{ {\cal P}_{\secp}\right\} _{\secp}$
the relation where,
\begin{align*}
\P_{\secp} & :=\left\{ (x,y):x\in\binset^{\secp},y=(x,f(g(x)))\right\} .
\end{align*}
\end{problem}

\begin{rem}
We observe that the relation above, satisfies the requirements in
\Cref{subsec:Computational-catalytic-setting-oracles}. Let
\begin{itemize}
\item $n,\ell:\mathbb{N}\to\mathbb{N}$ be functions, $n(\secp)=2\secp,\ \ell(\secp)=\secp$.
\item For $\secp\in\mathbb{N}$, let $\H_{\IO,\secp}=(\mathbb{C}^{2})^{\otimes n(\secp)}$
and $\H_{\catalyst,\secp}=(\mathbb{C}^{2})^{\otimes\ell(\secp)}$
be the family of Hilbert spaces indexed by $\secp$.
\item Let $O_{f},O_{g}$ be as above. 
\end{itemize}
The relation $\P$ (see \Cref{prob:the-probem}) is a valid relation
because for each $(x,y)\in\P_{\secp}$, it holds that $\max\left(|x|,|y|\right)=2\secp=n(\secp)$.

\end{rem}

\branchcolor{darkgray}{We start by establishing that with enough catalysts, one can solve
the problem ${\cal P}$ in the computational catalytic setting.}
\begin{lem}
\label{lem:PinFCAT_3lambda} Let $n,\ell:\mathbb{N}\to\mathbb{N}$
be functions, $n(\secp)=2\secp,\ \ell(\secp)=\secp$. Then, it holds
that 
\[
{\cal P}\in\Fcat{n,\ell}^{\fat O}.
\]
\end{lem}

\branchcolor{black}{\begin{proof}
Consider the family of unitaries $\fat U=\{U_{\secp}\}_{\secp}$ that
we describe below, where each $U_{\secp}$ acts on $3\secp$. For
clarity, we denote by register $A,B,C$ the first, second and third
$\secp$-qubits the unitary acts on. We describe the action of each
unitary $U_{\secp}\in\fat U$ as follows.
\begin{align*}
\ket x_{A}\ket y_{B}\ket{0^{\secp}}_{C} & \xmapsto{O_{g,AC}}\ket x_{A}\ket y_{B}\ket{g(x)}_{C}\\
 & \xmapsto{O_{f,CB}}\ket x_{A}\ket{y\oplus f(g(x))}_{B}\ket{g(x)}_{C}\\
 & \xmapsto{O_{g,AC}}\ket x_{A}\ket{y\oplus f(g(x))}_{B}\ket{0^{\secp}}_{C}.
\end{align*}
This unitary family $\fat U$ solves the relation $\rel$. Since
one can also efficiently describe this family using $\O(\secp)$-gates
in its canonical representation---three gates corresponding to the
oracles, and $\O(\secp)$ gates to swap registers $A,B$ using 2-qubit
$\mathsf{SWAP}$ gates---this unitary family can be produced by an
algorithm ${\cal A}\in\catgen{n,\ell}^{\fat O}$.
\end{proof}

}

\branchcolor{darkgray}{Proving exclusion from, even say $\Fcat{2\lambda,0}$, turns out to
be much more challenging. In fact, we will show that even with $2\lambda+c\lambda$
many qubits (where $c<1$), one cannot solve ${\cal P}$ (on most
inputs). We first setup some notation. }

Let $\Pi_{\secp},\Pi^{\perp}_{\secp}$ be projectors on $\left(n(\secp)+\ell(\secp)\right)$-qubits
such that $\Pi+\Pi^{\perp}=\mathbb{I}$. Every $\left(n(\secp)+\ell(\secp)\right)$-qubit
unitary $U_{\secp}$ with access to $\fat O=(O_{f},O_{g})$ can be
written as
\[
U_{\secp}=V_{q(\secp)}O_{f}V_{q(\secp)-1}O_{f}\cdots V_{2}O_{f}V_{1}
\]
where
\begin{itemize}
\item $V_{i}$, for $i\in[q(\secp)]$, contains gates from $G$ (the universal
gate set), and can potentially query the oracle $O_{g}$, and
\item $V_{i}$ is the unitary transformation right before the $i$-th query
to oracle $O_{f}$.Further, one can write the unitary $U$ (we suppress
the security parameter for brevity)
\begin{align*}
U & = & V_{q}O_{f}(\Pi+\Pi^{\bot})V_{q-1}O_{f}(\Pi+\Pi^{\bot})\cdots V_{2}O_{f}(\Pi+\Pi^{\bot})V_{1}\\
 & = & \underbrace{V_{q}O_{f}V_{q-1}O_{f}\cdots O_{f}V_{3}O_{f}V_{2}O_{f}}_{A_{1}}\Pi\underbrace{V_{1}}_{Z_{1}}+\\
 &  & \underbrace{V_{q}O_{f}V_{q-1}O_{f}\cdots O_{f}V_{3}O_{f}}_{A_{2}}\Pi\underbrace{V_{2}O_{f}\Pi^{\bot}V_{1}}_{Z_{2}}+\\
 &  & \vdots\quad\quad\\
 &  & \underbrace{V_{q}O_{f}}_{A_{q-1}}\Pi\underbrace{V_{q-1}O_{f}\Pi^{\bot}\cdots O_{f}\Pi^{\bot}V_{3}O_{f}\Pi^{\bot}V_{2}O_{f}\Pi^{\bot}V_{1}}_{Z_{q-1}}+\\
 &  & \underbrace{\mathbb{I}}_{A_{q}}\underbrace{V_{q}O_{f}\Pi^{\bot}V_{q-1}O_{f}\Pi^{\bot}\cdots O_{f}\Pi^{\bot}V_{3}O_{f}\Pi^{\bot}V_{2}O_{f}\Pi^{\bot}V_{1}.}_{Z_{q}}
\end{align*}
\end{itemize}
\begin{notation}
\label{nota:succint-Ufg}One can break the unitary $U$ into branches
based on the step where $f$ was queried when the state was parallel
to $\Pi$. 
\[
U=\sum^{q-1}_{i=1}A_{i}\Pi Z_{i}+\cancelto{\mathbb{I}}{A_{q}}Z_{q}
\]
\end{notation}

\branchcolor{darkgray}{We are now ready to state one of our main technical results, namely, that even with $c\lambda$
catalysts, for $c<1$, the problem $\mathcal{P}$ cannot be solved (on most instances). This provides a sharp separation.}

\begin{lem}[$\mathcal{P}\notin\Fcat{2\secp,c\secp}$ for $c<1$]
\label{lem:sep-cat-vs-lesscat}Let $n,\ell,\ell':\N\to\N$ be the
polynomials $n(\lambda)=2\lambda$, $\ell(\lambda)=\lfloor c\cdot\lambda \rfloor$
for $c<1$. Let $\P$ be the relation as above (see \Cref{prob:the-probem}).
Given oracle access to length preserving random oracles $\fat O=\left(\text{\ensuremath{O_{f}},\ensuremath{O_{g}}}\right)$,
 it holds that $\mathcal{P}\notin\Fcat{n,\ell}^{\fat O}$. 
\end{lem}

\branchcolor{black}
{
\begin{proof}
We give the proof by showing that any circuit mapping
$\ket{x,f(g(x)),0^\ell} \mapsto \ket{x,0^{\secp}}$ has negligible success probability when $\ell(\secp)=c\secp$.

Suppose, for contradiction, that
$\mathcal{P}\in\PPTFcat[2\secp][\ell][\fat O]$.
Since the relation requires the entire $2\secp$-qubit
system register to contain $(x,f(g(x)))$, and the
catalytic register must be restored exactly, the inverse
of the corresponding unitary circuit satisfies
\begin{equation}
\label{eq:relational-cat-inverse-action}
U\ket{x,f(g(x)),0^\ell}
=
\ket{x,0^{\secp},0^\ell}
\end{equation}
for every $x$, up to a phase that does not affect any
success probability.
Reversing the circuit replaces each elementary gate by
its adjoint and preserves the oracle-query counts,
since $O_f^\dagger=O_f$ and $O_g^\dagger=O_g$.

Fix the security parameter $\secp$, write
$\ell=\ell(\secp)$, and suppress the security parameter
in the circuit notation.
Let $q-1$ denote the number of queries to $O_f$, and let
$q_g$ denote the number of queries to $O_g$.
The non-oracle gates of $U$ are fixed independently of
$f,g,x$. We first fix any oracle-independent randomness
in the choice of the circuit; the bounds below also hold
after averaging over this randomness.

Recall that
$f,g:\{0,1\}^{\secp}\to\{0,1\}^{\secp}$
are independent uniformly random functions, and
$x\in\{0,1\}^{\secp}$ is independently uniform.
Define
\begin{equation}
\label{eq:relational-cat-success}
\ssearch
:=
\E_{f,g,x}
\left\Vert
\Pisearch U\ket{x,f(g(x)),0^\ell}
\right\Vert^2,
\end{equation}
where
$\Pisearch:=\den{x}_Q\otimes\den{0^{\secp}}_R
\otimes\mathbb{I}_C$.
The registers are $Q$ ($\secp$-qubit query register),
$R$ ($\secp$-qubit response register), and
$C$ ($\ell$-qubit catalytic register).
In particular, $\Pisearch$ is independent of $f(g(x))$
when $x$ is fixed.

The circuit in \Cref{eq:relational-cat-inverse-action}
would have $\ssearch=1$.
We will bound $\ssearch$ for an arbitrary circuit $U$,
without requiring that its catalytic register be restored.

If $q-1=0$, the circuit is independent of $f$.
Conditional on $g,x$, averaging over the uniform value
$f(g(x))$ gives
\begin{align}
\E_f
\left\Vert
\Pisearch U\ket{x,f(g(x)),0^\ell}
\right\Vert^2
&=
\frac{1}{2^{\secp}}
\tr\left[
\Pisearch U
\left(\den{x}_Q\otimes\mathbb{I}_R
\otimes\den{0^\ell}_C\right)
U^\dagger
\right]
\nonumber\\
&\leq
\frac{\operatorname{rank}(\Pisearch)}{2^{\secp}}
=
\frac{2^\ell}{2^{\secp}}.
\label{eq:relational-cat-no-f-queries}
\end{align}
This is negligible for the stated choice of $\ell$.
Hence assume $q\geq 2$.

one may condition on $g,x$ when defining the branches
below, but retain the averaging over them in all
success probabilities.
Define
$\Pi:=\den{g(x)}_Q\otimes\mathbb{I}_{RC}$ and
$\Pi^\perp:=\mathbb{I}-\Pi$.
Recalling \Cref{nota:succint-Ufg}, write
$U=V_qO_fV_{q-1}O_f\cdots V_2O_fV_1$ and
\begin{equation}
\label{eq:relational-cat-decomposition}
\Pisearch U\ket{x,f(g(x)),0^\ell}
=
\Pisearch
\left(\sum_{i=1}^{q-1}A_i\Pi Z_i+Z_q\right)
\ket{x,f(g(x)),0^\ell},
\end{equation}
where
\begin{itemize}
\item
$Z_1:=V_1$ and
$Z_i:=V_iO_f\Pi^\perp Z_{i-1}$ for $2\leq i\leq q$;
\item
$A_{q-1}:=V_qO_f$ and
$A_i:=A_{i+1}V_{i+1}O_f$ for $1\leq i<q-1$.
\end{itemize}
Each $A_i$ is unitary, and each $Z_i$ is a contraction.
Conditional on $g,x$ and the restriction of $f$ to all
points other than $g(x)$, every $Z_i$ is independent
of $f(g(x))$, since every $f$-query within $Z_i$ is
preceded by $\Pi^\perp$.

With this notation in place, we proceed in two steps:
(1) upper bound $\ssearch$ in terms of the weights of
the branches $\Pi Z_i\ket{x,f(g(x)),0^\ell}$; and
(2) construct an inverter whose success probability
is bounded below by these weights.

\textbf{Step (1).}

The upper bound is computed as follows:
\begin{align}
\left\Vert
\Pisearch U\ket{x,f(g(x)),0^\ell}
\right\Vert^2
&=
\left\Vert
\Pisearch
\left(\sum_{i=1}^{q-1}A_i\Pi Z_i+Z_q\right)
\ket{x,f(g(x)),0^\ell}
\right\Vert^2
\nonumber\\
&\leq
q\Bigg(
\sum_{i=1}^{q-1}
\left\Vert
\Pisearch A_i\Pi Z_i\ket{x,f(g(x)),0^\ell}
\right\Vert^2
 +
\left\Vert
\Pisearch Z_q\ket{x,f(g(x)),0^\ell}
\right\Vert^2
\Bigg)
\nonumber\\
&\leq
q\Bigg(
\sum_{i=1}^{q-1}
\left\Vert
\Pi Z_i\ket{x,f(g(x)),0^\ell}
\right\Vert^2
+
\tr\left[
M\den{x,f(g(x)),0^\ell}
\right]
\Bigg),
\label{eq:relational-cat-branch-bound}
\end{align}
where $M:=Z_q^\dagger\Pisearch Z_q$.
The first inequality follows from Cauchy--Schwarz, while the second uses the unitarity of $A_i$ and the fact
that $\Pisearch$ is a contraction.

From the definitions of $\Pisearch$ and $Z_q$, one obtains
$0\preceq M\preceq\mathbb{I}$ and $\operatorname{rank}(M)\leq 2^\ell$. Thus one may write $M=\sum_{j=1}^{2^\ell}c_j\den{m_j}$, where $0\leq c_j\leq 1$ and the $\ket{m_j}$ are orthonormal, padding with zero eigenvalues if necessary. For fixed $g,x$ and $f$ away from $g(x)$, both $c_j$ and $\ket{m_j}$ can be chosen independently of $f(g(x))$.

Let $\bar f$ denote the restriction of $f$ to
$\{0,1\}^{\secp}\setminus\{g(x)\}$.
Since the vectors
$\{\ket{x,a,0^\ell}:a\in\{0,1\}^{\secp}\}$
are orthonormal, for each normalized $\ket{m_j}$ one obtains
$\sum_a|\braket{m_j}{x,a,0^\ell}|^2\leq 1$.
Consequently,
\begin{align}
\E_{f,g,x}
\tr\left[
M\den{x,f(g(x)),0^\ell}
\right]
&=
\E_{g,x,\bar f}
\left[
\frac{1}{2^{\secp}}
\sum_{j=1}^{2^\ell}c_j
\sum_{a\in\{0,1\}^{\secp}}
\left|\braket{m_j}{x,a,0^\ell}\right|^2
\right]
\nonumber\\
&\leq
\frac{1}{2^{\secp}}
\E_{g,x,\bar f}
\sum_{j=1}^{2^\ell}c_j
\nonumber\\
&\leq
\frac{2^\ell}{2^{\secp}}.
\label{eq:relational-cat-no-hit-bound}
\end{align}

Averaging \Cref{eq:relational-cat-branch-bound}
and using \Cref{eq:relational-cat-no-hit-bound},
we obtain
\begin{align}
\frac{\ssearch}{q}
-\frac{2^\ell}{2^{\secp}}
&\leq
\E_{f,g,x}
\sum_{i=1}^{q-1}
\left\Vert
\Pi Z_i\ket{x,f(g(x)),0^\ell}
\right\Vert^2
\label{eq:piZupperbound-cat}\\
&=
\E_{f,g,x}
\sum_{i=1}^{q-1}
\left\Vert
\bra{x,f(g(x)),0^\ell}Z_i^\dagger\Pi
\right\Vert^2
\nonumber\\
&=
\E_{f,g,x}
\sum_{i=1}^{q-1}
\tr\left[
\left(
\den{x}_Q\otimes\den{f(g(x))}_R
\otimes\den{0^\ell}_C
\right)
Z_i^\dagger\Pi Z_i
\right]
\nonumber\\
&=
\E_{g,x,\bar f}
\sum_{i=1}^{q-1}
\frac{1}{2^{\secp}}
\tr\left[
\left(
\den{x}_Q\otimes\mathbb{I}_R
\otimes\den{0^\ell}_C
\right)
Z_i^\dagger\Pi Z_i
\right]
\nonumber\\
&=
\E_{f,g,x}
\sum_{i=1}^{q-1}
\frac{1}{2^{\secp}}
\tr\left[
\left(
\den{x}_Q\otimes\mathbb{I}_R
\otimes\den{0^\ell}_C
\right)
Z_i^\dagger\Pi Z_i
\right].
\label{eq:search-upperbound-cat}
\end{align}
In the second-last equality, we averaged over $f(g(x))$,
using its conditional uniformity and the independence
of $Z_i$ from that value.
In the last equality, we restored the expectation
over all of $f$, since the expression no longer
depends on $f(g(x))$.

\textbf{Step (2).}

The inversion algorithm ${\cal I}$ receives
$z=g(x)$ and has oracle access to $O_g$.
It knows the circuit $U$, but is not given $x$
or access to the original oracle $O_f$.
It proceeds as follows:
\begin{enumerate}
\item
Sample $i$ uniformly from $\{1,\ldots,q-1\}$ and
sample an independent uniformly random function
$f:\{0,1\}^{\secp}\to\{0,1\}^{\secp}$.

\item
Create the state
$\den{z}_Q\otimes\mathbb{I}_{RC}/2^{\secp+\ell}$
by preparing $Q$ in $\ket{z}$ and choosing independent
uniform computational-basis states for $R$ and $C$.

\item
Apply $V_i^\dagger$.
For $j=i-1,i-2,\ldots,1$, apply $O_f^\dagger=O_f$,
perform the projective test $\{\Pi^\perp,\Pi\}$,
and, if the outcome is $\Pi^\perp$, apply $V_j^\dagger$.
If any test returns $\Pi$, return $\bot$ and stop.
Here $\Pi=\den{z}_Q\otimes\mathbb{I}_{RC}$.

\item
Measure $Q$ in the computational basis and return
the resulting string.
\end{enumerate}

The branch on which the algorithm does not abort has
Kraus operator 
$$
V_1^\dagger\Pi^\perp O_f^\dagger V_2^\dagger \cdots\Pi^\perp O_f^\dagger V_i^\dagger = Z_i^\dagger.
$$
Thus the algorithm implements the trace-nonincreasing
operation $\tau\mapsto Z_i^\dagger\tau Z_i$ on that
branch. All aborts count as failures. Note that this is possible because ${\cal I}$ knows $z$. Each $V_j^\dagger$ is implemented by reversing the corresponding circuit. Since $O_g^\dagger=O_g$, the total number of queries to $g$ is at most $q_g$. The independently sampled $f$ is simulated internally.

Define
$\sinv:=\Pr_{g,x,{\cal I}}[{\cal I}^{g}(g(x))=x]$.
Writing
$\Pi'_\ell:=\mathbb{I}_C-\den{0^\ell}_C$,
the success probability can be computed as
\begin{align}
\sinv
&=
\frac{1}{(q-1)2^{\secp+\ell}}
\E_{f,g,x}
\sum_{i=1}^{q-1}
\tr\left[
\left(\den{x}_Q\otimes\mathbb{I}_{RC}\right)
Z_i^\dagger\Pi Z_i
\right]
\nonumber\\
&=
\frac{1}{(q-1)2^{\secp+\ell}}
\E_{f,g,x}
\sum_{i=1}^{q-1}
\tr\left[
\left(
\den{x}_Q\otimes\mathbb{I}_R
\otimes\left(\den{0^\ell}_C+\Pi'_\ell\right)
\right)
Z_i^\dagger\Pi Z_i
\right]
\nonumber\\
&\geq
\frac{1}{(q-1)2^{\secp+\ell}}
\E_{f,g,x}
\sum_{i=1}^{q-1}
\tr\left[
\left(
\den{x}_Q\otimes\mathbb{I}_R
\otimes\den{0^\ell}_C
\right)
Z_i^\dagger\Pi Z_i
\right]
\nonumber\\
&=
\frac{1}{(q-1)2^\ell}
\E_{f,g,x}
\sum_{i=1}^{q-1}
\left\Vert
\Pi Z_i\ket{x,f(g(x)),0^\ell}
\right\Vert^2.
\label{eq:inv-success-lowerbound-cat}
\end{align}
The first equality follows from the unnormalized state
$Z_i^\dagger\Pi Z_i/2^{\secp+\ell}$ on the branch
that does not abort. The inequality drops a nonnegative term, since $\den{x}_Q\otimes\mathbb{I}_R\otimes\Pi'_\ell$
and $Z_i^\dagger\Pi Z_i$ are positive semidefinite.
The final equality is the same averaging identity
used in \Cref{eq:search-upperbound-cat}.

Combining
\Cref{eq:piZupperbound-cat,eq:search-upperbound-cat,eq:inv-success-lowerbound-cat},
we obtain
\begin{align}
2^\ell(q-1)\sinv
&\geq
\frac{\ssearch}{q}-\frac{2^\ell}{2^{\secp}},
\nonumber\\
\sinv
&\geq
\frac{\ssearch}{q(q-1)2^\ell}
-
\frac{1}{(q-1)2^{\secp}}.
\label{eq:relational-cat-inversion-reduction}
\end{align}

On the other hand, the random-function inversion bound
obtained gives
\begin{equation}
\label{eq:relational-cat-inversion-upper-bound}
\sinv
\leq
\frac{(2q_g+1)^2}{2^{\secp}}.
\end{equation}
Therefore,
\begin{align}
\ssearch
&\leq
q\,2^\ell(q-1)\sinv
+
q\,\frac{2^\ell}{2^{\secp}}
\nonumber\\
&\leq
q\left(1+(q-1)(2q_g+1)^2\right)
2^{\ell-\secp}.
\label{eq:relational-cat-success-upper-bound}
\end{align}
This bound also holds when $q=1$, by \Cref{eq:relational-cat-no-f-queries}. For polynomially bounded $q,q_g$, it is negligible whenever $\secp-\ell(\secp)=\omega(\log\secp)$. In particular, for $\ell(\secp)=\lfloor c\secp\rfloor$ with fixed $c<1$, it is bounded by a polynomial in $\secp$ times $2^{-(1-c)\secp}$, which is negligible.
\end{proof}}

\begin{rem}[Queries made by $\A$ don't matter]
The queries to $\fat O$ made $\PPT$ algorithm $\A^{\fat O}$ only
affects the success probability negligibly. The algorithm $\A^{\fat O}$
can only query polynomially many points of the random oracle $\fat O$,
while we need the output of the unitary to satisfy the relation on
all inputs.
\end{rem}

\branchcolor{darkgray}{It immediately follows that one can solve a problem with $\lambda$ many catalysts but not with $0$ catalysts. We separate out the $0$ catalyst case because of its implications in separating computational catalytic ergotropy and computational ergotropy.}

\begin{lem}[$\mathcal{P}\notin\Fcat{2\secp,0}$]
\label{lem:sep-cat-vs-nocat-evaluation} Let $n,\ell:\N\to\N$ be
the polynomials $n(\lambda)=2\lambda$ and $\ell(\lambda)=0$. Let
$\P$ be the relation as above (see \Cref{prob:the-probem}). Given
oracle access to length preserving random oracles $\fat O=\left(\text{\ensuremath{O_{f}},\ensuremath{O_{g}}}\right)$,
 it holds that $\mathcal{P}\notin\Fcat{n,\ell}^{\fat O}$.
\end{lem}

\branchcolor{black}{\begin{proof}
Set $\ell(\secp)=0$ in \Cref{lem:sep-cat-vs-lesscat}. For $q, q_g$ polynomially bounded, we have that $\ssearch$ is negligible. 

\end{proof}
}

\branchcolor{darkgray}{Combining the results so far, we obtain our final result here.}
\begin{thm}[{$\PPTFcat[2\secp][c\lambda]\subsetneq\PPTFcat[2\secp][\secp]$}]
\label{cor:sep-Fcat-2lambda-3-lambda}Let $n,\ell,\ell':\N\to\N$
be the polynomials $n(\lambda)=2\lambda$, $\ell'(\lambda)=\lambda$
and $\ell(\lambda)=c\cdot\secp$ for $c<1$. In the random oracle
model, it holds that 
\[
\Fcat{n,\ell}^{\fat O}\subsetneq\Fcat{n,\ell'}^{\fat O}.
\]
\end{thm}

\branchcolor{black}{From \Cref{lem:PinFCAT_3lambda} and \Cref{lem:sep-cat-vs-lesscat},
we know that $\mathcal{P}\in\PPTFcat[n][\ell'][\fat O]$ and $\mathcal{P}\notin\PPTFcat[n][\ell][\fat O]$.
As a direct consequence we have that $\PPTFcat[n][\ell][\fat O]\subsetneq\PPTFcat[n][\ell'][\fat O]$.}

\branchcolor{darkgray}{We conclude with the following remark.}
\begin{rem}
We emphasise that the proof of \Cref{lem:sep-cat-vs-lesscat} shows
something stronger---it establishes that even with $2\lambda+c\lambda$-many
qubits, the problem ${\cal P}$ remains intractable. Thus, even allowing
for arbitrary errors in the definition of computational catalytic
computation, preserves the separation.
\end{rem}

\subsection{Decision problem separation $\protect\cat{2\lambda,c\lambda}\subsetneq\protect\cat{2\lambda,\lambda}$
(for $c<1$)}

\branchcolor{darkgray}{Separations in relational problems are often easier to establish than
decision problems. In fact, we have reason to believe that in many
cases, decision separations relative to a random oracle are impossible---such
as the well known Aaronson Ambainis conjecture. Furthermore, decision
problems turn out to be much more relevant for cryptographic applications.
For instance, indistinguishability of encryptions may be viewed as
a decision problem.

Thus, here we show how to establish a separation between computational
catalytic computations using a decision problem. We start by defining
the decision variant $\mathsf{\widehat{cat}}$ of the relational problem
class $\mathsf{F\widehat{cat}}$. Here, we directly introduce the
oracle version. }

Denote a (promise) decision problem ${\cal P}$ by the sets $(\mathsf{YES},\mathsf{NO})$
where $\mathsf{YES}\subseteq\{0,1\}^{*}$ and $\mathsf{NO}\subseteq\{0,1\}^{*}$.
When ${\cal P}$ depends on the oracle $\fat O$, we denote it by
${\cal P}^{\fat O}$. Fix $n,\ell:\N\to\N$. 
\begin{defn}[$\cat{n,\ell}$]
We say a problem ${\cal P}^{\fat O}$ is in $\cat{n,\ell}$ if there
exists an algorithm ${\cal A}^{\fat O}\in\catgen{n,\ell}^{\fat O}$
such that 
\begin{itemize}
\item Completeness: For all $x\in\mathsf{YES}^{\fat O}$ such that $|x|=n(\lambda)$,
\[
\E_{\circdesc(\Pi,U)\leftarrow{\cal A}^{\fat O}(1^{\lambda})}\left\Vert \Pi U\ket{x,0^{n(\lambda)-|x|},0^{\ell(\lambda)}}\right\Vert^2 \ge2/3
\]
\item Soundness: For all $y\in\mathsf{NO}^{\fat O}$ such that 
\[
\E_{\circdesc(\Pi,U)\leftarrow{\cal A}^{\fat O}(1^{\lambda})}\left\Vert \Pi U\ket{x,0^{n(\lambda)-|x|},0^{\ell(\lambda)}}\right\Vert^2 \le 1/3
\]
\end{itemize}
\end{defn}

\branchcolor{darkgray}{The decision problem we consider, is closely related to the relational
problem we considered. }
\begin{problem}
\label{prob:the-probem-1}Let $\fat O=(O_{g},O_{f})$ be the oracles
corresponding to length preserving random functions $f,g$. The decision
problem ${\cal P}^{\fat O}$ is defined as follows: 
\[
\mathsf{YES}^{\fat O}:=\{(x,f(g(x)):x\in\{0,1\}^{*}\}
\]
and 
\[
\mathsf{NO}^{\fat O}:=\{0,1\}^{*}\backslash\mathsf{YES}^{\fat O}.
\]
\end{problem}

\branchcolor{darkgray}{While part of the argument we used to establish \Cref{lem:sep-cat-vs-lesscat} goes through, it fails when one tries to bound the contribution of the term $A_{q}\Pi Z_{q}\left|x,f(g(x))\right\rangle $. This has to do with the fact that for search, the projector $\Pisearch$ is rank-1 while for a decision problem, we have no such guarantee on the corresponding projector $\Pidec$. We now prove the decision analogue of Lemma \ref{lem:sep-cat-vs-lesscat}, where the catalyst register has $\ell(\lambda)= c\cdot\lambda$ qubits.}

\begin{lem}[$(x,f(g(x))\approx(x,r)$ with $\ell$ catalysts]
\label{lem:decision-cat-separation-robust}Let $n,\ell:\N\to\N$
be the polynomials $n(\lambda)=2\lambda$, $\ell(\lambda)=c\cdot\lambda$
for $c<1$. Given oracle access to $\fat O=(O_{g},O_{f})$, it holds
that for every ${\cal A}\in\cat{n,\ell}$ there is a negligible function
$\negl$ such that 
\[
\left|\E_{\substack{\desc(\Pidec_{\lambda},U_{\lambda})\leftarrow{\cal A}^{\fat O}\\
\fat O
}
}\norm{\Pi^{\mathsf{dec}}_{\secp}\left(U_{\secp}\ket{x,f(g(x))}\right)}-\E_{\substack{\desc(\Pidec_{\lambda},U_{\lambda})\leftarrow{\cal A}^{\fat O}\\
\fat O
}
}\norm{\Pi^{\mathsf{dec}}_{\secp}\left(U_{\secp}\ket{x,r}\right)}\right|\leq\negl(\secp).
\]
\end{lem}

\branchcolor{black}{
\begin{proof}
Similar to Lemma \ref{lem:sep-cat-vs-lesscat}, we show that if one can distinguish $(x,f(g(x)))$ from $(x,r)$ with probability $\sdec$ then there is an inversion algorithm that inverts $g(x)$ with success probability at least 
$$
\sinv\geq \dfrac{\sdec^2}{16(q-1)^2 \times 2^\ell},
$$ 
where $q-1\geq 1$ is the number of queries to $O_f$.

Fix the security parameter $\secp$, write
$\ell=\ell(\secp)$, and suppress the security parameter
in the circuit notation. Let $q_g$ denote the number
of queries to $O_g$. The circuit $U$ acts on registers
$Q,R,C$ containing $\secp,\secp,\ell$ qubits,
respectively, with $C$ initially in $\ket{0^\ell}$.
We implicitly extend $\Pidec$ by the identity on $C$.

The non-oracle gates of $U$ and the projector $\Pidec$
are fixed independently of $f,g,x,r$.
Recall that the functions
$f,g:\{0,1\}^{\secp}\to\{0,1\}^{\secp}$
are independent and uniformly random, and
$x,r\in\{0,1\}^{\secp}$ are independent uniform strings,
independent of $f,g$. All expectations below include $f,g,x,r$, as applicable.

If $q-1=0$, then conditional on $g,x$, the circuit is
independent of the uniform value $f(g(x))$. The two acceptance probabilities are therefore equal. Hence assume $q\geq 2$.

Define
\begin{equation}
\sdec:=
\Big|
\E_{f,g,x}
\left\Vert
\Pidec U\ket{x,f(g(x)),0^\ell}
\right\Vert^2
-
\E_{f,g,x,r}
\left\Vert
\Pidec U\ket{x,r,0^\ell}
\right\Vert^2
\Big|.
\label{eq:decision-cat-advantage}
\end{equation}
We may condition on $g,x$ when defining the branches
below, but retain the averaging over them in all
success probabilities.

Recalling the notation from \Cref{nota:succint-Ufg},
write $U=V_qO_fV_{q-1}O_f\cdots V_2O_fV_1$ and
\begin{equation}
\label{eq:decision-cat-decomposition}
U=\sum_{i=1}^{q-1}A_i\Pi Z_i+Z_q,
\end{equation}
where
\begin{itemize}
\item
$\Pi:=\den{g(x)}_Q\otimes\mathbb{I}_{RC}$ and
$\Pi^\perp:=\mathbb{I}-\Pi$;
\item
$Z_1:=V_1$ and
$Z_i:=V_iO_f\Pi^\perp Z_{i-1}$ for $2\leq i\leq q$;
\item
$A_{q-1}:=V_qO_f$ and
$A_i:=A_{i+1}V_{i+1}O_f$ for $1\leq i<q-1$.
\end{itemize}
Each $A_i$ is unitary, and each $Z_i$ is a contraction.

We proceed in two steps.
\begin{enumerate}
\item
We prove that
\begin{equation}
\label{eq:decision-cat-step-one}
\sdec
\leq
4\sum_{i=1}^{q-1}
\E\left\Vert
\Pi Z_i\ket{x,f(g(x)),0^\ell}
\right\Vert.
\end{equation}

\item
One can construct an inversion algorithm ${\cal I}$
whose success probability
$\sinv:=\Pr_{g,x,{\cal I}}[{\cal I}^{g}(g(x))=x]$
satisfies
\begin{align}
(q-1)^2 2^\ell\sinv
&\geq
(q-1)\sum_{i=1}^{q-1}
\E\left\Vert
\Pi Z_i\ket{x,f(g(x)),0^\ell}
\right\Vert^2
\nonumber\\
&\geq
\left(
\sum_{i=1}^{q-1}
\E\left\Vert
\Pi Z_i\ket{x,f(g(x)),0^\ell}
\right\Vert
\right)^2.
\label{eq:decision-cat-step-two}
\end{align}
\end{enumerate}

Step 2 uses exactly the same inversion algorithm as in \Cref{lem:sep-cat-vs-lesscat}. Given $z=g(x)$, the inverter samples $i$ uniformly from
$\{1,\ldots,q-1\}$, samples an independent random function
$f$, prepares
$\den{z}_Q\otimes\mathbb{I}_{RC}/2^{\secp+\ell}$,
and implements the branch $Z_i^\dagger$ before measuring $Q$.

To implement this branch, it runs the gates defining $Z_i$
in reverse order and implements each occurrence of
$\Pi^\perp$ by the projective test
$\{\Pi^\perp,\Pi\}$, aborting on outcome $\Pi$.
The branch on which it does not abort has Kraus operator
$Z_i^\dagger$.
These tests are available because the inverter knows $z$.
The inverter makes at most $q_g$ queries to $g$, since
$O_g^\dagger=O_g$.
Its independently sampled $f$ can be simulated internally:
the inversion bound restricts the queries to $g$, not
the inverter's other computation or workspace.

For completeness, define
\begin{align}
P_x
&:=\den{x}_Q\otimes\mathbb{I}_{RC},
\nonumber\\
P_{x,0}
&:=\den{x}_Q\otimes\mathbb{I}_R
   \otimes\den{0^\ell}_C.
\label{eq:decision-cat-inverter-projectors}
\end{align}
The success probability, including the possibility of abort,
satisfies
\begin{align}
\sinv
&=
\frac{1}{(q-1)2^{\secp+\ell}}
\sum_{i=1}^{q-1}
\E_{f,g,x}
\operatorname{tr}\left(P_xZ_i^\dagger\Pi Z_i\right)
\nonumber\\
&\geq
\frac{1}{(q-1)2^{\secp+\ell}}
\sum_{i=1}^{q-1}
\E_{f,g,x}
\operatorname{tr}\left(P_{x,0}Z_i^\dagger\Pi Z_i\right)
\nonumber\\
&=
\frac{1}{(q-1)2^\ell}
\sum_{i=1}^{q-1}
\E_{f,g,x}
\left\Vert
\Pi Z_i\ket{x,f(g(x)),0^\ell}
\right\Vert^2.
\label{eq:decision-cat-inverter-bound}
\end{align}
The inequality follows from $P_x\succeq P_{x,0}$ and
$Z_i^\dagger\Pi Z_i\succeq 0$.
The final equality follows by conditioning on $g,x$ and
the restriction of $f$ to all points other than $g(x)$.
Under this conditioning, $Z_i$ is independent of the
uniform value $f(g(x))$, because each $f$-query within
$Z_i$ is preceded by $\Pi^\perp$.
Averaging over that value gives the factor $2^{-\secp}$.
Cauchy--Schwarz, together with
$(\E X)^2\leq\E X^2$, then proves
\Cref{eq:decision-cat-step-two}.

It remains to establish Step 1.
For any vectors $u,v$ with $\|u\|\leq 1$ and
$\|u+v\|\leq 1$, the reverse triangle inequality gives
\begin{align}
\left|\|u+v\|^2-\|u\|^2\right|
&=
\left(\|u+v\|+\|u\|\right)
\left|\|u+v\|-\|u\|\right|
\nonumber\\
&\leq 2\|v\|.
\label{eq:decision-single-vector-comparison}
\end{align}
Consequently, if random vectors $u_1,v_1,u_2,v_2$
satisfy $\|u_j\|\leq 1$ and $\|u_j+v_j\|\leq 1$
almost surely for $j\in\{1,2\}$, and
$\E\|u_1\|^2=\E\|u_2\|^2$, then
\begin{align}
\left|
\E\left(\|u_1+v_1\|^2-\|u_2+v_2\|^2\right)
\right|
&\leq
\E\left|\|u_1+v_1\|^2-\|u_1\|^2\right|
\nonumber\\
&\quad+
\E\left|\|u_2+v_2\|^2-\|u_2\|^2\right|
\nonumber\\
&\leq
2\left(\E\|v_1\|+\E\|v_2\|\right).
\label{eq:decision-general-vector-comparison}
\end{align}
This estimate does not require $\|v_j\|\leq 1$ or
negligibility of either expectation.

We apply \Cref{eq:decision-general-vector-comparison} with
\begin{align}
\ket{v_1}
&:=
\Pidec\sum_{i=1}^{q-1}A_i\Pi Z_i
\ket{x,f(g(x)),0^\ell},
\nonumber\\
\ket{u_1}
&:=
\Pidec Z_q\ket{x,f(g(x)),0^\ell},
\nonumber\\
\ket{v_2}
&:=
\Pidec\sum_{i=1}^{q-1}A_i\Pi Z_i
\ket{x,r,0^\ell},
\nonumber\\
\ket{u_2}
&:=
\Pidec Z_q\ket{x,r,0^\ell}.
\label{eq:decision-cat-comparison-vectors}
\end{align}
Since $Z_q$ is a contraction,
$\|u_1\|,\|u_2\|\leq 1$.
Moreover, $u_j+v_j$ is the projection of a unit vector,
so $\|u_j+v_j\|\leq 1$ for $j\in\{1,2\}$.

Conditional on $g,x$ and $f$ away from $g(x)$,
the operator $Z_q$ is fixed, while both $f(g(x))$ and $r$
are uniform strings.
Since $\Pidec$ is independent of the oracles and the
sampled strings, we obtain
\begin{equation}
\label{eq:decision-cat-no-hit-cancellation}
\E\|u_1\|^2=\E\|u_2\|^2.
\end{equation}

The triangle inequality, together with the unitarity of
$A_i$ and the fact that $\Pidec$ is a contraction, gives
\begin{align}
\E\|v_1\|
&=
\E\left\Vert
\Pidec\sum_{i=1}^{q-1}A_i\Pi Z_i
\ket{x,f(g(x)),0^\ell}
\right\Vert
\nonumber\\
&\leq
\sum_{i=1}^{q-1}
\E\left\Vert
\Pidec A_i\Pi Z_i\ket{x,f(g(x)),0^\ell}
\right\Vert
\nonumber\\
&\leq
\sum_{i=1}^{q-1}
\E\left\Vert
\Pi Z_i\ket{x,f(g(x)),0^\ell}
\right\Vert.
\label{eq:decision-cat-first-hit-bound}
\end{align}
Similarly,
\begin{align}
\E\|v_2\|
&\leq
\sum_{i=1}^{q-1}
\E\left\Vert
\Pidec A_i\Pi Z_i\ket{x,r,0^\ell}
\right\Vert
\nonumber\\
&\leq
\sum_{i=1}^{q-1}
\E\left\Vert
\Pi Z_i\ket{x,r,0^\ell}
\right\Vert
\nonumber\\
&=
\sum_{i=1}^{q-1}
\E\left\Vert
\Pi Z_i\ket{x,f(g(x)),0^\ell}
\right\Vert.
\label{eq:decision-cat-second-hit-bound}
\end{align}

Now using \Cref{eq:decision-general-vector-comparison},
we conclude that
\begin{align}
\sdec
&=
\left|
\E\left(
\|u_1+v_1\|^2-\|u_2+v_2\|^2
\right)
\right|
\nonumber\\
&\leq
2\left(\E\|v_1\|+\E\|v_2\|\right)
\nonumber\\
&\leq
4\sum_{i=1}^{q-1}
\E\left\Vert
\Pi Z_i\ket{x,f(g(x)),0^\ell}
\right\Vert.
\label{eq:decision-cat-quantitative-comparison}
\end{align}
This proves Step 1.

Combining Steps 1 and 2 gives
\begin{equation}
\label{eq:decision-cat-inversion-lower-bound}
\sinv
\geq
\frac{\sdec^2}{16(q-1)^2 \times 2^\ell}.
\end{equation}
On the other hand, the random-function inversion bound gives
\begin{equation}
\label{eq:decision-cat-inversion-upper-bound}
\sinv
\leq
\frac{(2q_g+1)^2}{2^{\secp}}.
\end{equation}
Therefore,
\begin{equation}
\label{eq:decision-cat-final-bound}
\sdec
\leq
4(q-1)(2q_g+1)\,2^{-(\secp-\ell)/2}.
\end{equation}
This is negligible for polynomially bounded $q,q_g$
whenever $\secp-\ell(\secp)=\omega(\log\secp)$. In particular, this holds for $\ell(\secp)=\lfloor c\secp\rfloor$ with any fixed $0\leq c<1$. Note that this does not require the catalytic register to be restored.
\end{proof}
}

\branchcolor{darkgray}{As before, combining the lemma and using essentially the same algorithm
that solved the relational problem, we obtain our main result for
a decision problem.}
\begin{thm}[$\cat{2\lambda,c\lambda}^{\fat O}\subsetneq\cat{2\lambda,\lambda}^{\fat O}$
for $c<1$]
\label{thm:sep-cat-vs-nocat-dec}Let $n,\ell,\ell':\N\to\N$ be
the polynomials $n(\lambda)=2\lambda$, $\ell'(\lambda)=\lambda$
and $\ell(\lambda)=c\cdot\secp$ for $c<1$. In the random oracle
model, it holds that 
\[
\cat{n,\ell}^{\fat O}\subsetneq\cat{n,\ell'}^{\fat O}.
\]
\end{thm}

\branchcolor{black}{\begin{proof}
Since from \Cref{lem:decision-cat-separation-robust} it holds that $(x,r)$ and $(x,f(g(x)))$ are indistinguishable with less than $\ell(\secp)=\secp$ catalysts, for the decision problem $\mathcal{P}$ (see \Cref{prob:the-probem-1}) it follows that: $\mathcal{P} \notin \cat{2\lambda,c\lambda}^{\fat O}$ while $\mathcal{P} \in \cat{2\lambda,\lambda}^{\fat O}$
\end{proof}
}

\section{Catalytic Work Extraction}\label{sec:Computational-Catalytic-Work}

\branchcolor{darkgray}{
Since the definition of computational work extraction in \Cref{subsec:The-access-model} appears to be too restrictive (see \Cref{prob:in-place-permutations-are-hard}), one can relax the definition to allow for catalytic computations in the work extraction process. In this section we allow the use of catalysts for work extraction and make statements about catalytic ergotropy and its computational variant, separations between computational ergotropy and computational catalytic ergotropy and catalytic work potential of Hamiltonians.
}

We start by defining catalytic ergotropy as follows.

\begin{defn}[$\caterg_{\rho,H}(\ell)$, $\caterg_{\rho,H}$]
Given a quantum state (complete description) $\rho\in\dens(n)$ and
a Hamiltonian $H\in\Ham(n)$, the catalytic ergotropy of the state
$\rho$ with respect to the Hamiltonian $H$ using $\ell$-catalysts
is given by
\[
\caterg_{\rho,h}(\ell):=\max_{\substack{\substack{\left|\eta\right\rangle \in{\cal H}_{\catalyst}\\
U\in\catU[n][\ell][\left|\eta\right\rangle ]
}
}
}\tr\left[H\otimes\mathbb{I}_{2^{\ell}}\left(\rho\otimes\den{\eta}^{\otimes\ell}-U\left(\rho\otimes\den{\eta}^{\otimes\ell}\right)U^{\dagger}\right)\right]
\]
and with arbitrarily many catalysts is given by
\[
\caterg_{\rho,h}:=\max_{\ell\geq0}\caterg_{\rho,h}(\ell).
\]
\end{defn}

\branchcolor{darkgray}{
We introduce the following notation to be consistent with \Cref{sec:comp-work-ex}. 
}

\begin{notation}
Let $\fat{\rho}=\{\rho_{\secp}\}_{\secp}$ be a family of states and
$\fat H=\{H_{\secp}\}_{\secp}$ be a family of Hamiltonians. We use
these alternate notations to represent the catalytic ergotropy for
families
\end{notation}

\begin{itemize}
\item $\caterg_{\fatrhoH}(\secp,\ell(\secp))$ $:=\caterg_{\rho^{\secp},h^{\secp}}(\ell(\secp))$
and
\item $\caterg_{\fatrhoH}(\secp)$ $:=\caterg_{\rho^{\secp},h^{\secp}}$.
\end{itemize}
\branchcolor{darkgray}{
In order define catalytic ergotropy, one has to extend the $n$-qubit Hamiltonian to define it on the catalytic space as well. 
It is crucial to ensure that the extended Hamiltonian does not trivially lead to higher work extraction.

As a consequence of \Cref{lem:dirty-vs-no-cat-unbounded},in the following theorem we show that the value of catalytic ergotropy  coincides with the value of ergotropy without catalysts in the information theoretic setting.

}
\begin{thm}[Equivalence $\erg$ and $\caterg$]
 Let $n\in\N$ denote the size of the system. For state $\rho\in\Den(n)$
and Hamiltonian $H\in\Ham(n)$ , the maximum extractable work remains
unchanged even with the assistance of a catalytic system, i.e. 
\begin{flalign*}
\erg_{\rho,H} & =\caterg_{\rho,H} & \forall\rho\in\dens(n),\forall h\in\hams(n),\forall n.
\end{flalign*}
\end{thm}

\branchcolor{black}{\begin{proof}
Since for unitaries $U\in\catU[n][\ell][\left|\eta\right\rangle ]$,
it holds that $U\left(\rho\otimes\den{\eta}\right)U^{\dagger}=\sigma\otimes\den{\eta}$,
one can write the resultant state as
\[
U\left(\rho\otimes\den{\eta}\right)U^{\dagger}=\tr_{\catalyst}\left[U\left(\rho\otimes\den{\eta}\right)U^{\dagger}\right]\otimes\tr_{\IO}\left[U\left(\rho\otimes\den{\eta}\right)U^{\dagger}\right]
\]
Using this in this in the definition of catalytic ergotropy, one can
write
\begin{align*}
\caterg_{\rho,h}(\ell) & =\max_{\substack{U\in\catU[n][\ell]}
}\tr\left[H\otimes\mathbb{I}_{2^{\ell}}\left(\rho\otimes\den{\eta}-U\left(\rho\otimes\den{\eta}\right)U^{\dagger}\right)\right]\\
 & =\max_{\substack{U\in\catU[n][\ell]}
}\tr\left[\left(H\rho\right)\otimes\left(\I_{2^{\ell}}\den{\eta}\right)\right]\\ & \qquad \qquad \qquad \qquad -\tr\left[\left(H\cdot\tr_{\catalyst}\left[U\left(\rho\otimes\den{\eta}\right)U^{\dagger}\right]\right)  \times\left(\I_{2^{\ell}}\cdot\tr_{\IO}\left[U\left(\rho\otimes\den{\eta}\right)U^{\dagger}\right]\right)\right]\\
 & =\max_{\substack{U\in\catU[n][\ell]}
}\tr\left[\left(H\rho\right)\times1\right]-\tr\left[\left(H\cdot\tr_{\catalyst}\left[U\left(\rho\otimes\den{\eta}\right)U^{\dagger}\right]\right)\times1\right]
\end{align*}
Since for all $n,\ell\in\N$ we have $C_{\uni(n)}=C_{\catU[n][\ell]}$
(from \Cref{claim:exactCatalystUseless,thm:approx-cat-restoration-is-useless}),
it holds that for every $U\in\catU[n][\ell]$ there exists a unitary
$V\in\uni(n)$ such that,
\[
\tr\left[\left(H\rho\right)\right]-\tr\left[\left(H\cdot\tr_{\catalyst}\left[U\left(\rho\otimes\den{\eta}\right)U^{\dagger}\right]\right)\right]=\tr\left[\left(H\rho\right)\right]-\tr\left[\left(H\cdot V\rho V^{\dagger}\right)\right].
\]
As a consequence one it holds that $\erg_{\rho,H}=\caterg_{\rho,H}(\ell)$
for all $\ell,$ states and Hamiltonians. Therefore one obtains that
$\erg_{\rho,H}=\caterg_{\rho,H}$. 
\end{proof}
}

\subsection{Computational Catalytic Work Extraction}\label{subsec:cat-erg-hat}

\branchcolor{darkgray}{
From the theorem above it holds that 
catalytic ergotropy, retains the meaning of the notion of ergotropy
as defined in \cite{Lenard:1978thm}, while allowing one to
use catalytic space.

Analogous to the defintion of computational ergotropy, one can naturally extend the notion of catalytic ergotropy to the computationally bounded setting as follows.

}

Let $n,\ell:\mathbb{N}\to\mathbb{N}$ be functions denoting the size
of the quantum system and the catalyst registers respectively. For
$\secp\in\mathbb{N}$, let $\fat{\rho}=\{\rho_{\secp}\}_{\secp}$
and $\fat H=\{H_{\secp}\}_{\secp}$ denote a family of states and
Hamiltonians, such that $\rho_{\secp}\in\Den(n(\secp))$, $H_{\secp}\in\Ham(n(\secp))$. 
\begin{defn}[Work extracted by an algorithm]
\label{def:caterghat} Let $\fatrhoH$ be as above, and let $\A$
be a quantum oracle algorithm. We define the work extracted by the
algorithm $\A$, using $\ell$ catalyst qubits, from the state $\rho_{\secp}$
with respect to the Hamiltonian $H_{\secp}$ as

\begin{align*}
\W^{(\ell)}_{\fat{\rho},\fat H,\A}(\secp) & :=\mathbb{\E}\Biggr[\tr\left[H_{\secp}\otimes\I_{2^{\ell(\secp)}}\left(\rho_{\secp}\otimes\den{0^{\ell(\secp)}}-U_{\secp}\left(\rho_{\secp}\otimes\den{0^{\ell(\secp)}}\right)U^{\dagger}_{\secp}\right)\right]:\\
 & \qquad\qquad\qquad\qquad\qquad\qquad\qquad\qquad\qquad\begin{array}{r}
\circdesc(U_{\secp})\leftarrow\A^{O_{\fat{\rho}},\hamdesc}(1^{\secp}),\\
U_{\secp}\in\catU[n(\secp)][\ell(\secp)][\left|0^{\ell(\lambda)}\right\rangle ]
\end{array}\Biggr]
\end{align*}
where the expectation is over the internal randomness of the algorithms,
and oracles.\footnote{We had to do this instead of directly using $\cat{n,\ell}$ because
we wanted to explicitly give the algorithm access to $O_{\fat{\rho}}$.}
\end{defn}

\branchcolor{darkgray}{Similar to computational ergotropy, one can define computational
catalytic ergotropy for distributions over states and Hamiltonians.}

Let ${\cal D}$ denote a probability distribution over families of
states $\{\fat{\rho}_{i}\}_{i}$ and families of Hamiltonian $\{\fat H_{i}\}_{i}$.
We use $\fatrhoH\leftarrow{\cal D}$ to denote the result of sampling
a pair of state and Hamiltonian families from ${\cal D}$.
\begin{defn}[Computational Catalytic Ergotropy ($\caterghat_{{\cal D},\ell}$)]
\label{def:computational-caterg-D-ell}Let $\ub,\lb:\mathbb{N}\to\mathbb{R}$
be functions, and ${\cal D}$ be a distribution over families of states
and Hamiltonians (as described above). We say that the \emph{computational
catalytic ergotropy using $\ell$ catalyst qubits and for the distribution
${\cal D}$}, denoted by $\caterghat_{{\cal D},\ell}$, is $(\lb,\ub)$
bounded if
\begin{align*}
\forall\quad{\cal A}\in\PPT,\text{ it holds that }\E_{\fatrhoH\leftarrow{\cal D}}\left[{\cal W}^{(\ell)}_{\fatrhoH,{\cal A}}\right]\fnleq & \ub,\text{ and }\\
\exists\quad{\cal A}\in\PPT\ \text{s.t.}\quad\E_{\fatrhoH\leftarrow{\cal D}}\left[{\cal W}^{(\ell)}_{\fatrhoH,{\cal A}}\right]\fngeq & \lb.
\end{align*}
We say that the \emph{computational catalytic ergotropy for the distribution
${\cal D}$}, denoted by $\caterghat_{{\cal D}}$, is $(\lb,\ub)$
bounded if
\begin{align*}
\forall\quad{\cal A}\in{\cal C},\text{and for all polynomials }\ell\text{, it holds that }\E_{\fatrhoH\leftarrow{\cal D}}\left[{\cal W}^{(\ell)}_{\fatrhoH,{\cal A}}\right]\fnleq & \ub,\text{ and }\\
\text{there exists }{\cal A}\in{\PPT},\text{and polynomial }\ell\ \text{s.t.}\quad\E_{\fatrhoH\leftarrow{\cal D}}\left[{\cal W}^{(\ell)}_{\fatrhoH,{\cal A}}\right]\fngeq & \lb.
\end{align*}

\end{defn}

\subsection{Separating $\protect\erghat_{{\cal D}}$ and $\protect\caterghat_{{\cal D}}$
in the random oracle model}
\branchcolor{darkgray}{
Clearly, since the catalytic model of computation allows for more operations in the computationally bounded setting, one can potentially extract more work using catalysts. We give the explicit distributions over families of states and Hamiltonians for which one attains maximal separation between ergotropy and catalytic ergotropy in the computationally bounded setting, relative to a random oracle.  

Before we show the separation, we make the following remarks.
}

\begin{rem}[Catalytic ergotropy relative to oracles]
\label{rem:caterg-distributions} One can naturally extend the definitions
above to the oracle setting.
\begin{enumerate}
\item The probability distribution ${\cal D}$ is over the family of states,
Hamiltonians $\fatrhoH$ and the oracle $\fat O$.
\item For a particular family of states, Hamiltonians and oracle, $(\fatrhoH,\fat O)\leftarrow{\cal D}$,
we define
\begin{align*}
\W^{(\ell)}_{\fat{\rho},\fat H,\A^{\fat O}}(\secp) & :=\mathbb{\E}\Biggr[\tr\left[H_{\secp}\otimes\I_{2^{\ell(\secp)}}\left(\rho_{\secp}\otimes\den{0^{\ell(\secp)}}-U_{\secp}\left(\rho_{\secp}\otimes\den{0^{\ell(\secp)}}\right)U^{\dagger}_{\secp}\right)\right],\\
 & \qquad\qquad\qquad\qquad\qquad\qquad\qquad\qquad\qquad:\begin{array}{r}
\circdesc(U_{\secp})\leftarrow\A^{O_{\fat{\rho}},\hamdesc,\fat O}(1^{\secp}),\\
U_{\secp}\in\catU[n(\secp)][\ell(\secp)][\left|0^{\ell(\lambda)}\right\rangle ]
\end{array}\Biggr]
\end{align*}
\item We say $\caterghat_{{\cal D}}$, is $(\lb,\ub)$ bounded if
\begin{align*}
\forall\quad{\cal A}\in\PPT,\text{and for all polynomials }\ell\text{, it holds that }\E_{\fatrhoH,\fat O\leftarrow{\cal D}}\left[{\cal W}^{(\ell)}_{\fatrhoH,{\cal A}^{\fat O}}\right]\fnleq & \ub,\text{ and }\\
\text{there exists }{\cal A}\in\PPT,\text{and polynomial }\ell\ \text{s.t.}\quad\E_{\fatrhoH,\fat O\leftarrow{\cal D}}\left[{\cal W}^{(\ell)}_{\fatrhoH,{\cal A}^{\fat O}}\right]\fngeq & \lb.
\end{align*}
\end{enumerate}
\end{rem}

\branchcolor{darkgray}{
For the proof, we use a modified version of the Hamiltonian lemma (\Cref{lem:trHrho=00003DtrHsigma}), that we state in the following remark. 
}

\begin{rem}[$\catgen{n,\ell}$ indistinguishability]
\end{rem}

\begin{itemize}
\item We say the family of states $\fat{\tau}_{1}$ and $\fat{\tau}_{2}$
are $\catgen{n,\ell}$ indistinguishable if for any algorithm $\A\in\cat{n,\ell}$,
such that $\desc(\Pi_{\secp},U_{\secp})\leftarrow\A(1^{\secp})$ ,
it holds that
\[
\left|\Pr\left[\left[\Pi_{\secp}U_{\secp}\tau_{1,\secp}U^{\dagger}_{\secp}\right]=1\right]-\Pr\left[\left[\Pi_{\secp}U_{\secp}\tau_{2,\secp}U^{\dagger}_{\secp}\right]=1\right]\right|\leq\negl
\]
for some negligible function $\negl$.
\item One can naturally extend the \Cref{lem:trHrho=00003DtrHsigma} to
$\catgen{n,\ell}$ indistinguishable states as follows: For every
family of extensive $k$-local extensive Hamiltonians, we have that
\[
\left|\tr\left[H_{\secp}U_{\lambda}\tau_{1,\secp}U^{\dagger}_{\lambda}\right]-\tr\left[H_{\secp}U_{\lambda}\tau_{2,\secp}U^{\dagger}_{\lambda}\right]\right|\leq\negl(\lambda)
\]
 where the unitary $U_{\secp}$ is produced by an algorithm $\A\in\cat{n,\ell}$,
i.e. $\desc\left(\Pi_{\secp},U_{\secp}\right)\leftarrow\A(1^{\secp})$.
\end{itemize}

With the notations in place, we define the distribution $\mathcal{D}$ over $\fatrhoH,\fat O$, for which we show the separation as follows.

Let $\fat O=(O_{f},O_{g})$ be the oracles corresponding
to the length preserving functions $f,g$ that are drawn uniformly
at random from the set of all length preserving functions. Corresponding
to each $\fat O$ we define the following families of states and Hamiltonians,
$\fat{\rho}=\{\rho_{\secp}\}_{\secp}$ and $\fat H=\{H_{\secp}\}_{\secp}$
as follows. 
\[
\rho_{\secp}:=\sum_{x\in\binset^{\secp}}\den{x,f(g(x))}\text{ and }
\]
\[
H_{\secp}:=\I_{2^{\secp}}\otimes\hw_{\secp}
\]
We denote the family of maximally mixed states $\fat{\sigma}=\{\sigma_{\secp}\}_{\secp}$,
where
\[
\sigma_{\secp}=\sum_{x,r\in\binset^{\secp}}\den{x,r}.
\]

\branchcolor{darkgray}{
The high level intuition the separation follows from \Cref{lem:sep-cat-vs-nocat-evaluation}. Since one can not distinguish between $\ket{x,f(g(x))}$ and $\ket{x,r)}$ in the absence of catalysts, one can not map the state $\rho_{\secp}$ to a low energy state of the Hamiltonian. Whereas, with sufficient catalysts, one can map $\ket{x,f(g(x))}\mapsto \ket{x,0}$, therefore minimising the energgy.
}

\begin{thm}[$\erghat_{{\cal D}}$ is strictly lesser than $\caterghat_{{\cal D}}$]
\label{thm:sep-cat-erg-hat-and-erg-hat}
In the random oracle model, 
for the family of states, Hamiltonians and the oracle, $\fatrhoH,\fat O$
sampled from ${\cal D}$ as defined above (see \Cref{def:explicit-fam-state-Ham-}),
it holds that 
\begin{align*}
\nexists\ \ \nonnegl\ \ \ \text{s.t. }\erghat_{{\cal D}} & \fngeq\nonnegl\text{, while }\\
\caterghat_{{\cal D}} & \fngeq\frac{n}{2}
\end{align*}
where $\nonnegl$ is any non-negligible function.
\end{thm}

\branchcolor{black}{\begin{proof}
We prove this statement by breaking it into two smaller claims.
\begin{claim*}
For any PPT algorithm $\A^{O_{\fat{\rho}},\hamdesc,\fat O}$ such
that, $\desc(U_{\secp})\leftarrow\A^{O_{\fat{\rho}},\hamdesc,\fat O}(1^{\secp})$,
where $U_{\secp}\in\uni(n)$, it holds that $\E_{{\cal D}}\left[\W_{\fatrhoH,\A}(\secp)\right]\leq\negl$.
\end{claim*}
\begin{proof}
Notice that a PPT algorithm that has access to $\fat O$ can efficiently
simulate $\A^{O_{\fat{\rho}},\hamdesc,\fat O}(1^{\secp}).$This is
because the given access to $O_{f},O_{g}$ the algorithm can efficiently
prepare the (classical) state $\rho_{\secp}$ upon receiving input
$1^{\secp}.$ Therefore, any unitary $U_{\secp}$ whose circuit description
is output by $\A^{O_{\fat{\rho}},\hamdesc,\fat O}(1^{\secp})$ belongs
to $\PPTcatU[2\secp][0]$.

From \Cref{lem:decision-cat-separation-robust} we know that, for
all $\fat U\in\PPTcatU[2\secp][0]$ and for any projector $\Pi^{\dec}$
that does not depend on $f,g$, there exists a negligible function
$\negl'$ such that
\[
\E_{{\cal D}}\left[\tr\left[\Pi^{\dec}U_{\secp}\rho_{\secp}U^{\dagger}_{\secp}\right]\right]\approx_{\negl'}\E_{{\cal D}}\left[\tr\left[\Pi^{\dec}U_{\secp}\sigma_{\secp}U^{\dagger}_{\secp}\right]\right].
\]
Since the Hamiltonian $H_{\secp}$ is extensive, and can be decomposed
into projectors that are independent of $f,g$, one can show (by arguments
similar to \Cref{lem:trHrho=00003DtrHsigma}) that
\begin{align*}
\E_{{\cal D}}\left[\tr\left[H_{\secp}\rho_{\secp}\right]-\tr\left[H_{\secp}U_{\secp}\rho_{\secp}U_{\secp}\right]\right]-\E_{{\cal D}}\left[\tr\left[H_{\secp}\sigma_{\secp}\right]-\tr\left[H_{\secp}U_{\secp}\sigma_{\secp}U_{\secp}\right]\right] & =\negl(\secp)\\
\implies\E_{{\cal D}}\left[\W_{\fatrhoH,\A^{\fat O}}(\secp)\right]-\E_{{\cal D}}\left[\W_{\fat{\sigma},\fat H,\A^{\fat O}}(\secp)\right] & \leq\negl(\secp) & \forall\A\in\PPT.
\end{align*}
Since $\fat{\sigma}$ is the family of maximally mixed states, on
average one can extract no work from the state $\sigma_{\secp}$ using
unitary operations. Therefore, as a consequence, we have that $\erghat_{{\cal D}}\fnleq\negl$.
\end{proof}

\begin{claim*}
$\caterghat_{{\cal D}}\fngeq\frac{n}{2}$.
\end{claim*}
\begin{proof}
From \Cref{lem:PinFCAT_3lambda} we know that there exists an algorithm
$\A^{\fat O}\in\catgen{n,\ell}^{\fat O}$ that produces a unitary
$U_{\secp}$ such that 
\[
U_{\secp}\left(\rho_{\secp}\otimes\den{0^{\ell(\secp)}}\right)U^{\dagger}_{\secp}=\underbrace{\sum_{x\in\binset^{\secp}}\den{x,0^{n(\secp)}}}_{\rho'_{\secp}}\otimes\den{0^{\ell(\secp)}}
\]
For such a unitary $U_{\secp}$ it holds that
\[
\E\left[\tr\left[H_{\secp}\rho_{\secp}\right]-\tr\left[H_{\secp}\rho'_{\secp}\right]\right]=\frac{n(\secp)}{2}.
\]
Therefore, there exists an oracle PPT algorithm $\A^{\fat O}$ such
that $\W_{\fatrhoH,\A^{\fat O}}\fngeq\frac{n}{2}$.
\end{proof}

\end{proof}
}

\subsection{Catalytic work potential}
\branchcolor{darkgray}{
From the above theorem, one can notice that there exists distributions over families of states and Hamiltonians for which one can extract more work through catalytic processes than without catalysts. We call such distributions as having \emph{catalytic work potential} as formally define it as follows.
}

\begin{defn}[Catalytic work potential]
 Let $n:\N\to\N$ be a function denoting the size of the system.
We say a distribution ${\cal D}$ over families of states $\fat{\rho}$,
Hamiltonians $\fat H$ (and oracles $\fat O$) have computational
\emph{catalytic-work potential} with gap $\Delta$ if either of the
following hold:

\begin{itemize}
\item $\erghat_{{\cal D}}\fnleq\ub$ but $\caterghat_{{\cal D}}\fngeq\lb$,
for some non-negligible function $\Delta:=\lb-\ub$.
\item $\erghat_{{\cal D}}$ is negligible (see \Cref{nota:erghat-is-negl})
but $\caterghat_{{\cal D}}\fngeq\Delta$, for some non-negligible
function $\Delta$. 
\end{itemize}
\end{defn}

We say the \emph{catalytic work potential is near-maximal} if $\Delta$ is at least the
energy difference between the average energy and the ground energy,
i.e. $\Delta\fngeq\bar{E}-E^{\min}$ where 
\begin{itemize}
\item $E=(E_{\secp})_{\secp}$ with $\bar{E}_{\secp}=\tr\left[H_{\secp}\right]/2^{n(\secp)}$ while 
\item $E^{\min}=(E^{\min}_{\secp})$ where $E^{\min}_{\secp}$ denotes the smallest eigenvalue of the Hamiltonian $(H_{\secp})$
\end{itemize}

One can extend this notion and say that a \emph{family of Hamiltonians has catalytic work potential}, if there exists a distribution over family of states from which one can extract more work (with respect to the Hamiltonian) by using catalysts.

\begin{thm}
\label{thm:cat-work-potential}
Let $n:\N\to\N$ be the polynomial $n(\lambda)=2\lambda$. Let $\fat H=\{H_{\lambda}\}_{\lambda}$
be a family of efficient extensive Hamiltonians acting on $n(\lambda)$
qubits.
\begin{itemize}
\item Denote its average energy by $\bar{E}_{\lambda}:=\tr\left[H_{\lambda}\right]/2^{n(\lambda)}$.
\item There is polynomial $\ell:\N\to\N$ and an algorithm ${\cal A}\in\cat{n,\ell}$
such that 
\[
\sigma_{\lambda}:=\E_{\circdesc(\Pi_{\lambda},U_{\lambda})\leftarrow{\cal A}(1^{\lambda})}U\left(\mathbb{I}_{2^\lambda}\otimes\den 0^{\lambda}\otimes\den 0^{\ell(\lambda)}\right)U^{\dagger}
\]
and its corresponding energy is $E_{\lambda}=\tr[H_{\lambda}\sigma_{\lambda}]=\nonnegl(\lambda)$.
\end{itemize}
Then $\left\{ H_{n}\right\} _{n}$ has catalytic work-potential of $\Delta=E-\bar{E}$, relative to a random oracle.
\end{thm}
\branchcolor{black}{
\begin{proof} (Proof sketch) Consider the following distribution over family of states $\fat{\rho}$ and oracles $\fat{O}$. Let $\fat{O}=(O_f ,O_g)$ be the random oracle, and the family corresponding to it be  $\fat{\rho}=\{\rho_\secp \}_\secp$ where $\rho_\secp = \sum_{x\in\binset^{\secp}}\den{x,f(g(x))}$. From \Cref{lem:decision-cat-separation-robust}, no computationally bounded algorithm can distinguish the state from maximally mixed state. Therefore it holds that the computational ergotropy for the distribution is negligible.

In the presence of catalysts, one can map the state $\rho_\secp = \sum_{x\in\binset^{n(\secp)}}\den{x,f(g(x))}$ to $\rho'_\secp = \sum_{x\in\binset^{\secp}}\den{x,0^{\ell(\secp)}}=\I\otimes\den{0^{\secp}}\otimes \den{0^{\ell(\secp)}}$

Since one can map the state $\rho'_\secp \mapsto \sigma_\secp$ efficiently (by premise), it holds that the family of  Hamiltonians $\fat{H}={H_\secp}_\secp$ has catalytic work potential of $\Delta=E-\bar{E}$. 
\end{proof}
}

\branchcolor{darkgray}{
Ask a consequence of the above theorem, we have the follwing. 
}

\begin{cor}
The Hamming weight Hamiltonian has catalytic work-potential with near-maximal
gap.
\end{cor}

\section{Pseudoergotropy}\label{sec:Pseudoergotropy}

\branchcolor{darkgray}{Ergotropy is a resource, in some ways analogous to a battery. One
can think of the battery as a quantum state $\rho$ which follows
dynamics corresponding to the Hamiltonian $H$ and one can extract
work from the battery by evolving through unitary evolutions. To ensure
that only the intended user is able to use the battery and extract
work from it, we define the notion of catalytic pseudoergotropy (used interchangably with pseudoergotropy). Informally,
we say that a distribution over families of state-Hamiltonian pair is has pseudoergotropy
if the states are efficiently constructible and have low ergotropy,
but are indistinguishable from a distribution over families state-Hamiltonian pairs with high ergotropy.
This notion is analogous to pseudorandom numbers. To an adversary
that does not know the secret key, the the high ergotropy state is
indistinguishable from the low ergotropy state, therefore it can extract
no work from the high ergotropic state. This is the case where the
adversary has stolen the battery, but the adversary can still not
extract any work out of it.

Symmetrically, one can also think about it this way, for computational
purposes, the low ergotropy state can mimic the high ergotropy state
and can be used as a resource. This is similar to how we think about
Pserudorandomness, a state with low randomness (the pseudorandom state)
which mimics a truly random state, to the entity that does not know
the secret key. 

In this section, we formalise the notion of pseudoergotropy and give explicit distributions that have catalytic pseudoergotropy. The construction we give in this section is also a counterexample to the Watanabe and Takagi's conjecture.

To begin with, we first set up some notations and definitions.}

A keyed distribution ${\cal D}$ is a probability distribution over
families of states, Hamiltonians and keys $(\fat{\rho},\fat H,\fat{\k})$.

Let $n,\kappa:\N\to\N$ be polynomials. Consider the keyed distributions
${\cal D}\text{ and }{\cal D}'$ over $(\fat{\rho},\fat H,\fat{\k})$
and $(\fat{\sigma},\fat H,\fat{\k})$ respectively, where the families
of states, Hamiltonians are defined on $n$-qubits and keys of are
of size $\kappa$.
\begin{rem}[Catalytic ergotropy for keyed distributions]
We extend $\caterghat$ for keyed distributions as follows: For a
particular family of states, Hamiltonians and oracle, $(\fatrhoH,\fat{\k})\leftarrow{\cal D}$,
we define the work extracted by an algorithm $\A$ as
\begin{align*}
\W^{(\ell)}_{\fat{\rho},\fat H,\A^{\fat{\k}}}(\secp) & :=\mathbb{\E}\Biggr[\tr\left[H_{\secp}\otimes\I_{2^{\ell(\secp)}}\left(\rho_{\secp}\otimes\den{0^{\ell(\secp)}}-U_{\secp}\left(\rho_{\secp}\otimes\den{0^{\ell(\secp)}}\right)U^{\dagger}_{\secp}\right)\right],\\
 & \qquad\qquad\qquad\qquad\qquad\qquad\qquad\qquad\qquad:\begin{array}{r}
\circdesc(U_{\secp})\leftarrow\A^{O_{\fat{\rho}},\hamdesc,\fat{\k}}(1^{\secp}),\\
U_{\secp}\in\catU[n(\secp)][\ell(\secp)][\left|0^{\ell(\lambda)}\right\rangle ]
\end{array}\Biggr]
\end{align*}
i.e. the algorithm knows the key $\fat{\k}$ corresponding to the
family of states $\fat{\rho}$ sampled from ${\cal D}$. The definition
of $\caterghat$ is then defined using $\W^{(\ell)}_{\fatrhoH,{\cal A}^{\fat k}}$
above, by proceeding as in \Cref{rem:caterg-distributions}.
\end{rem}

\branchcolor{darkgray}{
With the notations set up, we formally define Catalytic pseudoergotropy (used interchangably with pseudoergotropy) as follows.
}

\begin{defn}[Catalytic pseudoergotropy]
Let $c,d:\N\to\N$ be functions. We say the distribution ${\cal D}$
has catalytic pseudoergotropy $(c,d)$ with respect to ${\cal D}'$
if the following conditions hold.
\begin{itemize}
\item Ergotropy conditions:
\begin{align*}
\caterghat_{{\cal D}} & \fnleq c\\
\caterghat_{{\cal D}'} & \fngeq d.
\end{align*}
\item Indistinguishability condition. For all QPT algorithms ${\cal A}$,
there is a negligible function $\negl$ such that
\[
\left|\Pr_{(\fatrhoH,\fat k)\leftarrow{\cal D}}\left[1\leftarrow\A^{O_{\fat{\rho}}}(1^{\secp})\right]-\Pr_{(\fat{\sigma},\fat H,\fat k)\leftarrow{\cal D}'}\left[1\leftarrow\A^{O_{\fat{\sigma}}}(1^{\secp})\right]\right|\leq\negl(\secp).
\]
\item Efficiency condition: There exists an efficient sampling procedure
for the distribution ${\cal D}$.
\end{itemize}
\end{defn}

\subsection{Construction for Catalytic Pseudoergotropy}

\branchcolor{darkgray}{
In this section, we explicitly describe the distributions $\mathcal{D},\mathcal{D'}$ for which one can attain $(0,\nonnegl)$-catalytic pseudoergotropy. We start by setting up the following notations.
}

Consider a PRP family $(\feval,\finv)$ and let $n,m:\N\to\N$ be
s.t. $m>n$. We consider $\key_{f},\key_{g},\key_{h}\in\binset^{m}$,
and define the following functions $f=\{f_{\secp}\}_{\secp},g=\{g_{\secp}\}_{\secp},h=\{h_{\secp}\}_{\secp}$
where
\begin{itemize}
\item $f_{\secp}:\binset^{m(\secp)}\to\binset^{m(\secp)}$ as $f(x)=\feval\left(\key_{f},x\right)$, 
\item $g_{\secp}:\binset^{n(\secp)}\to\binset^{m(\secp)}$ as $g(x)=\feval\left(\key_{g},x\|0^{m-n}\right)$,
and 
\item $h_{\secp}:\binset^{m(\secp)}\to\binset^{n(\secp)}$ as $h(x)=\feval\left(\key_{h},x\right)_{|n}$\footnote{One note in the above lemma is that, the adversary $\mathcal{A}$
is only given access to the forward direction oracles, $f\left(.\right)$
or $g\left(h\left(.\right)\right)$ and not the inverses as well.
This observation will be crucial for us in catalytic ergotropy.}

(i.e. the $n$-bit prefix of the outcome).
\end{itemize}
\begin{defn}
\label{def:explicit-fam-state-Ham-PRF-1}Let $n,m:\N\to\N$ be polynomials
denoting the size of the quantum system and size of the key respectively.
We define a keyed distribution ${\cal D}'$ over families of states,
Hamiltonians and keys $(\fat{\rho},\fat H,\fat{\k})$ as follows.
\begin{itemize}
\item The family of keys $\fat{\k}=\left(\fat{\k}_{f},\fat{\k}_{g},\fat{\k}_{h}\right)$
are sampled uniformly from $\binset^{m}$
\item The family of states $\fat{\rho}=\left\{ \rho_{\secp}\right\} _{\secp}$
acting on $2n(\lambda)$-qubits where 
\[
\rho_{\lambda}:=\frac{1}{2^{m(\secp)}}\sum_{x\in\binset^{m(\secp)}}\den{g_{\secp}(h_{\secp}(x))}.
\]
\item The family of Hamiltonians $\fat H=\left\{ H_{\secp}\right\} _{\secp}$
(independent of $\fat{\k}$) where $H_{\lambda}=\I_{2^{n(\secp)}}\otimes H_{\hw(m(\secp)-n(\secp))}$
is the padded Hamming weight Hamiltonian (see \Cref{eq:n_ham_n}).
\end{itemize}
The keyed distribution ${\cal D}$ is defined exactly as above, except,
the family of states $\fat{\sigma}=\{\sigma_{\secp}\}_{\secp}$ is
defined as
\[
\sigma_{\secp}:=\frac{1}{2^{m(\secp)}}\sum_{x\in\binset^{m(\secp)}}\den{f_{\secp}(x)}=\frac{\I_{2^{m(\secp)}}}{2^{m(\secp)}}.
\]
\end{defn}

\branchcolor{darkgray}{
With the distributions set up we now show the main result of the section. 
}

\begin{thm}[${\cal D}$ is $\left(0,\frac{n}{4}\right)$-catalytic
pseudoergotropy w.r.t. ${\cal D}'$]
\label{thm:cons1-thomas-inspired-nonlocalH}Let the function $m=2n$. Assuming the existence of quantum secure PRPs, the keyed distribution
${\cal D}$ has $\left(0,\frac{n}{4}\right)$-catalytic
pseudoergotropy w.r.t. the family ${\cal D}'$.
\end{thm}

\branchcolor{darkgray}{

Notice that the only difference between the distributions $\mathcal{D}$ and $\mathcal{D'}$ is in the families of states used. The family in distribution $\mathcal{D}$ is maximally mixed, whereas in $\mathcal{D'}$ it is less mixed. In the presence of sufficient catalysts, one can convert the less mixed state to the ground of the Hamiltonian, by inverting the keyed permutation.

We prove \cref{thm:cons1-thomas-inspired-nonlocalH} formally through the following lemmas.
}

\begin{lem}
\label{lem:families-are-indistinguishable}For all polynomial functions
$t:\N\to\N$, it holds that 
\[
\left|\Pr_{(\fatrhoH,\fat k)\leftarrow{\cal D}}\left[1\leftarrow\A^{O_{\fat{\rho}}}(1^{\secp})\right]-\Pr_{(\fat{\sigma},\fat H,\fat k)\leftarrow{\cal D}'}\left[1\leftarrow\A^{O_{\fat{\sigma}}}(1^{\secp})\right]\right|\leq\negl(\secp).
\]
\end{lem}

\branchcolor{black}{\begin{proof}
This immediately follows from \Cref{lem:fx_equals_ghx}.
\end{proof}

}
\begin{lem}
\label{lem:phi-is-low-ergotropy}It holds that $\caterghat_{{\cal D}}\fnleq0$
for every Hamiltonian
\end{lem}

\branchcolor{black}{\begin{proof}
Since the family of states in ${\cal D}$ is the family of maximally
mixed states, it is invariant under unitary operations. Therefore,
$\caterghat_{{\cal D}}\fnleq0$.
\end{proof}

}
\begin{lem}
\label{lem:psi-is-efficiently-preparable}The distribution ${\cal D}$
is efficiently sampleable.
\end{lem}

\branchcolor{black}{\begin{proof}
Sample a family of keys $\fat{\k}=\left(\fat{\k}_{f},\fat{\k}_{g},\fat{\k}_{h}\right)$
uniformly at random. The family of states is just the family of maximally
mixed states.
\end{proof}

}
\begin{lem}
\label{lem:psi-has-high-caterghat}$\caterghat_{{\cal D}'}\fngeq\frac{\left(m-n\right)\cdot n}{2m}$.
%
\end{lem}

\branchcolor{black}{\begin{proof}
We compute $\E_{(\fatrhoH,\fat{\k})\leftarrow{\cal D}}\left[\tr\left[H_{\secp}\rho_{\secp}\right]\right]$
as follows
\begin{itemize}
\item $g_{\secp}(h_{\secp}(x))$ has at least $m(\secp)-n(\secp)$ zeroes.
This is because $g_{\secp}(\cdot)$ appends $0^{m(\secp)-n(\secp)}$
before applying the permutation.
\item In expectation over $x\in\binset^{m(\secp)}$ and the key $\k$, the
remaining qubits are indistinguishable from a uniformly distribution
over $0$ and $1$. Therefore, there are $\frac{n(\secp)}{2}$ qubits
that are 1.
\item Of these, only an $\frac{m(\secp)-n(\secp)}{m(\secp)}$ fraction aligns
with the Hamming weight part of the Hamiltonian $H_{\secp}.$ Therefore,
we have
\[
\E_{(\fatrhoH,\fat{\k})\leftarrow{\cal D}}\tr\left[H_{\secp}\rho_{\secp}\right]=\frac{m(\secp)-n(\secp)}{2}\times\frac{n(\secp)}{m(\secp)}.
\]
\end{itemize}
Consider the oracle algorithm $\A$, such that $U_{\secp}\leftarrow\A^{O_{\fat{\rho}},\hamdesc,\fat{\k}}(1^{\secp})$,
i.e. on input $1^{\secp}$ it outputs the description of a catalytic
unitary, that does the following:
\begin{itemize}
\item The catalytic unitary $U_{\secp}$ applies the inverse permutation
$\finv\left(\key_{g,\secp},\ \cdot\ \right)$ on the state $\rho_{\secp}$.
The state thus obtained is
\begin{align*}
\rho^{\prime}_{\secp} & =\frac{1}{2^{m(\secp)}}\sum_{x\in\{0,1\}^{m(\secp)}}\ket{\finv\left(\key_{g,\secp},g_{\secp}\left(h_{\secp}\left(x\right)\right)\right)}\bra{\finv\left(\key_{g,\secp},g_{\secp}\left(h_{\secp}\left(x\right)\right)\right)}_{\IO}\otimes\den{0^{m(\secp)}}_{\catalyst}\\
 & =\frac{1}{2^{m(\secp)}}\sum_{x\in\{0,1\}^{m(\secp)}}\ket{h_{\secp}\left(x\right)\|0^{m(\secp)-n(\secp)}}\bra{h_{\secp}\left(x\right)\|0^{m(\secp)-n(\secp)}}_{\IO}\otimes\den{0^{\otimes m(\secp)}}_{\catalyst}.
\end{align*}
\item Therefore, it follows that 
\[
\E_{(\fatrhoH,\fat{\k})\leftarrow{\cal D}}\tr\left[H_{\secp}\rho'_{\secp}\right]=0.
\]
\end{itemize}
One can compute the catalytic work extracted by the algorithm $\A$
as, 
\begin{align*}
\E_{{\cal D}}\left[\W^{(\ell)}_{\fatrhoH,\A^{\fat{\k}}}(\secp)\right] & =\E_{D}\left[\tr\left[H_{\secp}\rho_{\secp}\right]-\tr\left[H_{\secp}\rho'_{\secp}\right]\right]\\
 & \frac{m(\secp)-n(\secp)}{2}\times\frac{n(\secp)}{m(\secp)}.
\end{align*}
Setting $m=2n$ we get that $\caterghat_{{\cal D}}\fngeq\left(0,\frac{n}{4}\right)$.
\end{proof}

}

\begin{rem}
\label{rem:pseudoergotropy-free-energy}
The same state ensembles also exhibit a separation in nonequilibrium free energy. Fix an inverse temperature $\beta>0$ independent of $\lambda$, and equip the system with the Hamiltonian $\widetilde H_\lambda=0$. Its Gibbs state is $\tau_\lambda=\sigma_\lambda$. Define $F_{\beta,\lambda}(\omega)=\operatorname{tr}[\widetilde H_\lambda\omega]-\beta^{-1}S(\omega)$, where the entropy uses natural logarithms.

Since $h_\lambda$ is balanced and $g_\lambda$ is injective, $\rho_\lambda$ is uniform on $2^{n(\lambda)}$ computational-basis states. Hence $S(\rho_\lambda)=n(\lambda)\ln 2$ and $S(\sigma_\lambda)=m(\lambda)\ln 2$. For every key, it follows that $F_{\beta,\lambda}(\rho_\lambda)-F_{\beta,\lambda}(\tau_\lambda)=\beta^{-1}(m(\lambda)-n(\lambda))\ln 2$, whereas the reference state has zero free-energy excess. Taking $m(\lambda)=2n(\lambda)$ makes this gap linear in the number of qubits. Computational indistinguishability given polynomially many copies is unchanged. 

Thus, assuming quantum-secure pseudorandom permutations exist, these ensembles provide a counterexample to Watanabe and Takagi's conjecture that pseudo-nonequilibrium states do not exist~\cite{Watanabe2026}.
\end{rem}

\section{Classical vs Quantum Computational Ergotropy}
\branchcolor{darkgray}{
In this section we explore the notion of work extraction using classical operations instead of unitaries. Clearly allowing the set of operations to reversible classical computation (permutations), restricts the set of operations one can perform to extract work.  
In this section, assuming BPP$\subsetneq$BQP, we show that, there exists families of classical states and Hamiltonians from which one can extract non-negligible work catalytically using efficient unitary operations, while one can not do the same using permutations.

We proceed by setting up the notation for classical work extraction as follows.
}

\begin{defn}[Catalytic permutations $\catPerm(n,\ell,\varepsilon)$]
\label{def:approx-cat-perm}Let the $(n+\ell)$-qubit composite
Hilbert space be denoted by ${\cal H}:=\H_{\IO}\otimes\H_{\catalyst}.$
Then, the set of \emph{approximate catalytic permutations} $\catPerm(n,\ell,\varepsilon)$
is defined as 
\[
\catPerm(n,\ell,\varepsilon):=\left\{ \pi\in{\cal \Perm}(n+\ell):\begin{array}{r}
\exists\left|\eta\right\rangle \in\Cl\left({\cal H}_{\catalyst}\right)\ \forall\ \left|\phi\right\rangle \in\Cl\left(\H_{\IO}\right),\quad\exists\ \left|\psi\right\rangle \otimes\ket{\xi}\in\Cl\left(\H\right)\\
\text{s.t. }\TD{\pi\left(\phi\otimes\eta\right)\pi^{-1}}{\psi\otimes\xi}\le\varepsilon
\end{array}\right\} 
\]
where $\Cl(\H')$ represents all the computational basis states in
the Hilbert space $\H'$.
\end{defn}

\branchcolor{darkgray}{
Since the notion of using reversible classical computation is too restrictive, we study classical work extraction in the catalytic setting with approximate restoration.
More formally, we define the work extracted by a classical algorithm as follows.
}

\begin{defn}[Classically extracted work $(\WC_{\fat{\rho},\fat H,\A})$]
Let $\fatrhoH$ be families of classical states and Hamiltonians.
For an algorithm $\A$ we define the \emph{work extracted by an $\varepsilon$-approximate
$\ell$-catalytic permutation} (i.e. reversible classical circuit)
as follows.
\[
\WC^{(\ell,\varepsilon)}_{\fat{\rho},\fat H,\A}(\secp):=\E\left[\tr\left[H_{\secp}\left(\rho_{\secp}\otimes\den{0^{\ell}}-\pi_{\secp}\left(\rho_{\secp}\otimes\den{0^{\ell}}\right)\pi^{-1}_{\secp}\right)\right]:\begin{array}{c}
\circdesc(\pi_{\secp})\leftarrow\A^{O_{\fat{\rho}},\hamdesc}(1^{\secp}),\\
\pi_{\secp}\in\catPerm(n(\secp),\ell(\secp),\varepsilon)
\end{array}\right].
\]
Similarly for an algorithm $\A$ one can define the \emph{work extracted
by an $\varepsilon$-approximate $\ell$-catalytic unitary} as follows.
\[
\W^{(\ell,\varepsilon)}_{\fat{\rho},\fat H,\A}(\secp):=\E\left[\tr\left[H_{\secp}\left(\rho_{\secp}\otimes\den{0^{\ell}}-U_{\secp}\left(\rho_{\secp}\otimes\den{0^{\ell}}\right)U^{\dagger}_{\secp}\right)\right]:\begin{array}{c}
\circdesc(U_{\secp})\leftarrow\A^{O_{\fat{\rho}},\hamdesc}(1^{\secp}),\\
U_{\secp}\in\catU[n][\ell,\varepsilon]
\end{array}\right].
\]
\end{defn}
\branchcolor{darkgray}{
Naturally, one can extend the notion of work extracted by an algorithm to define the notion of classical catalytic ergotropy. For classical catalytic ergotropy, we allow at most negligible imperfections when returning the catalyst.
}
\begin{defn}[$\caterghatC(\ell)$]
We say $\caterghatC(\ell)$ is $(\lb,\ub)$-bounded if 
\begin{align*}
\forall\quad{\cal A}\in\PPT\ \ \exists\negl\ \ \text{s.t.},\E_{\fatrhoH\leftarrow{\cal D}}\left[{\cal \WC}^{(\ell,\negl)}_{\fatrhoH,{\cal A}}\right]\fnleq & \ub,\text{ and }\\
\exists\quad\A\in\PPT,\negl'\ \text{s.t.}\quad\E_{\fatrhoH\leftarrow{\cal D}}\left[{\cal \WC}^{(\ell,\negl')}_{\fatrhoH,{\cal A}}\right]\fngeq & \lb
\end{align*}
where $\ub,\lb$ are functions and $\negl,\negl'$ are negligible
functions.

\branchcolor{darkgray}{Analogously, we define quantum catalytic ergotropy with approximate
restoration of catalysts as follows.}
\end{defn}

\begin{defn}[$\caterghatQ(\ell)$]
We say $\caterghatQ(\ell)$ is $(\lb,\ub)$-bounded if 
\begin{align*}
\forall\quad{\cal A}\in\PPT\ \ \exists\negl\ \ \text{s.t.},\E_{\fatrhoH\leftarrow{\cal D}}\left[{\cal \W}^{(\ell,\negl)}_{\fatrhoH,{\cal A}}\right]\fnleq & \ub,\text{ and }\\
\exists\quad\A\in\PPT,\negl'\ \text{s.t.}\quad\E_{\fatrhoH\leftarrow{\cal D}}\left[\W^{(\ell,\negl')}_{\fatrhoH,{\cal A}}\right]\fngeq & \lb
\end{align*}
where $\ub,\lb$ are functions and $\negl,\negl'$ are negligible
functions.
\end{defn}

\branchcolor{darkgray}{
In subsequent discussions, we set up some notation and lemmas, to prove a separation between the classical catalytic ergotropy, and quantum catalytic ergotropy in \Cref{subsec:classical-quantum-sep}.
}

\subsection{Pseudo-deterministic to catalytic circuits}

\branchcolor{darkgray}{
A pseudo-deterministic circuit is one that for each input $x$ produces a fixed answer $y$ with $1-\negl$ probability. Here we show that any pseudo-deterministic circuit can be turned into a catalytic circuit with negligible restoration error on the catalyst.
}

Consider a relational problem $\rel$ that has only one solution for
each input, i.e. for each $x$, there is a unique $y$ such that $(x,y)\in\rel$.
This is essentially function evaluation. 

For this exposition, we restrict ourselves to such relations. We say a QPT algorithm ${\cal A}$ pseudo-deterministically
solves $\rel$ if there is a negligible function $\negl$ such that
$\Pr[y\leftarrow{\cal A}(x)]=1-\negl(|x|)$ for all inputs $x$. We
assume that ${\cal A}$ defers all measurements to the end of the
computation.

For the claim below, let $m,n,s:\N\to\N$ be functions denoting the size of the input, output and the workspace respectively.
\begin{claim}[Pseudo-Deterministic to Catalytic]
\label{claim:pseudo-det-to-catalytic}Let $\rel,n,m,s$ be as above.
Suppose ${\cal A}$ is a pseudo-deterministic QPT algorithm that,
on an $n(\lambda)$-bit input $x$, uses an $s(\lambda)$-qubit workspace
and produces an $m(\lambda)$-bit output $y$. Then, one can construct
an algorithm ${\cal B}$ that uses an $n(\lambda)+m(\lambda)$ qubit
$\IO$ register and $\ell(\lambda)$ qubit $\catalyst$ register where
$\ell:=n+m+s$. Furthermore, there is a negligible function $\negl'$
such that on input $x$, ${\cal B}$ applies a unitary $U_{\lambda}$
satisfying
\[
\left\Vert U_{\lambda}\left|x,0\right\rangle _{\IO}\left|0\right\rangle _{\catalyst}-\left|x,y\right\rangle _{\IO}\left|0\right\rangle _{\catalyst}\right\Vert \le\negl'(\lambda)
\]
for $\lambda>\lambda'$ for some $\lambda'\in\N$. 
\end{claim}

\branchcolor{black}{\begin{proof}
The proof is elementary. Suppose the action of ${\cal A}$ on input
$x$ is given by the unitary $V_{\lambda}$ that (from the fact that
${\cal A}$ is a pseudo-deterministic QPT algorithm implies) acts
as 
\[
V_{\lambda}\left|x\right\rangle \left|0^{m+s}\right\rangle =\sqrt{1-a^{2}_{xy}}\left|x,y\right\rangle \otimes\left|\varphi_{xy}\right\rangle +a_{xy}\left|\Psi_{\bad}\right\rangle 
\]
 followed by a measurement, where $0\le a_{xy}\le1$ is at most a
negligible function $\negl$, and $\left|\varphi_{xy}\right\rangle ,\left|\Psi\right\rangle $
are normalised vectors. 

We define a unitary $U_{\lambda}$ that acts on registers $\mathsf{I}$
($n$-qubits), $\mathsf{O}$ ($m$-qubits), $\mathsf{C}_{1}.\mathsf{I}$
($n$-qubits), $\mathsf{C}_{1}\cdot\mathsf{O}$ ($m$-qubit) and $\mathsf{C}_{2}$
($s$-qubit) as follows (see also \Cref{fig:pseudo-det-to-catalytic})
\[
U_{\lambda}=\CNOT_{\mathsf{I},\mathsf{C_{1}.\mathsf{I}}}(V_{\lambda})_{\mathsf{C}_{1}\mathsf{C}_{2}}\CNOT_{\mathsf{C}_{1},\IO}(V^{\dagger}_{\lambda}){}_{\mathsf{C_{1}C_{2}}}\SWAP_{\IO,\mathsf{C}_{1}}
\]
where $\CNOT_{\mathsf{A}\mathsf{B}}$ is a controlled-NOT gate, with
$\mathsf{A}$ as control and $\mathsf{B}$ as the target. 

By direct calculation (and using the fact that $a_{xy}\le\negl$),
one can verify that there is a negligible function $\negl'$ such
that 
\[
\left\Vert U_{\lambda}\left|x,0\right\rangle _{\IO}\left|0,0\right\rangle _{\mathsf{C}_{1}\mathsf{C}_{2}}-\left|x,y\right\rangle _{\mathsf{IO}}\otimes\left|0,0\right\rangle _{\mathsf{C}_{1}\mathsf{C}_{2}}\right\Vert \le\negl'(\lambda).
\]

\begin{figure}
\begin{centering}
\includegraphics[width=10cm]{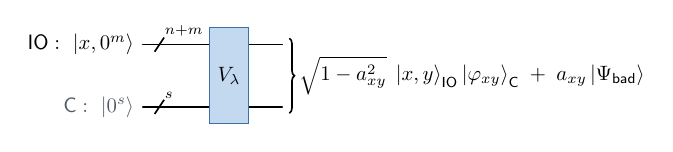}\\
\par\end{centering}
\begin{centering}
\includegraphics[width=13cm]{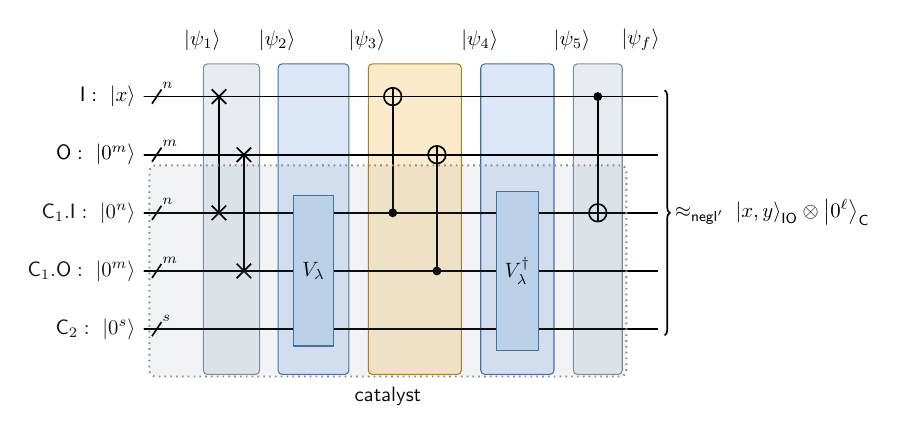}
\par\end{centering}
\caption{Pseudo-deterministic to catalytic}
\label{fig:pseudo-det-to-catalytic}

\end{figure}
\end{proof}

}

\subsection{Classical circuit to classical Hamiltonian}\label{subsec:Classical-circuit-to}

\branchcolor{darkgray}{
The Cook-Levin theorem gives a conditions on the transcript of a Turing machine to ensure that a Turing machine arrived at the through valid intermediate computation. In this part, we adapt the same method for circuits, and show the procedure to construct an extensive Hamiltonian corresponding to a poly-sized classical circuit such that:
\begin{itemize}
\item given an input where the circuit accepts, one can efficiently construct
the (classical) ground state of the Hamiltonian, and the converse,
\item given a ground state of the Hamiltonian, one can efficiently construct
an input where the circuit accepts. 
\end{itemize}

We set up the notation and describe the circuit-Hamiltonian conversion procedure below.
}

\paragraph{Circuit.}

Consider a polynomially-sized classical circuit $C$ specified by
$w$ wires and by gates $G_{1},\dots G_{g}$ where the gate $G_{\ell}$
acts on $W_{\ell}$ wires (where $W_{\ell}$ is a set of wires). For
simplicity, assume that each gate act on $2$ input wires and produce
$2$ output wires. Denote the output of the circuit (decision bit)
by the state of the first wire (after all the gates have been applied). 

Denote by $t_{0},t_{1},\dots t_{\ell},\dots t_{g}$ the state of the
wires (transcript) starting with the input $t_{0}$, where $t_{\ell}$
is the transcript after gate $G_{\ell}$ is applied. 

\paragraph{Hamiltonian.}

Given a circuit as described above, we define a Hamiltonian $H$ acting
on $gw$ wires/bits (where $w$ is the number wires and $g$ is the
number of gates in $C$, as described above), where the bits are labelled
as $(i,j)\in\{1,2\dots g\}\times\{1,2\dots w\}$ as follows
\begin{itemize}
\item Let $\pred_{\ell}$ be the predicate that applies $G_{\ell}$ on wires
$\{(\ell-1,i):i\in W_{\ell}\}$ and outputs $0$ (no penalty) if the
result matches the wires $\{(\ell,i):i\in W_{\ell}\}$ (i.e. it applies
$G_{\ell}$ on $t_{\ell-1}$ on the $\ell-1$th block of wires and
checks if it matches $t_{\ell}$ on the $\ell$th block of wires;
see ) but $1$ (penalty) otherwise.
\item Let $\pred_{\dec}$ be the predicate that outputs the \emph{opposite}
of the decision bit (i.e. state of wire $(g,1)$). 
\end{itemize}
Define the $4$-local Hamiltonian $H$ as
\[
H:=\sum^{g}_{\ell=1}\pred_{\ell}+\pred_{\dec}.
\]
 
\begin{claim}
\label{claim:CircToHam}Let $C$ and $H$ be as above. 
\begin{itemize}
\item Given an input $t_{0}$ on which $C$ accepts, one can efficiently
construct the ground state $t_{0},t_{1},\dots t_{g}$ of $H$
\item Given a ground state $t_{0},t_{1},\dots t_{g}$ of $H$, the string
$t_{0}$ is an accepting input of $C$. 
\end{itemize}
\end{claim}

\branchcolor{black}{\begin{proof}[Proof Sketch (elementary).]
The Hamiltonian $H$ acts on a sequence of `transcripts' $t_{0},t_{1},\dots t_{g}$
and applies an energy penalty for every inconsistency. It also applies
a penalty if the decision bit is $0$. 

Clearly, if $t_{0},t_{1},\dots t_{g}$ is in fact the transcript produced
by running $C$ such that the decision bit in $t_{g}$ is $1$, then
clearly $t_{0},t_{1}\dots t_{g}$ is the ground state of $H$.

Conversely, if the sequence $t_{0},t_{1},\dots t_{g}$ is a ground
state of $H$, it means that on input $t_{0}$ the circuit $C$ produces
$t_{g}$ as the final transcript and the decision bit (i.e. the first
bit of $t_{g}$) is $1$. Thus, $t_{0}$ is an accepting input of
$C$. 
\end{proof}

}

\subsection{Putting everything together---the separation}\label{subsec:classical-quantum-sep}

\branchcolor{darkgray}{
With all the parts in place, we now establish the final separation.

}

\paragraph{(Amplified, pseudo-deterministic) Public-Coin Proof of Quantumness.}

A \emph{public-coin proof of quantumness} is specified by a pair of
algorithms $({\cal P},{\cal V})$ where
\begin{itemize}
\item The prover ${\cal P}$ is a QPT algorithm that takes as input $c$
(challenge) and produces an output $w$ (witness)
\item The verifier ${\cal V}$ is a PPT algorithm that 
\begin{itemize}
\item on input $1^{\lambda}$ samples $c$ from some distribution over the
set of challenges $C$ according to distribution ${\cal D}_{C}(1^{\lambda})$
and sends it to ${\cal P}$,
\item on input $w$ from ${\cal P}$, computes a deterministic function
to compute $b(w,c)\in\{\mathsf{accept},\mathsf{reject}\}$ and outputs
$b$. 
\end{itemize}
\end{itemize}
A public-coin proof of quantumness protocol must satisfy the following conditions.
\begin{itemize}
\item (completeness) there is a QPT prover ${\cal P}$ that can make the
verifier accept with probability at least $1-\negl$ and
\item (soundness) no PPT prover ${\cal P}^{*}$ can make the verifier accept
with probability more than $\negl$.
\item We also require an additional `pseudo-deterministic' property: for
each $c$ there is a unique $w$ such that $b(w,c)=\mathsf{accept}$. 
\end{itemize}
\branchcolor{darkgray}{We will work in the common reference string model. This is like the
random oracle model but now, every party gets access to a string,
sampled from some fixed distribution---as opposed to getting access
to an oracle for a function sampled from the uniform distribution
over all functions. }

\begin{thm}
\label{thm:classical-vs-quantum-erg}
Suppose $({\cal P},{\cal V})$ specifies a public-coin proof of quantumness
as above. Suppose every party has access to the common reference string
$\{c_{\lambda}\}_{\lambda}$ sampled from the distribution ${\cal D}_{C}$.
Then, for each $\{c_{\lambda}\}_{\lambda}$, there is a family of
4-local extensive classical Hamiltonians $\fat H:=\{H_{\lambda}\}_{\lambda}$
(normalisation is as in \Cref{def:H-efficient}) and corresponding
family of states $\fat{\rho}:=\{\rho_{\lambda}\}_{\lambda}$ such 
that
\begin{align}
\nexists\quad\nonnegl\quad\text{s.t .}\caterghatC(\ell) & \fngeq\nonnegl\label{eq:soundness-cq-sep}\\
\exists\quad\negl'\quad\text{s.t. }\caterghatQ(\ell) & \fngeq1-\negl'\label{eq:completeness-cq-sep}
\end{align}
where ${\cal D}$ is the distribution over $\fatrhoH$, specified
implicitly by ${\cal D}_{c}$, $\ell$ is some polynomial and $\nonnegl$
is a non-negligible function while $\negl'$ is some negligible function.
\end{thm}

\branchcolor{black}{\begin{proof}[Proof Sketch.]
 Since the verifier is a PPT algorithm, one can construct a reversible classical circuit corresponding to it.
Let $\fat H=\{H_\secp\}_\secp$, where the \emph{reversible} circuit for the verifier is converted
 into $H_{\lambda}$ (using \Cref{claim:CircToHam}). We define $\fat{\rho}=\{\rho_\secp\}_\secp$, where  $\rho_{\lambda}$ is to be the transcript $t_{0},t_{1},\dots t_{g}$
where $t_{0}$ is the all zero string, and this is propagated correctly
through the verifier's circuit (see \Cref{subsec:Classical-circuit-to}).

Notice that 
\begin{itemize}
\item $\E_{{\cal D}_{C}}\tr(H_{\lambda}\rho_{\lambda})=1-\negl(\lambda)$
because no penalties until the very last step where the decision bit
is false---and so 1 unit penalty (with $1-\negl$ probability from
the guarantee of the PoQ)
\item $\rho_{\lambda}$ can be converted to an all zero string catalytically
(since we used a reversible circuit for the verifier, thus each element
in the transcript can be reversed)
\end{itemize}

One can show that
$\exists\quad\negl'\quad\text{s.t. }\caterghatQ(\ell)  \fngeq1-\negl'$ \Cref{eq:completeness-cq-sep}), through the following arguments.
The PoQ has almost perfect completeness and has the extra pseudo-deterministic
property; and these imply there is a pseudo-deterministic quantum
algorithm $V_{\lambda}$ that on input $c$ produces a $w$ such that
$b(w,c)=\mathsf{accept}$ with $1-\negl$ probability.

Using \Cref{claim:pseudo-det-to-catalytic} one can run $V_{\lambda}$
as catalytic unitary $U_{\lambda}'$ on an all zero string and produce
the response $w$.
Concretely, the PPT algorithm ${\cal A}^{c,\hamdesc,O_{\fat{\rho}}}$
produces the description of $U_{\lambda}$ that 
\begin{itemize}
\item first undoes the verifier's reversible computations to obtain an all
zero state starting from $t_{0},t_{1},\dots t_{g}$, and 
\item then applies the catalytic unitary $U'_{\lambda}$ that has the input
$c$ hard-coded.
\end{itemize}

One can show that $\nexists\quad\nonnegl\quad\text{s.t .}\caterghatC(\ell)  \fngeq\nonnegl$ 
(i.e. \Cref{eq:soundness-cq-sep}) through the following arguments.
Since $\E_{{\cal D}_{C}}\tr(H_{\lambda}\rho_{\lambda})$ is already
the `first excited' state, the only way to lower the energy is to
move to the ground state.
Using \Cref{claim:CircToHam} it follows that if the PPT algorithm
is able to produce the ground state, this would mean that from the
ground state, one can also efficiently extract $w$ (an accepting
input to the verifier's circuit). 
However, the PoQ has almost perfect soundness and therefore no PPT
algorithm on input $c$, can produce $w$ such that $b(w,c)=\mathsf{accept}$
with more than $\negl'$ probability.
\end{proof}

}
\branchcolor{darkgray}{
As a consequence of the above separation, we make the following remark.
}

\begin{rem}[Catalysts are useful computationally]
\indent (i) information-theoretically, quantum ergotropy and classical ergotropy for this family is identical. \\
\indent (ii) computationally, quantum ergotropy and classical ergotropy for this family is identical. \\
\indent (iii) Classical catalytic ergotropy is negligible.\\
\indent (iv) Quantum catalytic ergotropy is constant. \\
\end{rem}

\branchcolor{darkgray}{
This is another feature of thermodynamics that seems to only appear in the computational theory—and it is again aided by catalysts.
}

\newpage{}

\appendix

\section{Deferred Proofs}\label{sec:defElementary}
\begin{lem}[Restatement of \Cref{lem:fx_equals_ghx}; \cite{arnon2023CET}]
\label{lem:fx_equals_ghx-1}Consider a PRP family $(\feval,\finv)$
and let $n,m\in\mathbb{N}$ be s.t. $m>n$. We consider $\key_{f},\key_{g},\key_{h}\in\binset^{m}$,
and define the following functions: $f:\binset^{m}\to\binset^{m}$
as $f(x)=\feval\left(\key_{f},x\right)$, $g:\binset^{n}\to\binset^{m}$
as $g(x)=\feval\left(\key_{g},x\|0^{m-n}\right)$, and $h:\binset^{m}\to\binset^{n}$
as $h(x)=\feval\left(\key_{h},x\right)_{|n}$ (i.e. the $n$-bit prefix
of the outcome). Then for any polynomial function $m=m\left(n\right)$
and for any polynomial-time quantum algorithm $\A$ with binary output
it holds that 
\[
\left|\Pr_{\key_{g},\key_{h}}\left[\A^{g(h(\cdot))}\left(1^{n}\right)=1\right]-\Pr_{k_{f}}\left[\A^{f(\cdot)}\left(1^{n}\right)=1\right]\right|=\negl\left(n\right).
\]
\end{lem}

\branchcolor{black}{\begin{proof}
(Proof sketch) The proof proceeds through a sequence of hybrids. Consider
an adversary $\A$ and define the random variable $H_{0}=H_{0}(n)=\A^{g(h(\cdot))}\left(1^{n}\right)$,
where $g,h$ are defined based on uniformly random keys $\key_{g},\key_{h}$.

We then let $\pi_{g},\pi_{h}$ be uniformly random permutations on
$\binset^{m}$. In $H_{1}$ we change $h$ to be defined using $\pi_{h}$
instead of $\feval(\key_{h},\cdot)$. It holds that $\norm{H_{0}-H_{1}}_{1}=\negl\left(n\right)$
because of PRP pseudorandomness, where $\norm{\cdot}_{1}$ denotes
$L_{1}$ norm. Similarly, in $H_{2}$ we also change $g$ to be defined
using $\pi_{g}$ instead of $\feval(\key_{g},\cdot)$. Again we have
$\norm{H_{1}-H_{2}}=\negl\left(n\right)$ for the same reason.

The next few hybrids do not rely on computational assumptions but
rather on query lower bounds to random permutations and functions.
Let $Q$ be an upper bound on the number of queries made by $\A$
to its oracle. The next bounds only rely on the assumption that $Q=\poly\left(n\right)$.

In hybrid $H_{3}$ we change $\pi_{h}$ to be a random $\emph{function}$
rather than a permutation. It holds that $\norm{H_{2}-H_{3}}=\negl\left(n\right)$
using the query lower bound for distinguishing random functions from
random permutations. Hybrid $H_{4}$ does the same for $\pi_{g}$,
and by the same theorem $\norm{H_{3}-H_{4}}=\negl\left(n\right)$.

In $H_{5}$ we make a more radical change. Let $r:\binset^{m}\to\binset^{m}$
be a random function, we let $H_{5}=\A^{r(\cdot)}\left(1^{n}\right)$.
We show that $\norm{H_{4}-H_{5}}=\negl\left(n\right)$ using Zhandry's
small domain theorem (SDT). Fix the function $g(x)=\pi_{g}\left(x\|0^{m-n}\right)$.
This function defines a multi-set $S_{g}\subseteq\binset^{m}$ of
size $2^{n}$ (namely, it possibly has repetitions). Let $D_{1,g}$
be the uniform distribution over $S_{g}$. Let $D_{2}$ be the uniform
distribution over $\binset^{m}$ Using SDT, an algorithm that distinguishes
$H_{4}$ from $H_{5}$ with some advantage can be transformed into
an algorithm that distinguishes $D_{1,g}$ from $D_{2}$ by getting
polynomially samples from either distribution. Since the reduction
is uniform, this distinguisher has advantage in expectation over $g$
as well. However, when $g$ is chosen at random, polynomially many
samples from $D_{1,g}$ are statistically indistinguishable from polynomially
many samples from $D_{2}$, which completes the argument.

Next, we define $H_{6}=\A^{\pi_{f}(\cdot)}(1^{n})$, where $\pi_{f}$
is a random permutation. Here again we use the function to permutation
indistinguishability to argue that $\norm{H_{5}-H_{6}}=\negl(n)$.
Our last hybrid again relies on a cryptographic assumption. In $H_{7}$
we replace $\pi_{f}$ with $f$ and use PRP pseudorandomness to deduce
that $\norm{H_{6}-H_{7}}=\negl(n)$. We can conclude that $\norm{H_{0}-H_{7}}=\negl(n)$
and the lemma follows.
\end{proof}

}

\begin{lem}[Restatement of \cref{lem:expanding-RO-uniformity}]
\label{lem:expanding-RO-uniformity-1} Let $\fat O$ be the length
doubling random oracle, and $q_{1},q_{2}:\mathbb{N}\to\mathbb{N}$
be polynomials. For every $q_{1}$-query quantum algorithm $\mathcal{B}$
with oracle access to $\fat O$, there exists a $q_{2}$-query quantum
algorithm $\A$ with oracle access to $O_{\secp}$ corresponding to
the function $f_{\secp}$ such that,
\begin{align}
 & \left|\Pr\left[\mathcal{B}^{\fat O}(1^{\secp},\rho_{\secp})=1\right]-\Pr\left[\mathcal{B}^{\fat O}(1^{\secp},\sigma_{\secp})=1\right]\right|\nonumber \\
= & \left|\Pr\left[\cA^{O_{\secp}}(1^{\secp},\rho_{\secp})=1\right]-\Pr\left[\cA^{O_{\secp}}(1^{\secp},\sigma_{\secp})=1\right]\right|\label{eq:BO=00003DAO-1}\\
\leq & \negl(\secp).\label{eq:AO=00003Dnegl-1}
\end{align}
\end{lem}

\branchcolor{black}{\begin{proof}
We first prove that there exists an algorithm ${\cal A}$ such that
the output distributions of ${\cal B}^{\fat O}(1^{\lambda},\tau_{\lambda})$
and ${\cal A}^{O_{\lambda}}(1^{\lambda},\tau_{\lambda})$ are identical
for any $\tau_{\lambda}$ that depends only on $O_{\lambda}$. Observe
that each $f_{\lambda}$ is sampled uniformly and independently from
${\cal F}_{\lambda}$.  

Now, for every algorithm ${\cal B}^{\fat O}$, one can construct such
an algorithm ${\cal A}^{O_{\lambda}}$ as follows: ${\cal A}$ runs
${\cal B}$ identically except when ${\cal B}$ makes queries. When
${\cal B}$ makes a query to $\fat O$ at $\{0,1\}^{n(\lambda)}$,
then ${\cal A}$ uses $O_{\lambda}$, and when ${\cal B}$ makes any
other query, it simulates a random oracle independently. This establishes
\Cref{eq:BO=00003DAO-1}.

To establish \Cref{eq:AO=00003Dnegl-1}, we use a hybrid argument
and the One-Way to Hiding (O2H) Lemma. We define a sequence of hybrid
experiments $\mathsf{Exp}_{0},\mathsf{Exp}_{1},\mathsf{Exp}_{2}$
as follows
\begin{itemize}
\item $\mathsf{Exp}_{0}$: The challenger samples $O_{\secp}\leftarrow\mathcal{F_{\secp}}$
and $r\leftarrow\{0,1\}^{n(\secp)}$. It outputs $\mathcal{A}^{O_{\secp}}(O_{\secp}(r))$.
\item $\mathsf{Exp}_{1}$: The challenger samples $O_{\secp}\leftarrow\mathcal{F_{\secp}}$,
$r\leftarrow\{0,1\}^{n(\secp)}$ and $u\leftarrow\binset^{2n(\secp)}$.
It defines $O_{\secp}'$ such that $O_{\secp}'(r)=u$ and $O_{\secp}'(x)=O_{\secp}(x)$
for all $x\neq r$. It outputs $\mathcal{A}^{O_{\secp}'}(u)$.
\item $\mathsf{Exp}_{2}$: The challenger samples $O_{\secp}\leftarrow\mathcal{F_{\secp}}$,
$r\leftarrow\{0,1\}^{n(\secp)}$ and $u\leftarrow\binset^{2n(\secp)}$.
It outputs $\mathcal{A}^{O_{\secp}}(u)$.
\end{itemize}
\begin{claim}
$\Pr_{O_{\secp},r}\left[1\leftarrow\Exp_{0}\right]=\Pr_{O_{\secp},r,u}\left[1\leftarrow\Exp_{1}\right]$.\label{claim:exp0=00003D1-1}

\branchcolor{black}{\begin{proof}[Proof of \Cref{claim:exp0=00003D1}]
Evaluating a random oracle at any point yields a uniformly distributed
string. Since at every point $x\neq r,O_{\secp}(x)=O_{\secp}'(x)$,
the joint distribution $(O_{\secp},O_{\secp}(r))$ in $\Exp_{0}$
is statistically identical to the joint distribution $(O_{\secp}',u)$
in $\Exp_{1}.$ Therefore, 
\[
\Pr_{O_{\secp},r}\left[1\leftarrow\Exp_{0}\right]=\Pr_{O_{\secp},r,u}\left[1\leftarrow\Exp_{1}\right].
\]
\end{proof}

}

$\left|\Pr_{O_{\secp},r,u}\left[1\leftarrow\Exp_{1}\right]-\Pr_{O_{\secp},u}\left[1\leftarrow\Exp_{2}\right]\right|\leq\negl(n)$.\label{claim:exp1=00003D2-1}

\branchcolor{black}{\begin{proof}[Proof of \Cref{claim:exp1=00003D2}]
The inputs to $\cA$ is identical in both $\Exp_{1}$ and $\Exp_{2}$,
the only difference in the experiments is the oracle the $\cA$ has
access to differs in a single point-- $r$. Intuitively, for the
outputs of $\Exp_{1}$ and $\Exp_{2}$ to be different, the algorithm
$\cA$ should query the point $r.$ But since $\cA$ is QPT and can
query in superposition, we do the following.
\begin{proof}
We construct an algorithm $\mathcal{B}^{O_{\secp}}$ that on input
$u$, simulates $\cA^{O_{\secp}}(u)$ upto the $i$-th query, where
$i\leftarrow\{1,2,\cdots,q\}$, measures the query register and returns
the measurement output $T$. From \Cref{lem:O2H}, we know that such
an algorithm $\mathcal{B}^{O_{\secp}}$ will return $r$ with probability
$P_{\mathsf{guess}}$ such that, 
\[
\dfrac{\left|\Pr_{O_{\secp},r,u}\left[1\leftarrow\Exp_{1}\right]-\Pr_{O_{\secp},u}\left[1\leftarrow\Exp_{2}\right]\right|}{2q}\leq\sqrt{P_{\mathsf{guess}}}
\]

Notice that in $\Exp_{2}$, the variable $r$ is sampled uniformly
at random from $\binset^{n}$ and is independent of both the classical
input $u$ and the oracle $O_{\secp}$. Since $u,r$ are completely
uncorrelated, from the perspective of algorithm $\mathcal{B}^{O_{\secp}}$,
upon receiving input $u$, the probability with which it returns the
random string $r$ is $\frac{1}{2^{n}}=\negl(n).$ Taking a union
bound over $\poly(n)$ many trials, we get $P_{\mathsf{guess}}\leq\negl(n)$.
Therefore, 
\begin{align*}
\dfrac{\left|\Pr_{O_{\secp},r,u}\left[1\leftarrow\Exp_{1}\right]-\Pr_{O_{\secp},u}\left[1\leftarrow\Exp_{2}\right]\right|}{2q} & \leq\negl(n)
\end{align*}
Since $\cA$ is a QPT algorithm, $q=\poly(n),$ therefore we obtain
$\left|\Pr_{O_{\secp},r,u}\left[1\leftarrow\Exp_{1}\right]-\Pr_{O_{\secp},u}\left[1\leftarrow\Exp_{2}\right]\right|\leq\negl(n)$,
proving the claim.
\end{proof}

\end{proof}

}
\end{claim}

Now using the hybrid argument, we conclude $\left|\Pr_{O_{\secp},r}\left[1\leftarrow\Exp_{0}\right]-\Pr_{O_{\secp},u}\left[1\leftarrow\Exp_{2}\right]\right|\leq\negl(n)$
proving the theorem statement.
\end{proof}

}

\bibliographystyle{alpha}
\bibliography{cry}

\end{document}